%% file: paper.tex
\documentclass[a4paper,11pt]{article}

\input{macros}

\input{tikzstuff}

\usepackage{booktabs}
\usepackage{multirow}

\title{Finding path motifs in large temporal graphs\\using algebraic fingerprints%
\thanks{An earlier version of this work appeared in the SIAM International
Conference on Data Mining (SDM20) titled ``Pattern detection in large temporal
graphs using algebraic fingerprints''. A final version of this work will
appear in the Big Data journal special issue titled ``Best of SDM 2020''.}
} 
\author{Suhas Thejaswi\thanks{Department of Computer Science, Aalto University, Finland.}
\and Aristides Gionis\thanks{Department of Computer Science, KTH Royal Institute
of Technology, Sweden, and Department of Computer Science, Aalto University, Finland.}
\and Juho Lauri}

\date{}

\begin{document}
\maketitle


\begin{abstract}
We study a family of pattern-detection problems in \emph{vertex-colored} \emph{temporal} graphs. In particular, given a vertex-colored temporal graph and a multiset of colors as a query, we search for \emph{temporal paths} in the graph that contain the colors specified in the query. These types of problems have several applications, for example in recommending tours for tourists or detecting abnormal behavior in a network of financial transactions.

\smallskip
For the family of pattern-detection problems we consider, we establish complexity results and design an algebraic-algorithmic framework based on {\em constrained multilinear sieving}. We demonstrate that our solution scales to massive graphs with up to a billion edges for a multiset query with five colors and up to hundred million edges for a multiset query with ten colors, despite the problems being \np-hard.
Our implementation, which is publicly available, exhibits practical edge-linear scalability and is highly optimized. For instance, in a real-world graph dataset with more than six million edges and a multiset query with ten colors, we can extract an optimum solution in less than eight minutes on a \emph{Haswell} desktop with four cores.

\end{abstract}

\section{Introduction}
Pattern mining in graphs has become increasingly popular due to applications in
analyzing and understanding structural properties of data originating from
information networks, social networks, transportation networks, and many more.
Searching for patterns in graphs is a fundamental graph-mining task that has
applications in computational biology and analysis of metabolic
networks~\cite{lacroix2006motif}, discovery of controversial discussions in
social media~\cite{colettoKGL2017}, and understanding the connectivity of the
brain~\cite{honey2007network}, among others.  At the same time, real-world data
are inherently complex. To accurately represent the heterogeneous and dynamic
nature of real-world graphs, we need to enrich the basic graph model with
additional features. Thus, researchers have considered \emph{labeled
graphs}~\cite{yang2013community}, or \emph{heterogeneous
graphs}~\cite{liu2010mining}, where vertices and/or edges are associated with
additional information represented with labels, and \emph{temporal
graphs}~\cite{holme2012temporal}, where edges are associated with timestamps
that indicate when interactions between pairs of vertices took place.

\smallskip
In this paper we study a family of pattern-detection problems in graphs that are
both \emph{labeled} and \emph{temporal}.  In particular, we consider graphs in
which each vertex is associated with one (or more) labels, to which we refer as
\emph{colors}, and each edge is associated with a timestamp.  We then consider a
\emph{motif query}, which is a multiset of colors.  The problem we consider is
to decide whether there exists a \emph{temporal path} whose vertices contain
exactly the colors specified in the motif~query.  A temporal path in a temporal
graph refers to a path in which the timestamps of consecutive edges are strictly
increasing. If such a path exists, we also want to find it and return it as
output.

\smallskip
The family of problems we consider have several applications. One
application is in the domain of tour recommendations~\cite{DeFAGLY2010} for
travelers or tourists in a city.  In this case, vertices correspond to
locations.  The colors associated with each location represent different
activities that can be enjoyed in that particular location.  For example,
activity types may include items such as museums, archaeological sites, or
restaurants.  Edges correspond to transportation links between different
locations, and each transportation link is associated with a timestamp
indicating departure time and duration.  Furthermore, for each location we may
have information about the amount of time recommended to spent in that location,
e.g., minimum amount of time required to enjoy a meal or appreciate a museum.
Finally, the multiset of colors specified in the motif query represents the
multiset of activities that a user is interested in enjoying.  In the
tour-recommendation problem we would like to find a temporal path, from a
starting location to a destination, which satisfies temporal constraints (e.g.,
feasible transportation links, visit times, and total duration) as well as the
activity requirements of the~user, i.e., 
what kind of places they want to visit.

\smallskip
Another application is in the domain of analyzing networks of financial
transactions.  Here, the vertices represent financial entities, the vertex
colors represent features of the entities, and the temporal edges represent
financial transactions between entities, annotated with the time of the
transaction, amount, and possibly other features.  An analyst may be interested
in finding long chains of transactions among entities that have certain
characteristics, for example, searching for money laundering activities may
require querying for paths that involve public figures, companies with certain
types of contracts, and banks in offshore locations.

\smallskip
The use cases outlined above, as well as similar applications, 
can be abstracted and formulated as problems of finding paths 
in vertex-colored (vertex-labeled) and temporal graphs. 
More concretely, in this paper we consider the following problems:

\squishlist
\item \ktemppath:
decide if there exists a temporal path of length $k-1$;
\item \pathmotif:
decide if there exists a temporal path whose vertices contain the set of
colors specified by a motif query;
\item \colorfulpath:
decide if there exists a temporal path of length $k-1$ whose vertices have all the
colors precisely once;
\item \sdcolorfulpath:
decide if there exists a temporal path of length $k+1$, whose internal vertices have all
the colors precisely once,
and with a given source $s$ and destination $d$;
\item \rainbowpath: decide if there exists a temporal path of length $k-1$ having $k$ distinct
colors in a graph with $q > k$ colors;
\item \ectemppath:
decide if there exists a temporal path of length $k-1$ with specific edge timestamps;
\item \ecpathmotif:
decide if there exists a temporal path with specific edge timestamps and vertices containing the
set of colors specified by a motif query;
\item \vcpathmotif:
decide if there exists a temporal path whose vertices contain the set of colors specified by motif
query in the specified order; and
\item \vccolorfulpath:
decide if there exists a temporal path of length $k-1$, whose vertices have distinct
colors in the specified order.
\squishend

\smallskip
These problem variants can be useful in
different scenarios of our application domain, 
depending on the user constraints and/or requirements.
To motivate some of the different problem variants we presented above, 
consider again the tour-recommendation use-case in which
the vertices correspond to locations and
vertex labels correspond to location types, e.g.,
museum, restaurant, caf\'e, etc. The temporal edges between vertices correspond to
travel connections at specific time\-stamps. 
The tour-recommendation problem asks 
to find an itinerary by taking into
consideration the tourist's preferences
with respect to the locations which they want to visit.
The case that a tourist prefers not to visit more than one location of the same type 
can be modeled as a \colorfulpath problem. 
The case that, in addition to the previous constraint, 
a tourist knows their start and end location
(for example, starting at the hotel they stay and ending at a favorite restaurant)
can be modeled as an \sdcolorfulpath problem.
Finally, the case that a tourist wants to maximize  
the number of different types of locations they visit
can be modeled as a \rainbowpath path.

\smallskip
Most of the problems we consider are \np-hard;\footnote{\vccolorfulpath problem
is solvable in polynomial time (see \S\,\ref{sec:algorithm-ext:vccolorfulpath}).} thus, there is no known efficient
algorithm to find an exact solution. In such cases most algorithmic solutions
resort to heuristics or approximation schemes for the reason of scalability. In this paper we
present an (exact) \emph{algebraic approach} based on \emph{constrained
multilinear sieving} for pattern detection in temporal graphs and demonstrate
that our approach is scalable to large graphs.

\smallskip
The algorithms based on 
constrained multilinear detection offer the theoretically best-known results
for a set of fundamental combinatorial problems
including $k$-path~\cite{BjorklundHKK2017}, Hamiltonian path~\cite{Bjorklund2014}, 
and many variants of the graph motif problem~\cite{BjorklundKK2016}.
The implementations based on multilinear siev\-ing are known to saturate the
empirical arithmetic and memory bandwidth on modern CPU and GPU
micro-architectures. Furthermore, these implementations can scale to large
graphs as well as large query sizes \cite{BjorklundKKL2015, kaskiLT2018}.

\smallskip
Even though these algebraic techniques have been studied extensively in the
algorithms community, they have not been applied to data-mining problems to the
best of our knowledge. As such, this is the first work to do so and also to
apply these techniques for pattern detection in temporal graphs.

\smallskip
Our key contributions are as follows:
\squishlist
\item We introduce a set of pattern-detection problems that originate in the
\emph{vertex-colored} and \emph{temporal} graphs.
For the problems we consider we present \np-hardness results,
while showing that they are {\em fixed-parameter tractable}~\cite{CyganFKLMPPS2015},
meaning that if we restrict the size of the motif query the
problems are solvable in polynomial time in the size of the host graph.
\item We present a general algebraic-algorithmic framework based on
constrained multilinear sieving.
Our solution exhibits edge-linear scalability.
The algorithmic approach described in this work is not
limited to temporal paths, but rather it can be extended to study information
cascades, temporal arborescences and temporal subgraphs. An overview of our key
results is given in Table~\ref{table:introduction:1}.

\item We extend the vertex-localization variant of the constrained multilinear
sieving to solve path problems in temporal graphs. In this approach, we work
with a family of polynomials one for each vertex, rather than a single
polynomial, there by isolating the vertices which are part of a match. Most
importantly, vertex-localization comes with no additional cost with respect to
either space or time. This approach is effective for preprocessing the graph and
extracting a solution for many variants of the temporal-path problem.

\item We engineer a memory-efficient implementation of the algebraic algorithm
and demonstrate with extensive experiments that our implementation can scale to
graphs with up to a billion edges for multiset query with five colors and up to
one hundred million edges for multiset query with ten colors.
\item Open-source release: our implementations and datasets are released as open
source \cite{conf-code, journal-code}.
\squishend


\begin{table*}[t]
\caption{An overview of our key results. Here, $n$ is the number of vertices, $m$ is the
number of edges, $t$ is the maximum timestamp, $k-1$ is the length of path and $q$ is
the number of colors in the graph.}
\label{table:introduction:1}
\centering
\footnotesize
\begin{tabular}{l l l l}
\toprule
Problem & Hardness & Time complexity & Space complexity\\
\midrule
\ktemppath     & \np-complete (Lemma~\ref{lemma:temppath:1})   & $\bigO(2^k k (nt+m))$     & $\bigO(nt)$\\
\pathmotif     & \np-complete (Lemma~\ref{lemma:pathmotif-np})  & $\bigO(2^k k (nt+m))$     & $\bigO(nt)$\\
\colorfulpath   & \np-complete (Lemma~\ref{lemma:colorfulpath:1}) & $\bigO(2^k k (nt+m))$     & $\bigO(nt)$\\
\sdcolorfulpath & \np-complete (Lemma~\ref{lemma:sd-rainbow})    & $\bigO(2^k k (nt+m))$     & $\bigO(nt)$\\
\rainbowpath  & \np-complete (Lemma~\ref{lemma:colorful})      & $\bigO(q^k 2^k k (nt+m))$ & $\bigO(nt)$\\
\ectemppath    & \np-complete (Lemma~\ref{lemma:ectemppath:1})  & $\bigO(2^k (nk+m))$       & $\bigO(n)$\\
\ecpathmotif   & \np-complete (Lemma~\ref{lemma:ecpathmotif:1}) & $\bigO(2^k (nk+m))$       & $\bigO(n)$\\
\vcpathmotif   & \np-complete (Lemma~\ref{lemma:vcpathmotif:1}) & $\bigO(2^k k (nt+m))$     & $\bigO(nt)$\\
\vccolorfulpath & Polynomial                                     & $\bigO(mt)$               & $\bigO(nt)$\\
\bottomrule
\end{tabular}
\end{table*}

\section{Related work}

Pattern detection and pattern counting are fundamental problems in data mining.
In the context of paths and trees, pattern matching problems have been
extensively studied in non-temporal graphs both in theory
\cite{BjorklundHKK2017, BjorklundKK2016, cicaleseGGLLRT13, GagieHLW13,
GiaquintaG13, KowalikL2016}
as well as applications
\cite{BensonGL2016, ColettoGGL2017, Holmes2012, MiloSKC2002}.
For many restricted variants of path problems Kowalik and Lauri~\cite{KowalikL2016} 
presented complexity results and deterministic algorithms with runtime
bounds that are optimal under plausible complexity-theoretic assumptions. 
Most of these problems are known to be fixed-parameter tractable and 
the best known randomized algorithms for a subset
of path and subgraph pattern detection problems is due to {Bj{\"o}rklund~et~al.}~\cite{BjorklundHKK2017, BjorklundKK2016}. 
An algorithmic technique known as color coding can be used to
approximately count the patterns in $\bigO^*(2^k)$ time, however, these
algorithms require $\bigO^*(2^k)$ memory \cite{AlonPIFC2008}.\footnote{The 
notation $\bigO^*$ hides factors bounded polynomially in the input size.} A practical
implementation of color coding using adaptive sampling and succint encoding was
demonstrated by Bressan et al.~\cite{BressanLP2019} for a pattern counting problem.
However, the techniques based on color coding are mostly used to detect and count
patterns in graphs with no vertex labels.

\smallskip
Algebraic algorithms based on multilinear and constrained multilinear sieving
are due to the pioneering work of Koutis~\cite{koutis-icalp08, koutis-dagstuhl, koutis-ipl}, 
Williams~\cite{williams-ipl}, Koutis and
Williams~\cite{koutis-williams-icalp09, koutisW2016}. The approach has been
extended to various combinatorial problems using a multivariate variant of the
sieve by Bj{\"o}rklund et al.~\cite{BjorklundHKK2017}. 
The authors introduced decision oracles which were used by Dell et al.~\cite{DellLM2020} to approximately count motifs.
A practical implementation of multilinear sieving and its scalability
to large graphs has been demonstrated by Bj{\"o}rklund et
al.~\cite{BjorklundKKL2015}.  Furthermore, its parallelizability to
vector-parallel architectures and scalability to large multiset sizes was shown
by Kaski et al.~\cite{kaskiLT2018}.

\smallskip
In the recent years there has been a lot of progress with respect to mining
temporal graphs. The most relevant work includes methods for efficient
computation of network measures, such as centrality, connectivity, density,
and motifs~\cite{Dechter1991, Holme2015, holme2012temporal, Kostakos2009, latapy2018stream}, 
as well as mining frequent subgraphs in temporal networks
\cite{LiuBC2019, ParanjapeBL2017, wackersreuther2010frequent}.
Path problems in temporal graphs are well studied 
\cite{GeorgeKS2007, WuCHKLX2014}.
In fact, many variants of these path problems are known to be solvable in polynomial-time
\cite{WuCHKLX2014, WuC2016}. Perhaps surprisingly, a simple variant where one 
wants to check the existence of a temporal path with waiting time constraints 
was shown to be \np-complete by Casteigts et al.~\cite{CasteigtsHMZ2019}. 
A known variant of the temporal-path problem is finding top-$k$ shortest paths,
which not only asks us to find a shortest path, but also the next $k-1$ shortest
paths; which may be longer than the shortest path~\cite{GuptaAH2011}. Here by
shortest path we mean that the total elapsed time of the temporal path is minimized. 
Note that the top-$k$ shortest path is different from the \ktemppath problem studied
in our work.

\smallskip
With the availability of social media data in the recent years there has been
growing interest to study pattern mining problems in temporal graphs. Paranjape
et al.~\cite{ParanjapeBL2017} presented efficient algorithms for counting small
temporal patterns. Liu et al.~\cite{LiuBC2019} presented complexity results and
approximation methods for counting patterns in temporal graphs. However, they
mainly study temporal graphs with no vertex-labels (colors).
Kovanen et al.~\cite{KovanenKKKS2011} studied a general variant of the temporal
subgraph problem in temporal graphs with vertex labels.  Aslay et
al.~\cite{AslayMFGG2018} presented methodologies for counting frequent patterns
with vertex and edge labels in streaming graphs. However, most of these
approaches were limited to small pattern sizes of up to three vertices.

\smallskip
To the best of our knowledge, there is no existing work
related to detecting and extracting temporal patterns with vertex labels.
The problems considered in this paper are closely related to variants of classical
problems such as the orienteering problem, TSP and Hamiltonian path
\cite{VansteenwegenSO2011, garey2002computers}. A~motivating application for the
problems can be traced to the context of tour recommendations~\cite{DeFAGLY2010,
GionisLPT2014}.
%


\section{Method overview}

Our method relies on the \emph{algebraic-fingerprinting}
technique~\cite{koutis-williams-icalp09,williams-ipl}. As this technique is not
well known in the data-mining community, we provide a bird's eye view.
The approach is described in more detail in Section~\ref{sec:algorithm}.

\smallskip
In a nutshell, the problem is to decide the existence of a pattern, or
a structure in the data. The idea is to encode the pattern-discovery problem as a
polynomial over a set of variables. The variables represent entities of the
problem instance (e.g., vertices and/or edges), and their values represent
possible solutions (e.g., whether a vertex belongs to a path). The challenge is
to find a poly\-nomial encoding that has the property that a solution to our
problem exists if and only if the poly\-nomial evaluates to a non-zero term.
We can then verify the existence of a solution, using \emph{polynomial iden\-ti\-ty
testing}, in particular, by evaluating random sub\-sti\-tu\-tions of variables: if one
of them does not evaluate to zero, then the polynomial is not identically zero.
Thus, the method can give false negatives, but the error probability can be
brought arbitrarily close to zero.

\smallskip
It should also be noted that an explicit representation of the polynomial can be
exponentially large. However, we do not need to represent the polynomial
explicitly, since we only need to be able to evaluate the variable substitutions
fast.

\smallskip
This paper is organized as follows.
In the next section we will introduce the terminology. 
In Section~\ref{sec:problems} we introduce the path problems in temporal graphs and 
in Section~\ref{sec:algorithm} we present an algebraic algorithm to solve the
temporal-path problems. In Section~\ref{sec:algorithm-ext}, we extend our
algebraic framework to solve temporal path problems with additional constraints.
In Sections~\ref{sec:implementation}, \ref{sec:experimental-setup}, and
\ref{sec:experimental-eval} we discuss implementation, experimental setup and
experimental evaluation, respectively. 
Finally, we conclude in Section~\ref{sec:conclusion}.

\section{Terminology}
\label{sec:terminology}

In this section we introduce the basic terminology used in the paper.

\smallskip
A \emph{graph} $G$ is a tuple $(V,E)$ where $V$ is a set of \emph{vertices}
and $E$ is a set of unordered pairs of vertices called \emph{edges}. We denote
the number of vertices $|V| = n$ and the number of edges $|E|= m$.
Vertices $u$ and $v$ are {\em adjacent} if there exists an edge $(u,v) \in E$.
The set of vertices adjacent to vertex $u$ excluding $u$ itself is the 
\emph{neighborhood} of vertex $u$ and it is denoted by $N(u)$.
A \emph{walk} between any two vertices is an alternating sequence of vertices
and edges $u_1 e_1 u_2 \dots e_{k} u_{k+1}$ such that there exists an edge
$e_i=(u_i, u_{i+1}) \in E$ for each $i \in [k]$.\footnote{For
convenience, we write $[k] = \{1,2,\dots,k\}$ for a positive integer $k$.}
We call the vertices $u_1$ and $u_{k+1}$ the \emph{start} and \emph{end}
vertices of the walk, respectively. The vertices $v_2,\dots,v_k$ are called
internal vertices.
The \emph{length} of a walk is the number of edges in the walk.
A \emph{path} is a walk with no repetition of vertices.

\smallskip
A \emph{temporal graph} $G^{\tau}$ is a tuple $(V,E^{\tau})$, where $V$ is a set
of vertices and $E^{\tau}$ is a set of temporal edges. A \emph{temporal edge} is
a triple $(u,v,j)$ where $u, v \in V$ and $j \in Z_{+}$ is a
\emph{timestamp}. The \emph{maximum timestamp} in $G^\tau$ is denoted by $t$.
The total number of edges at time instance $j \in [t]$ is denoted by $m_j$ and
the total number of edges in a temporal graph is $m = \sum_{j \in [t]} m_j$.
A vertex $u$ is adjacent to vertex $v$ at timestamp $j$ if there exists an edge
$(u,v,j) \in E^\tau$. The set of vertices adjacent to vertex $u$ at time step
$j$ is denoted by $N_j(u)$. The set of vertices adjacent to vertex $u$ excluding
$u$ itself is
denoted by $N(u) = \bigcup_{j \in [t]} N_j(u)$.
A temporal graph can also be defined as $G^{\tau} =
\bigcup_{j \in [t]} G^j$, where $G^j=(V,E^j)$ is a \emph{snapshot} of the graph at
time instance $j \in [t]$, where $t$ is the maximum time instance. In our
discussions we mostly use the former definition of a temporal graph.

\smallskip
A \emph{temporal walk} $W^\tau$ between any two vertices in a temporal graph is an
alternating sequence of vertices and temporal edges $u_1 e_1 u_2 e_2 \dots e_k
u_{k+1}$ such that there exists an edge $e_{i}=(u_i,u_{i+1},j) \in E^{\tau}$
for all $i \in [k]$ and for any two edges $e_i=(u_i,u_{i+1},j)$,
$e_{i+1}=(u_{i+1},u_{i+2},j')$ in the walk $W^\tau$, it is $j<j'$.
In other words, the timestamps of
the edges should always be in strictly increasing order. We call the vertices
$v_1$ and $v_{k+1}$ the \emph{start} and \emph{end} vertices, respectively. The
vertices $v_2,\dots,v_k$ are called internal vertices.
The \emph{length} of a temporal walk is the number of edges in the temporal walk.
A \emph{temporal path} is a temporal walk with no repetition of vertices.

\smallskip
To distinguish a graph from a temporal graph sometimes we explicitly refer to 
a graph as \emph{non-temporal} graph or \emph{static} graph.

\section{Path problems in temporal graphs}
\label{sec:problems}
In this section we will introduce a set of path problems in temporal graphs. 
An exact algorithm based on {\em multilinear sieving} is presented in the next section. 

\smallskip
Before proceeding, we note that in all our hardness proofs it is 
straightforward to establish membership to \np.
Thus, to streamline the presentation, we explicitly omit showing this part
in all subsequent proofs.
Let us begin our discussion with the $k$-path problem for non-temporal graphs
before continuing to path problems in temporal graphs.

\subsection{$k$-path problem in non-temporal graphs (\kpath)}
Given a graph $G=(V,E)$ and an integer $k \leq n$ the \kpath problem asks to
decide whether there exists a path of length at least $k-1$ in $G$.

\smallskip
The \kpath problem is \np-com\-plete~\cite[ND29]{garey2002computers} with the
fastest known randomized fixed-parameter algorithm for the problem is due to
Bj{\"o}rklund et al.~\cite{BjorklundHKK2017} and has complexity
$\bigO^*(1.66^k)$. The fastest known deterministic algorithm for the
problem is due to Fomin et al.~\cite{Fomin2016efficient} and has complexity
$\bigO^*(2.62^k)$.

\subsection{$k$-path problem in temporal graphs (\ktemppath)}
Given a temporal graph $G^{\tau}=(V,E^{\tau})$ and an integer $k \leq n$
the \ktemppath problem asks to decide whether there exists a \emph{temporal path}
of length at least $k-1$ in $G^{\tau}$. An example of the \kpath problem is
illustrated in Figure~\ref{fig:temppath:1}.

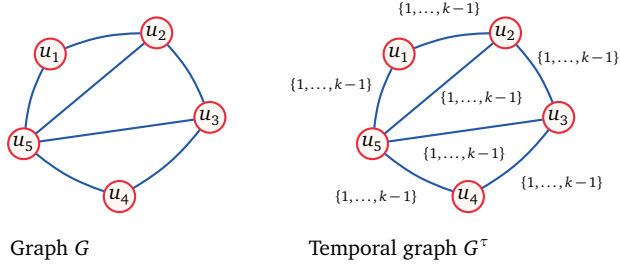
\begin{figure}[t]
\centering
\input{tikz-np-hardness-1}\hspace{7mm}
\input{tikz-np-hardness-2}
\caption{\label{fig:temppath:1}
An illustration of graph construction $G^\tau$.}
\end{figure}

\smallskip
For the \np-hardness, we reduce the \kpath problem to \ktemppath problem.

\begin{lemma}
\label{lemma:temppath:1}
Problem \ktemppath is \np-complete.
\end{lemma}
\begin{proof}
To prove our claim, we proceed by giving a polynomial-time reduction from the \kpath problem.

Given an instance $\langle G=(V,E), k \rangle$ of \kpath, 
we construct a temporal graph $G^{\tau}=(V^{\tau},E^{\tau})$
such that $V^{\tau} = V$ and $E^{\tau} = \bigcup_{i\in [k]} E^i$, where
$E^i = \bigcup_{(u,v)\in E} (u,v,i)$. 
The construction is illustrated in Figure~\ref{fig:temppath:1}.
We claim that there exists a \kpath in $G$ if and only if there exists a \ktemppath in $G^{\tau}$.

Assume that there exists a path $P=u_1 e_1 u_2 \dots e_{k-1} u_{k}$ of length
$k$ in $G$. By construction of $G^\tau$, we know that all edges in $E$ are
present in $E^\tau$ for every timestamp in $[k-1]$. Thus, we can construct a
temporal path $P^\tau$ such that for all vertices in $u_i \in P$ we keep it as
it is in $P^\tau$ and for each edge $e_i=(u_i, u_{i+1}) \in P$ we replace it by
$e_i^* = (u_i, u_{i+1}, i)$ in $P^\tau$ (by construction such an edge always
exist in $G^\tau$). So $P^\tau = u_1 e_1^* u_2 \dots e_{k-1}^* u_{k}$ is a
temporal path of length $k-1$ in $G^\tau$.

Conversely, assume that there exists a temporal path
$P^\tau=u_1 e_1^* u_2 \dots e_{k-1}^* u_{k}$ of length $k$ in $G^\tau$.
We construct a path $P=u_1 e_1 u_2 \dots e_{k-1} u_{k}$ by replacing
$e_i^* = (u_i, u_{i+1}, i)$ by $e_i=(u_i, u_{i+1})$
(such an edge always exist in $G$ by construction).
Evidently, $P$ is a path of length $k-1$ in $G$, completing the proof.
\end{proof}

\subsection{Path motif problem in temporal graphs (\pathmotif)}
Given a vertex-colored temporal graph $G^\tau=(V,E^\tau)$ and a multiset $M$ of
colors the \pathmotif problem asks to decide whether there exists a temporal path $P^\tau$ in
$G^\tau$ such that the vertex colors of $P^\tau$ agree with $M$. An example of
the \pathmotif problem is illustrated in Figure~\ref{fig:pathmotif:1}.

\begin{figure}[t]
\centering
\input{tikz-pathmotif-1}\hspace{20mm}
\input{tikz-pathmotif-2}
\caption{\label{fig:pathmotif:1}
An example of \pathmotif problem in temporal graphs.}
\end{figure}
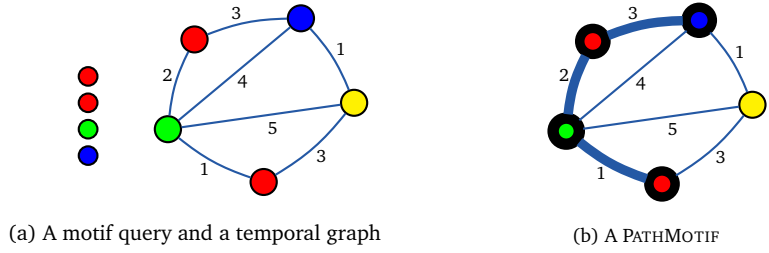

\smallskip
The \pathmotif problem is \np-complete and a reduction from \ktemppath to
\pathmotif is straightforward.

\begin{lemma}
\label{lemma:pathmotif-np}
Problem \pathmotif is \np-complete.
\end{lemma}
\begin{proof}
Given an instance of a \ktemppath with the input temporal graph
$G^{\tau}=(V,E^{\tau})$ we reduce it, in polynomial time, to \pathmotif as follows. 
We build a vertex-colored temporal graph such that all its vertices have the same color
and the multiset $M$ comprises the color $1$ exactly $k$ times.
More precisely, we let $G_c^{\tau}=(V_c,E^{\tau}_c)$ with the vertex
set $V_c = V$, the edge set $E_c^\tau = E^\tau$, the color mapping
$c:V\rightarrow \{1\}$ and the multiset $M=\{1^{k}\}$.
This finishes the construction.
We claim that there exists a \ktemppath in $G^{\tau}$ if and only if there
exists a \pathmotif in $G_c^{\tau}$. 
  
Let $P$ be a temporal path of length $k-1$ in $G^\tau$. We choose $P_c = P$ as a
\pathmotif of length $k$ in $G_c^\tau$ since all vertices have the same color
and the multiset $M$ agree with the colors of vertices in $P_c$.
In the other direction, it suffices to observe that any \pathmotif of 
length $k-1$ in $G_c^\tau$ is also a temporal path in $G$.
By Lemma~\ref{lemma:temppath:1} we have that \ktemppath is \np-hard,
so by our construction \pathmotif is \np-complete, which is what we wanted to show.
\end{proof}

\subsection{Colorful path problem in temporal graphs (\colorfulpath)}

Given a temporal graph $G^\tau=(V, E^\tau)$, an integer $k \leq n$, and a
coloring function $c:V \rightarrow [k]$, the \colorfulpath problem asks us to
decide whether there exists a temporal path $P^\tau$ of length $k-1$ such that
all vertex colors of $P^\tau$ are different. 

\smallskip
In \cite{conf-paper}, we introduced the colorful path problem as the rainbow
path problem, nevertheless to keep the problem definition concurrent with
non-temporal graphs we rename the problem as the colorful path problem.

\smallskip
The \colorfulpath problem is a special case of the \pathmotif problem, where all
the colors in the multiset $M$ are different, that is $M=[k]$. 
It is easy to see that the \colorfulpath problem in static graphs can be reduced to
the \colorfulpath problem in temporal graphs by replacing each static edge with
$k-1$ temporal edges.
So, the \colorfulpath problem is \np-complete. 
We skip the proof as the construction
is similar to that of Lemma~\ref{lemma:temppath:1}.

\smallskip
The \colorfulpath problem in static graphs is defined as follows: given a static
graph $G=(V,E)$ and a coloring function $c: V \rightarrow[k]$, the problem asks
us to find a path $P$ of length $k-1$ such that all vertex colors of $P$ are
different. 
The \colorfulpath problem in static graphs is known to be
\np-complete~\cite{Chen2011,Uchizawa2013}; however, it can be noted 
that these works consider a variant in which one requires the \emph{internal} 
vertices of a path not to repeat colors. Nevertheless, this variant is 
computationally equivalent to \colorfulpath (see e.g.,~\cite{KowalikL2016}).

\begin{figure}[t]
\centering
\input{tikz-colorfulpath-1}\hspace{20mm}
\input{tikz-colorfulpath-2}
\caption{\label{fig:colorfulpath:1}
An example of \colorfulpath problem in temporal graphs.}
\end{figure}
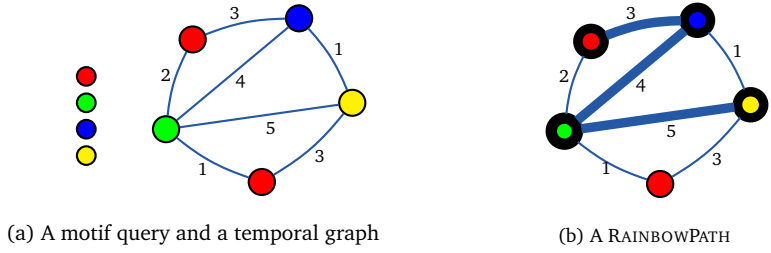

\begin{lemma}
\label{lemma:colorfulpath:1}
Problem \colorfulpath is \np-complete.
\end{lemma}

\subsection{Colorful $(s,d)$-connectivity in temporal graphs (\sdcolorfulpath)}

Given a temporal graph $G^\tau=(V, E^\tau)$, an integer $k \leq n$,
a source vertex $s \in V$, a destination vertex $d \in V$, and a coloring
function $c: V \setminus \{s,d\} \rightarrow [k]$, the \sdcolorfulpath
problem is to decide whether there exists a temporal path $P^\tau$ of length $k+1$
between vertices $s$ and $d$ such that the colors of internal vertices in
$P^\tau$ are different.

\smallskip
The \sdcolorfulpath problem is \np-hard. The proof, obtained by a reduction from
\colorfulpath problem is presented in Lemma~\ref{lemma:sd-rainbow}.

\begin{lemma}
\label{lemma:sd-rainbow}
Problem \sdcolorfulpath is \np-complete.
\end{lemma}

\begin{proof}
We proceed via a polynomial-time reduction from \colorfulpath.
Specifically, given an instance $\langle G^\tau, k, c \rangle$ of \colorfulpath, we
construct an instance $\langle G^{\tau'}, k', c' \rangle$ of \sdcolorfulpath, where
$G^{\tau'}=(V',E^{\tau'})$, $V'=V \cup \{s,d\}$, 
$E^{\tau'}=\bigcup_{(u,v,i) \in E^\tau} (u,v,i+1) \cup \bigcup_{u \in V} (s,u,1)
\bigcup_{u \in V} (u,d,t+2)$, $c'(v) = c(v)$ for all $v \in V$, $c(s) = k+1$,
$c(d) = k+2$, and $k' = k+2$. Here, $t$ is  the maximum timestamp in
$G^\tau$. 
We claim that there exists a \colorfulpath of length $k-1$ in
$G^\tau$ if and only there exists a \sdcolorfulpath of length $k+1$ between
vertices $s$ and $d$ in $G^{\tau'}$.

Let $P^\tau = v_1 e_1 v_2 \dots e_{k-1} v_{k}$ be a \colorfulpath in $G^\tau$.
We construct a path $P^{\tau'} = s e_0^* v_1 e_1^* \dots e_{k-1}^* v_{k} e_{k}^*d$ 
such that $e_0^*=(s,v_1,1)$, $e_{k+1}^*=(v_{k+1},d,t+2)$ and
$e_i^*=\{(u,v,j+1): e_i=(u,v,j)\}$ for all $i \in [k]$. Since the colors of $v_i
\in V(P^\tau)$ are all different, the vertices in path $P^{\tau'}$ will also
have different colors. So $P^{\tau'}$ is a colorful path of length $k+1$ between
vertices $s$ and $d$.

Let $P^{\tau'}=s e_0^* v_1 e_1^* \dots e_{k-1}^* v_k e_k^* d$ be a temporal path of
length $k+1$ in $G^{\tau'}$. We construct a path $P^\tau= v_1 e_1 \dots e_{k-1}
v_{k}$ such that $e_i=\{(u,v,j-1): e_i^*=(u,v,j)\}$ for all $i \in [k-1]$. Since
the colors of vertices in $P^{\tau'}$ are all different which means the colors
of vertices in $P^\tau$ will also be different. So $P^\tau$ is a \colorfulpath of
length $k-1$ in $G^\tau$.
\end{proof}

\subsection{Rainbow path problem in temporal graphs (\rainbowpath)}

Given a vertex-colored temporal graph $G^\tau=(V,E^\tau)$, an integer $k$ and a
coloring function $c: V \rightarrow [q]$ with $k < q \leq n$, the
\rainbowpath problem is to decide whether there exists a temporal path $P^\tau$
of length $k-1$ in $G^\tau$ such that the vertex colors of $P^\tau$ are
different. 

\smallskip
For the \np-hardness, we reduce the \kpath problem in non-temporal graphs to the
\rainbowpath problem in temporal graphs.
The proof is presented in Lemma~\ref{lemma:colorful}.

\begin{lemma}
\label{lemma:colorful}
Problem \rainbowpath is \np-complete.
\end{lemma}

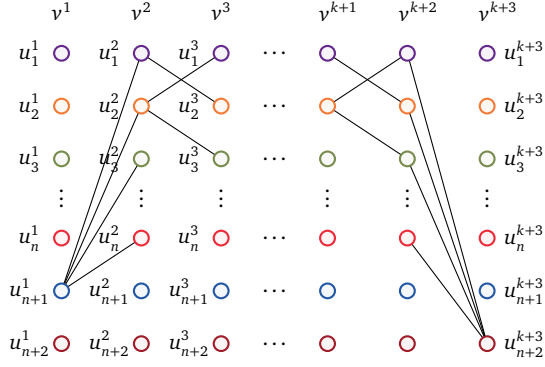
\begin{figure}
\begin{center}
\input{tikz-rainbowpath-hardness}\
\caption{\label{fig:rainbowpath:1}
Construction of $G^{\tau}$.}
\end{center}
\end{figure}

\begin{proof}

We reduce the \kpath problem in general graphs to the \rainbowpath problem in 
temporal graphs. Given an instance of a $k$-path problem on a graph $G=(V,E)$, 
we construct a vertex-colored temporal graph
$G^{\tau}=(V^{\tau}, E^{\tau})$
such that the vertex set
$V^{\tau}= \bigcup_{i=1}^{k+3}V^i$,
where $V^i = \bigcup_{j=1}^{n+2} u_j^i$ for all $i \in [k+3]$,
%
and the edge set $E^{\tau}=\bigcup_{i=1}^{k+2} E^i$ where
for all $i \in \{2,3,\dots,k+1\}$ and for each edge $(u_x, u_y) \in E$ we add edges
$(u_x^i, u_y^{i+1}, i), (u_x^{i+1}, u_y^i, i)$ to $E^{i}$.
Additionally, add edges
$(u_{n+1}^1, u_{j}^2, 1)$ and
$(u_{j}^{k+2}, u_{n+2}^{k+3}, k+2)$ for all $j \in [n]$.
The edges between $V^i$ and $V^{i+1}$ are the temporal edges in $E^{i}$ at time
instance $i$ for all $i \in [k+2]$. Finally, for all $i \in [n+2]$ $j \in [k+3]$
the vertex $u_i^j$ is assigned color $i$. The construction of graph $G^\tau$ is
illustrated in Figure~\ref{fig:rainbowpath:1}.

We claim that there exists a path of length $k$ in $G$ if and only if there
exists a \rainbowpath of length $k+2$ in $G^{\tau}$. So the \rainbowpath
problem in temporal graphs is at least as hard as the \kpath problem in graphs.

Let us assume that there exists a path
$P=u_1 e_1 u_2 \dots e_k u_{k+1}$
of length $k$ in $G$. We construct a temporal path
$P^\tau = u_1^* e_1^* u_2^* \dots e_{k+2}^* u_{k+3}^*$ in $G^\tau$ as follows.
For all $i \in \{2,\dots,k+2\}$ we have $u_{i-1} \in P$ so
$u_i^*=u_{i-1}^{i}$ and
$u_1^*=u_{n+1}^1$, $u_{k+3}^*=u_{n+2}^{k+3}$.
%
%
The edges in $P^\tau$ are constructed as follows:
$e_i^* = (u_{i-1}^{i},u_{i}^{i+1},i+1)$ for all $i \in \{2,\dots,k+2\}$
where the original edge in the path $P$ is $e_i=(u_{i-1}, u_{i})$ (such a
temporal edge always exists by construction of $G^\tau$). Finally,
$e_1^* = (u_{n+1}^1,u_{1}^2,1)$ where $u_1$ is the first vertex in $P$,
$e_{k+2}^* = (u_{k+1}^{k+2}, u_{n+2}^{k+3}, k+2)$ where $u_{k+1}$ is the end
vertex of path $P$.
Since $P$ is a path there exists no two vertices $u_i, u_j \in P$ such that
$u_i=u_j$, so no two vertices $u_i^*, u_j^* \in P^\tau$ have the same color.
Consequently, $P^\tau$ is a \rainbowpath of length $k+2$ in $G^\tau$.

Let us assume that
$P^\tau = u_1^1 e_1^1 u_2^2 \dots e_{k+2}^{k+2} u_{k+1}^{k+2}$ is a
\rainbowpath of length $k+2$ in $G^\tau$. We construct a path
$P=u_1^* e_1^* u_2^* \dots, e_k^* u_{k+1}^*$
such that for all $i \in \{2,\dots,k+2\}$ $u_{i-1}^* = u_{i}$ such that
$u_{i+1}{i+1} \in P^\tau$ and
for each edge
$e_{i}^{i} = (u_i^i, u_{i+1}^{i+1}, i) \in P^\tau$
such that $i \in \{2, \dots, k+2\}$ we replace
$e_{i-1}^* = (u_i, u_{i+1})$ in $P$.
Since $P^\tau$ is a rainbow path there exists no two vertices in $P^\tau$ that
have the same color, implying that there exist no two vertices $v_i, v_j \in P$
such that $v_i = v_j$. Thus, $P$ is a path of length $k$ in $G$ and we are done.
\end{proof}

\subsection{Temporal path problem with edge constraints (\ectemppath)}
Given a temporal graph $G^\tau=(V,E^\tau)$, an integer $k \leq n $, a tuple
$T=(j_1,\dots,j_{k-1})$ of timestamps such that 
$j_i < j_{i+1}$ for all $i \in [k-2]$, the problem asks us to 
decide the existence of a temporal path $P^\tau$ such that the timestamps in
$P^\tau$ agree with the timestamps in $T$ in the specified order.

\smallskip
To establish the \np-hardness we reduce the \kpath problem in non-temporal
graphs to the \ectemppath problem in temporal graphs. The proof of the reduction is
presented in Lemma~\ref{lemma:ectemppath:1}.

\begin{lemma}
\label{lemma:ectemppath:1}
Problem \ectemppath is \np-complete.
\end{lemma}
\begin{proof}
We proceed by a polynomial-time reduction from \kpath 
(in non-temporal graphs) to \temppath with edge constraints.

Given an instance $\langle G=(V,E), k \rangle$ of \kpath, 
we construct an instance of \ectemppath problem such that 
$G^\tau = (V, E^\tau)$, where $V\tau = V$ and $E^\tau = \bigcup_{(u,v) \in
E, i \in [k]} (u, v, i)$, and the set time constraints on edges $T=\{1,2,\dots,k\}$. 
We claim that there exists a solution for the \kpath problem in $G$ if and only if
there exists a solution for \ectemppath problem in $G^\tau$.

Let $P=\{v_1,e_1,v_2,\dots,e_k,v_{k+1}\}$ be a solution for the \kpath problem
in $G$. We construct a temporal path
$P^\tau=\{v_1,e_1^{'},v_2,\dots,e_k^{'},v_{k+1}\}$ such that
$e_i^{'}=(v_{i},v_{i+1},i)$ for each $i \in [k]$. By construction, we have that
such an edge always exists in $E^\tau$. The proof in other direction is analogous
to the previous.
The claim follows.
\end{proof}

\subsection{Path motif problem with edge constraints (\ecpathmotif)}
Given a vertex-colored temporal graph $G^\tau=(V,E^\tau)$, an integer $k \leq n$, 
a multiset $M$ of colors and a tuple $T=(j_1,\dots,j_{k-1})$ of timestamps
such that $j_i < j_{i+1}$ for all $i \in [k-2]$. The problem asks us to decide
the existence of a temporal path $P^\tau$ such that the timestamps of $P^\tau$
agree to that of timestamps in $T$ in the specified order and the vertex colors
of $P^\tau$ agree to that of colors in $M$. 

\smallskip
A reduction from \colorfulpath problem in non-temporal graphs to \ecpathmotif
problem in temporal graphs is straightforward. The
construction is similar to that of Lemma~\ref{lemma:temppath:1}, so we omit a
detailed proof.

\begin{lemma}
\label{lemma:ecpathmotif:1}
Problem \ecpathmotif is \np-complete.
\end{lemma}


\subsection{Path motif problem with vertex constraints (\vcpathmotif)}
Given a vertex-colored temporal graph $G^\tau=(V,E^\tau)$, an integer $k \leq n$, 
a tuple $M$ of colors. The problem asks us to find a temporal path
of length $k-1$ such that the vertex colors of the path match to that of colors in
$M$ in the specified order. 

\smallskip
A reduction from the \ktemppath problem to \vcpathmotif problem with multiset
$M=\{1^k\}$ is straightforward. The hardness follows.

\begin{lemma}
\label{lemma:vcpathmotif:1}
Problem \vcpathmotif is \np-complete.
\end{lemma}


\subsection{Colorful path problem with vertex constraints (\vccolorfulpath)}
Given a temporal graph $G^\tau=(V,E^\tau)$, an integer $k \leq n$, a coloring
function $c:V \rightarrow [k]$ and a tuple $M$ of $k$ different colors. The
problem asks us to decide the existence of a temporal path $P^\tau$ of length
$k-1$ such that the vertex colors of $P^\tau$ agree to that of $M$ in the
order specified. 

\smallskip
The \vccolorfulpath problem is solvable in polynomial time. We present a
dynamic programming algorithm in Section~\ref{sec:algorithm-ext:vccolorfulpath}.

\section{Algebraic algorithm for temporal paths}
\label{sec:algorithm}

We now present an algorithm for the \ktemppath and \pathmotif problems.
Our algorithm relies on a \emph{polynomial encoding} of temporal walks and the
\emph{algebraic fingerprinting}
technique~\cite{BjorklundKK2016,koutis-icalp08,koutis-williams-icalp09,williams-ipl}.
The algorithm is presented in three steps:
\squishlist
\item[($i$)] a dynamic programming recursion to generate a polynomial
encoding of temporal walks; 
\item[($ii$)] an algebraic algorithm to detect the existence of a
multilinear monomial in the polynomial generated using the recursion in ($i$) ---
furthermore, we prove that the existence of a multilinear monomial implies the existence
of a temporal path; and
\item[($iii$)] finally, an extension of the approach to detect temporal paths with
additional color constraints via constrained multilinear detection.
\squishend

\smallskip
Let us begin our discussion with the concept of polynomial encoding of walks and
continue to polynomial encoding of temporal walks.

\smallskip
Let $P$ be a multivariate polynomial such that every monomial $M$ is of the form
\[
x_1^{d_1} x_2^{d_2} \dots x_q^{d_q} y_1^{f_1} y_2^{f_2} \dots y_r^{f_r}.
\]
A monomial is \emph{multilinear} if $d_i \in \{0,1\}$ for all $i \in [q]$, 
and $f_j \in \{0,1\}$ for all $j \in [r]$. A monomial is \emph{$x$-multilinear} if
$d_i \in \{0,1\}$ for all $i \in [q]$ --- in other words, we do not take into account the
degrees of the $y$-variables. The \emph{degree} of a monomial $M$ is the sum of the
degrees of all its variables. 
More restrictively, the \emph{$x$-degree} of $M$ is the sum of the degrees of all its $x$-variables.

\subsection{Monomial encoding of a walk}
Let $W = v_1 e_1 v_2 \dots e_{k-1} v_k$ be walk in a temporal
graph $G = (V, E)$ for any integer $k > 1$. 
Let $\{x_{v_1},\dots,x_{v_n}\}$ be a set of
variables corresponding to the vertices in $V = \{ v_1, \ldots, v_n \}$ and 
let $\{y_{uv,\ell}: (u,v) \in E, \ell \in [k]\}$ be a set of variables such that $y_{uv,\ell}$ 
corresponds to an edge $(u,v) \in E$ 
that appears at position $\ell$ in~$W$.
A monomial encoding of $W$ is represented as
\[
x_{v_1} y_{v_1v_2,1} x_{v_2} y_{v_1v_2,2} \dots y_{v_{k-1}v_k,k-1}, x_{v_k}.
\]

\smallskip
It is easy to see that this encoding of $W$ is multilinear if and only if $W$ is a path
for in a path no vertex repeats.
Moreover, for the later discussion, the reader should expect a convention where 
$x$-variables correspond to vertices and $y$-variables to edges of a graph in question.

\subsection{Generating polynomial for walks}

\smallskip
Consider an example illustrated in Figure~\ref{fig:walk-encoding}.
Here, let $P_{v_2,\ell-1}$, $P_{v_3,\ell-1}$, and $P_{v_4,\ell-1}$ denote
the polynomial encoding of all walks of length $\ell-1$ ending at $v_2$, $v_3$, and $v_4$,
respectively. 
We construct the polynomial encoding denoting all walks with
length $\ell$ and ending at vertex $v_1$ by writing
\begin{eqnarray*}
P_{v_1, \ell}  = &x_{v_1} y_{v_1v_2,\ell-1} \, P_{\ell-1, v_2} +
                 x_{v_1} y_{v_1v_3,\ell-1} \, P_{\ell-1, v_3} +
                 x_{v_1} y_{v_1v_4,\ell-1} P_{\ell-1, v_4}.
\end{eqnarray*}
The intuition is that we can construct a walk of length $\ell$ for the vertex
$v_1$ using the walks of length $\ell-1$ for its neighbors in $N(v_1) = \{ v_2,v_3,v_4 \}$.
Further, the intuition is that such a setup is appealing algorithmically for an
application of dynamic programming as we will see later on.

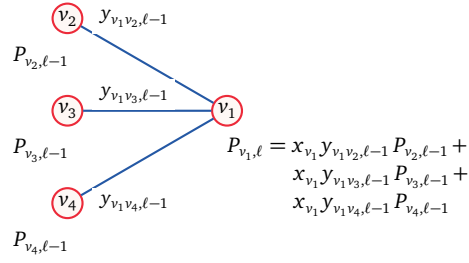
\begin{figure}[t]
\centering
\input{tikz-polynomial-encoding-1}
\caption{\label{fig:walk-encoding}
An illustration of the polynomial encoding of walks.}
\end{figure}

\smallskip
In general, we write $P_{u,\ell}$ to denote the polynomial encoding of 
all walks of length $\ell-1$ ending at vertex $u$.
As such, a generating function $P_{u,\ell}$ for each $u \in V$ and
$\ell \in Z_{+}$ can be written as
\[
P_{u,1} = x_u,
\]
for each $u \in V$, and then
\begin{equation}
\label{eq:gen-polynomial:1}
P_{u, \ell} = x_u  \sum_{v \in N(u)} y_{uv,\ell-1}  P_{v, \ell-1},
\end{equation}
for each $u \in V$ and $\ell \in \{2,\dots,k\}$.

\smallskip
Finally, we denote the polynomial $\mathcal{P}_\ell = \sum_{u \in V} P_{u,\ell}$ for each
$\ell \in [k]$. In other words, $\mathcal{P}_\ell$ is the polynomial
encoding of all walks of length $\ell-1$.

\smallskip
We demonstrate the polynomial encoding of walks in a non-temporal graph using a
toy example.
Let $G=(V,E)$ be a graph with vertex set $V=\{v_1,v_2,v_3\}$ and edge
set $E=\{(v_1,v_2), (v_2,v_3)\}$. Figure~\ref{fig:walk-encoding-demo:1}
(a), (b) and (c) illustrate the encoding of
all walks with length zero, one and two, respectively.
Observe that monomials that correspond to paths are multilinear
and they are highlighted in bold.

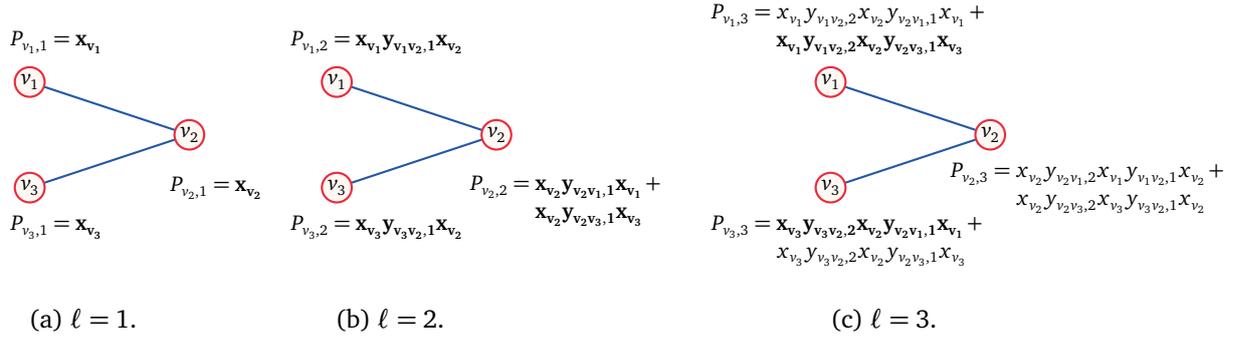
\begin{figure*}[t]
\centering
\setlength{\tabcolsep}{0.1cm}
\begin{tabular}{c c c}
  \input{tikz-polynomial-encoding-demo-1} &
  \input{tikz-polynomial-encoding-demo-2} &
  \input{tikz-polynomial-encoding-demo-3} 
\end{tabular}
\caption{\label{fig:walk-encoding-demo:1}%
An example demonstrating the polynomial encoding of
walks of length zero (a), one (b) and two (c) in a non-temporal
graph. Observe that monomials correspond to paths are multilinear
(highlighted in bold).%
}
\end{figure*}

\subsection{Monomial encoding of a temporal walk}

Let $W^\tau = v_1 e_1 v_2 \dots e_{k-1} v_k$ be a temporal walk in a temporal
graph $G^\tau = (V, E^\tau)$. Let $\{x_{v_1},\dots,x_{v_n}\}$ be a set of
variables corresponding to the vertices in $V = {v_1,\ldots,v_n}$ and 
let $\{y_{uv,\ell,i}: (u,v,i) \in E^\tau, \ell \in [k]\}$ be a set of variables 
where $y_{uv,\ell,i}$ corresponds to an
edge $(u,v,i) \in E^\tau$
that appears at position $\ell$ in~$W^\tau$.
A monomial encoding of $W^\tau$ is represented as
\[
x_{v_1} y_{v_1v_2,1,i_1} x_{v_2} y_{v_1v_2,2,i_2} \dots y_{v_{k-1}v_k,k-1,i_{k-1}}, x_{v_k},
\]
where $i_{1},\dots,i_{k-1}$ denote the
timestamps on the edges $e_1,\dots,e_{k-1}$, respectively.

\smallskip
We make the following observation regarding the monomial encoding of a temporal walk.

\begin{lemma}
\label{lemma:ktemppath:1}
The monomial encoding of a temporal walk $W^\tau$ is multilinear 
if and only if $W^\tau$ is a temporal path.
\end{lemma}
\begin{proof}
Suppose that the temporal walk $W^\tau$ is a temporal path.
By definition, each vertex of $W^\tau$ appears exactly once, so all $x$ terms in its monomial encoding are unique. 
From this, it is also evident that the $y$ terms are unique. 
Thus, the monomial encoding of $W^\tau$ is multilinear.

For the other direction, suppose that $W^\tau$ is not a temporal path.
Because at least one vertex in $W^\tau$ repeats, there is at least one $x$ term in the monomial encoding of $W^\tau$ of degree at least two.
It follows that the encoding is not multilinear, concluding the proof.
\end{proof}

\subsection{Generating polynomial for temporal walks}

In this section, we present a dynamic programming
recursion to generate temporal walks.

\smallskip
Let $P_{u,\ell,i}$ denote the encoding of all walks of length $\ell-1$ ending at
vertex $u$ at latest time $i \in [t]$.

\smallskip
Again, consider the example illustrated in Figure~\ref{fig:poly-encoding:1}, 
where $v_1$ is a vertex such that $N_i(v_1) = \{v_2,v_3,v_4\}$. 
Following our notation, $P_{v_2,\ell-1,i-1}$, $P_{v_3,\ell-1,i-1}$ and
$P_{v_4,\ell-1,i-1}$ represent the polynomial encoding of walks ending at
vertices $v_2$, $v_3$ and $v_4$, respectively, such that all walks have length
$\ell-2$ and end at latest time $i-1$. Further, $P_{v_1,\ell,i-1}$ denotes the polynomial
encoding of all walks of length $\ell-1$, ending at $v_1$ at latest time $i-1$.

The polynomial encoding to represent walks of length $\ell-1$ ending at $v_1$
and at latest time $i$ can be written as
\begin{eqnarray*}
\label{eq:gen-polynomial:2}
P_{v_1,\ell,i} & = & x_{v_1} y_{v_1v_2,\ell-1,i} \, P_{v_2,\ell-1,i-1} + \notag \\
               &   & x_{v_1} y_{v_1v_3,\ell-1,i} \, P_{v_3,\ell-1,i-1} + \notag \\
               &   & x_{v_1} y_{v_1v_4,\ell-1,i} \, P_{v_4,\ell-1,i-1} + P_{v_1,\ell,i-1}.
\end{eqnarray*}

\begin{figure}[t]
\centering
\input{tikz-polynomial-encoding-2}
\caption{\label{fig:poly-encoding:1}
An illustration of the polynomial encoding of temporal walks.}
\end{figure}
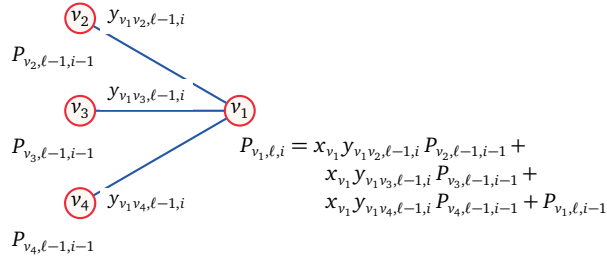

\smallskip
Intuitively, the above equation says that we can reach vertex $v_1$ at time
step $i$ if we have already reached any of its neighbors in $N_i(v_1)$ by latest
time $i-1$. Notice that the term $P_{v_1,\ell,i-1}$ is included so that if we have
reached $v_1$ at latest time $i-1$ we can choose to stay at $v_1$ for time $i$.

\smallskip
By generalizing the above idea, a generating function is obtained by setting
\[ P_{u,1,i} = x_u,\]
for each $u \in V$ and $i \in [t]$, and
\begin{equation}
\label{eq:gen-polynomial:3}
P_{u,\ell,i} = x_u \sum_{v \in N_i(u)} y_{uv,\ell-1,i} P_{v,\ell-1,i-1} + P_{u,\ell,i-1},
\end{equation}
for each $u \in V$, $\ell \in \{2,\dots,k\}$, and $i \in [t]$.

\smallskip
Furthermore, let us form the polynomial
$\mathcal{P}_{\ell, i} = \sum_{u \in V} P_{u, \ell, i}$,
for each $\ell \in [k]$ and $ i \in [t]$. Put differently, $\mathcal{P}_{\ell,i}$
denotes the polynomial encoding of all walks of length $\ell-1$ ending at latest
time~$i$. 

\smallskip
Now, we observe that the problem of detecting a \ktemppath is equivalent to finding an
$x$-multilinear monomial in $\mathcal{P}_{k,t}$. From the construction of the
generating function~(\ref{eq:gen-polynomial:3}) it is clear that the $y$
variables are always distinct, so detecting an $x$-multilinear monomial is
equivalent to detecting a multilinear monomial.

\smallskip
We demonstrate the polynomial encoding of temporal walks using a small example
graph. Consider a temporal graph $G^\tau=(V,E^\tau)$ with vertex set $V =
\{v_1,v_2,v_3\}$ and edges $E^\tau = \{(v_1,v_2,1)\mathrel{,} (v_1,v_2,2),
(v_2,v_3,1), (v_2,v_3,2)\}$. The polynomial encoding of temporal walks of length
zero, one and two is illustrated in Figure~\ref{fig:tempwalk-encoding-demo:1}
(a,b), (c,d) and (e,f), respectively. Again, observe that multilinear monomials
highlighted in
bold correspond to temporal paths in the graph.

\begin{figure*}[]
\centering
\setlength{\tabcolsep}{0.01cm}
\begin{tabular}{c c}
  \input{tikz-polynomial-encoding-temp-11} &
  \input{tikz-polynomial-encoding-temp-12} \\[2ex]
  \input{tikz-polynomial-encoding-temp-21} &
  \input{tikz-polynomial-encoding-temp-22} \\[2ex]
  \input{tikz-polynomial-encoding-temp-31} &
  \input{tikz-polynomial-encoding-temp-32} 
\end{tabular}
\caption{\label{fig:tempwalk-encoding-demo:1}%
An example to demonstrate the polynomial encoding denoting
temporal walks of length zero (a, b), one (c, d) and two (e, f) in a temporal
graph. Observe that monomials corresponding to a temporal path are multilinear
(highlighted in bold).%
}
\end{figure*}

\begin{lemma}
\label{lemma:ktemppath:2}
The polynomial encoding $P_{u,\ell,i}$ in Equation~(\ref{eq:gen-polynomial:3}) contains an
$x$-multilinear monomial of $x$-degree $\ell$ if and only if there exists a temporal
path $W^\tau$ of length $\ell-1$ ending at vertex $u$ at latest time~$i$.
\end{lemma}
\begin{proof}
Suppose that there exists an $x$-multilinear monomial $M$ of $x$-degree $\ell$ in $P_{u,\ell,i}$. By Lemma~\ref{lemma:ktemppath:1}, we know that $M$ corresponds to a temporal path $W^\tau$.
Because $M$ has $x$-degree $\ell$, we have that $W^\tau$ has length $\ell-1$.
Moreover, by the construction of the generating function, all the walks end at vertex $u$ meaning $W^\tau$ also ends at vertex $u$.
Finally, all the monomial encodings of walks of length $\ell-1$ reaching $u$ before time $i$ are preserved using the term $P_{u,\ell,i-1}$, so the time of reaching $u$ can be at most $i$ for $W^\tau$, as required.

For the other direction, we have again by Lemma~\ref{lemma:ktemppath:1} that 
there exists a multilinear monomial $M$ in $P_{u,\ell,i}$. 
Since the length of $W^\tau$ is $\ell-1$, we have that $M$ has $\ell$ unique $x$ terms. In other words, the $x$-degree of $M$ is $\ell$ and our claim follows.
\hfill\end{proof}

\subsection{Multilinear sieving}

From Lemma~\ref{lemma:ktemppath:2} the problem of deciding the existence of a
\ktemppath in $G^\tau$ reduces to detecting the existence of a multilinear
monomial term in $\mathcal{P}_{k,t}$.

\smallskip
Let $L$ be the set of $k$ labels and $[n]$ the set of vertices in $V$.
We express $L$ and $[n]$ in terms of a set of new variables. 
More precisely, for each vertex $i \in [n]$ and label $j \in L$ we introduce a new variable $z_{i,j}$.
The vector of all variables of $z_{i,j}$ is denoted as
$\mathbf{z}$ and the vector of all $y$-variables as $\mathbf{y}$.
At this point, we turn to the multilinear sieving technique of Bj\"{o}rklund, Kaski and Kowalik~\cite{Bjorklund2014}.

\begin{lemma}[Multilinear sieving~\cite{Bjorklund2014}]
\label{lemma:ktemppath:3}
The polynomial $\mathcal{P}_{k,t}$ has at least one multilinear monomial if and
only if the polynomial
\begin{equation}
\label{eq:ktemppath:3}
Q(z, \mathbf{y}) = \sum_{A \subseteq L}
                            \mathcal{P}_{k,t}(z_1^A,\dots,z_n^A, \mathbf{y})
\end{equation}
is not identically zero, where $z_i^A = \sum_{j \in A} z_{i,j}$
for all $i \in [n]$ and $A \subseteq L$.
\end{lemma}
From the lemma, we can determine the existence of a multilinear monomial in
$\mathcal{P}_{k,t}$ by making $2^k$ substitutions of the new variables in $\mathbf{z}$ in Equation~(\ref{eq:ktemppath:3}). 
In fact, as detailed in~\cite{Bjorklund2014}, these substitutions can be \emph{random} for a low-degree polynomial which is not identically zero has only few roots.
Indeed, if one evaluates the said polynomial at a random point, one is likely to
witness that it is not identically zero (see e.g., Schwartz~\cite{schwartz}, Zippel~\cite{Zippel1979}).
Therefore, the algorithmic framework this approach gives rise to is randomized
with a false negative probability $(2k-1)/2^{b}$ where the arithmetic is over
the finite field GF($2^b$) (for full technical details,
see~\cite{BjorklundKK2016}).\footnote{In practice, to make the probability of a
false negative very low one can choose $b$ to be large enough, say $b = 64$.}
This result can be used to design an algorithm for the \ktemppath problem.

\begin{lemma}
\label{lemma:ktemppath:4}
There exists an algorithm for solving the \ktemppath problem
in $\bigO(2^k k(nt + m))$ time and $\bigO(nt)$ space.
\end{lemma}
\begin{proof}
The algorithm will first construct the polynomial $\mathcal{P}_{k,t}$ representing all temporal walks of length $k-1$ using the recursion in Equation~(\ref{eq:gen-polynomial:3}). Afterwards, it checks if there exists an $x$-multilinear monomial term in $\mathcal{P}_{k,t}$ via Lemma~\ref{lemma:ktemppath:3}. Then, by Lemma~\ref{lemma:ktemppath:2}, the existence of an $x$-multilinear monomial term of size $k$ implies the existence of a \ktemppath of length $k-1$ and vice versa.

In more detail, for $\ell \in [k]$, computing $\mathcal{P}_{\ell,i}$
requires $\bigO(n + m_i)$ additions and multiplications in the field $GF(2^b)$, where $m_i$ is the number of edges at time instance $i \in [t]$. Furthermore, computing the polynomial for all $i \in [t]$ requires $\bigO(nt+m)$ additions and multiplications. 
Finally, we need to iterate the computations over $\ell \in
[k]$, so, we have $\bigO(k(nt+m))$ additions and multiplications. 
The overall
complexity of the algorithm is $\bigO(2^k k (nt+m) (A(2^b)+M(2^b))$, where
$A(2^b)$ and $M(2^b)$ are the complexities of addition and multiplication in the field $GF(2^b)$, respectively.  
Asymptotically, it is known that $M(2^b) = \bigO(b \log b)$ and
$A(2^b) = \bigO(\log b)$ (see~\cite{LinAHC2016}). Therefore, the overall complexity of the algorithm is $\bigO(2^k k (nt+m) b \log{b})$.

In order to compute $\mathcal{P}_{\ell,i}$ we need to remember the computations of $P_{u,\ell,i-1}$ and $\mathcal{P}_{\ell-1,i-1}$ for all $v \in V \setminus u$ and for all $i \in [t]$.
However, we can reuse space since computing the $\ell$-th polynomial encoding only depends on the $(\ell-1)$-th polynomial encoding.
For each vertex $v \in V$ we need a field variable and we
need to compute this for all $i \in [t]$ (this is also required for even staying at the vertex without moving). Thus, the memory requirement is $\bigO(nt)$.
\hfill\end{proof}
As a corollary, one can also observe that the given algorithm establishes that \ktemppath is fixed-parameter tractable in polynomial space.

\subsection{Constrained multilinear sieving}
We now move on to extending the previous approach to detect \pathmotif using the \emph{constrained} multilinear sieving technique due to~\cite{BjorklundKK2016}.

\smallskip
If we observe carefully, to obtain a \pathmotif we need to find a multilinear
monomial term in the polynomial $\mathcal{P}_{k,t}$ such that the vertex colors
corresponding to the $x$-variables with degree one agree with that of the
multiset $M$. This can be done by imposing additional constraints while
evaluating the sieve.

\smallskip
Let $C$ be a set of $n$ colors and $c: [n] \rightarrow C$ a function that
associates each $i \in [n]$ to a color in $C$. For each color $s \in C$, let us
denote the number of occurrences of color $s$ by $\mu(s)$.  A monomial
$x_1^{d_1} \dots x_q^{d_q} y_1^{f_1} \dots y_r^{f_r}$ is properly colored if for
all $s \in C$ it holds that $\mu(s) = \sum_{i \in c^{-1}(s)} d_i$, i.e., the
number of occurrences of color $s$ is equal to the total degree of $x$-variables
representing the vertices with color $s$.

\smallskip
For each $s \in C$, let $S_s$ be the set of $\{\mu(s)\}$ with color $s$
such that $S_s \cap S_{s'} = \emptyset$ for all $s \neq s'$. For $i \in [n]$ and $d \in
S_{c(i)}$ we introduce a new variable $v_{i,d}$. Let $L$ be a set of $k$
labels. For each $d \in \cup_{s \in C} S_s$ and each label $i \in L$ we
introduce a new variable $w_{i,d}$.

\begin{lemma}[Constrained multilinear sieving~\cite{BjorklundKK2016}]
\label{lemma:motif:1}
The polynomial $\mathcal{P}_{k,t}$ has at least one monomial that is both
$x$-multilinear and properly colored if and only if the polynomial
\begin{equation}
Q(z,\mathbf{w},\mathbf{y}) = \sum_{A \subseteq L}
                             \mathcal{P}_{k,t}(z_1^A,\dots,z_n^A, \mathbf{y})
\end{equation}
is not identically zero, where
\begin{equation}
\label{eq:motif:1}
z_i^A = \sum_{j \in A} z_{i,j},
\text{\: and \:}
z_{i,j} = \sum_{d \in S_{c(i)}} v_{i,d} w_{d,j}.
\end{equation}
\end{lemma}

To obtain an algorithm running time $\bigO(2^k k(nt + m))$ using $\bigO(nt)$
space for \pathmotif we proceed similarly to the proof of
Lemma~\ref{lemma:ktemppath:4}, but instead apply Lemma~\ref{lemma:motif:1} to
the constructed polynomial. Further and similarly, it follows that \pathmotif
problem is fixed-parameter tractable.

%

\subsection{Vertex-localization}
\label{sec:vertex-localization}
The idea of vertex-localization was introduced by Kaski~{et al.}~\cite{kaskiLT2018} 
while engineering the multilinear sieving approach to scale for
large \emph{query} sizes (for non-temporal graphs). The approach is effective
to sieve out the vertices that are not incident to at least one match, thus
reducing the graph size. We extend their idea of vertex-localization to
temporal graphs. More precisely, we extend the sieve construction to identify
the set of vertices that have at least one temporal path agreeing with specific
constraints and end at these vertices.

\smallskip
In this approach instead of working with the polynomial encoding of temporal
walks $\mathcal{P}_{k,t}$ for all vertices,
we operate on a family of multivariate polynomials
$P_{u,k,t}$ for each individual $u \in V$ simultaneously. We localize our
evaluation to each individual vertex $u$ by performing $2^k$ random evaluations
of the variable $z$ on the localized polynomial $P_{u,k,t}$, independently.
Furthermore, the polynomial ${P}_{u,k,t}$ has at least one monomial that is both
$x$-multilinear and properly colored if and only if the polynomial
\begin{equation}
Q_u(z,\mathbf{w},\mathbf{y}) = \sum_{A \subseteq L}
                             {P}_{u,k,t}(z_1^A,\dots,z_n^A, \mathbf{y})
\end{equation}
is not identically zero, where
\begin{equation}
\label{eq:motif:1}
z_i^A = \sum_{j \in A} z_{i,j},
\text{\: and \:}
z_{i,j} = \sum_{d \in S_{c(i)}} v_{i,d} w_{d,j}.
\end{equation}

\smallskip
Thus, if the polynomial $Q_u(z,\mathbf{w},\mathbf{y})$ evaluates to a non-zero term, it
follows from Lemma~\ref{lemma:ktemppath:2} that there exists at least one
temporal path of length $k-1$ ending at vertex $u$ satisfying the
constraints. By independently evaluating the set of polynomials 
$Q_u(z,\mathbf{w},\mathbf{y})$ for each $u \in V$,
one can obtain a set of vertices for which there exists at least
one temporal path satisfying the constraints and ending at these vertices. Most
importantly, this redesign of the sieve comes with no change in the asymptotic
time and space complexity. It is to be noted that there is a per-vertex false negative
probability of $(2k-1)/2^b$ where the arithmetic is over the finite field
GF($2^b$). Again, in practice we can choose a large enough field size (say $b=64$)
to keep the error probability very low.

For the full technical details of multilinear sieving with vertex-localization, we
refer the reader to the work of Kaski et al.~\cite[\S2]{kaskiLT2018}

\subsection{Finding optimum timestamp}
In this section we describe a procedure to obtain an optimum timestamp. By
optimum we mean that the maximum timestamp of the edges in the temporal path is
minimized. We refer to our algorithm for the decision variant
as \emph{decision oracle}. 

\smallskip
To find the minimum (optimum) timestamp $t' \in [t]$,
we make $\bigO(\log t)$ queries to the decision oracle using binary
search on range $\{1,\dots,t\}$. More precisely, we construct a polynomial encoding of all
walks of length $k-1$ which end at latest time $t'$ and query the oracle for
the existence of a path. 
Thus, we need at most $\log{t}$ queries to the decision
oracle to find the minimum timestamp $t'$ for which there exists a match.
As such, the overall complexity of finding the optimum timestamp is $\bigO(2^k k (nt+m) \log{t})$.

\subsection{Extracting a solution}
\label{sec:extract-solution}
In the previous sections we described an algebraic solution for the decision
version of the \pathmotif problem. 
That is, the described algorithm strictly answers YES/NO to the question of whether a solution exists.
However, in many cases, we need to extract an explicit solution if such exists. We propose two approaches to extract an optimal
solution. Our first approach uses the self-reducibility of the problem while the second approach uses vertex-localization.

\medskip
\spara{Using self-reducibility.}
The extraction process works in two steps: 
\squishlist
\item[($i$)] extract a $k$-vertex temporal subgraph that contains a
temporal path matching the multiset query;
\item[($ii$)] extract a temporal path in the subgraph from ($i$) using temporal DFS.
\squishend

\smallskip
{\it Extracting a subgraph.} We use the decision oracle as a subroutine to find a
solution in $\bigO(n)$ queries as follows: for each vertex $v \in V$, we
remove the vertex $v$ and the edges incident to it and query the oracle. If
there is a solution (i.e., the oracle returns YES), then we continue to next vertex; otherwise we put back $v$
and the edges incident to it, and continue to next vertex. In this way, we can
obtain a subgraph with $k$ vertices in at most $n-k$ queries to the oracle.
However, the number of queries to the decision oracle can be reduced to $\bigO(k \log n)$ queries in expectation by recursively dividing the graph in to two halves~\cite{Bjorklund2014}.

\smallskip
{\it Extracting a temporal path.} 
We pick an arbitrary start vertex in the subgraph
obtained from step~($i$) and find a temporal path connecting all the $k$
vertices using temporal DFS. If such a path does not exist, then we continue to the next vertex. 
Even though the worst case complexity of this approach is $\bigO(k!)$, we will demonstrate later that it is highly practical. 
In contrast, extracting a solution can be done using
$\bigO(k)$ queries to the decision oracle using vertex-localized sieving. 
However, we leave a vertex-localization variant of sieving for future work.

\smallskip
In summary, the overall complexity of extracting an optimum solution using
self-reducibility is
\[\bigO((2^k (nt + m) (k \log{n} + \log{t})) + k!).\]

\medskip
\spara{Using vertex-localization.} 
Again, the process works in two steps: ($i$) we identify the set of vertices for
which there exists at least one temporal path satisfying the constraints and end
at the vertex, ($ii$) we choose a vertex from the set of vertices obtained
from step ($i$) and perform a reverse temporal-DFS (with decreasing order of the
timestamps) such that the path ends at the chosen vertex.

\smallskip
As described earlier (cf.~\S\ref{sec:vertex-localization}), vertex-localization
comes with no additional cost with respect
either space or time. The temporal DFS takes $\bigO(\Delta^k)$ time, where
$\Delta$ is the maximum degree of the temporal graph.
The details of the algorithm and its runtime analysis is described in
Section~\ref{sec:experimental-setup:baseline}. So, the overall complexity of
extracting an optimum solution using vertex-localization is 
\[\bigO((2^k k (nt + m) \log(t)) + \Delta^k).\]

To summarize, the runtimes of the decision and extraction variants of the
algorithm are reported in Table~\ref{table:algorithm:1}.

\begin{table*}[t]
\caption{A summary of algorithm runtimes for the \pathmotif and \ktemppath
problems.}
\label{table:algorithm:1}
\centering
\footnotesize
\begin{tabular}{l l}
\toprule
Algorithm & Time complexity \\
\midrule
{\it \underline{Deciding the existence of a solution}} & \\[1ex]
Deciding existence & $\bigO(2^k k (nt+m))$ \\
Finding optimum timestamp & $\bigO(2^k k (nt+m) \log t)$ \\
\midrule
{\it \underline{Extracting an optimum solution}} & \\[1ex]
Self-reducibility + Temporal-DFS & $\bigO(2^k k (nt+m) (k \log n + \log t) + k!)$ \\
Vertex-localization + Temporal-DFS & $\bigO((2^k k (nt + m) \log(t)) + \Delta^k)$\\
\bottomrule
\end{tabular}
\end{table*}


\section{Extending the framework}
\label{sec:algorithm-ext}

Our multilinear sieving method is not limited to solving the \ktemppath and
\pathmotif problems. In this section, to demonstrate this, 
we describe extensions of our method to solve many similar variants of these problems.

\subsection{Solving \colorfulpath and \sdcolorfulpath problems}
The \colorfulpath problem is a special case of the \pathmotif problem with a multiset
query of different colors. As such, one can use the
algorithm described to solve the \pathmotif problem to
solve the \colorfulpath problem. Furthermore, to adapt the algorithm for the \sdcolorfulpath
problem, it suffices to assign unique colors to the source and the destination vertices, 
and include these colors in the multiset. 
More precisely, given an instance of \sdcolorfulpath, 
we extend the given vertex-coloring function $c$ by also setting $c(s) = k+1$ and $c(d) = k+2$, 
where $k+1$ and $k+2$ are new unique colors.
Further, we set $M' = M \cup \{k+1,k+2\}$ and apply the algorithm described for \pathmotif
on the produced instance $(G^\tau, c, M', k+2)$.

\subsection{Solving the \rainbowpath problem}
Given a \rainbowpath problem with coloring function $c:V \rightarrow [q]$, we
query the decision oracle of the \pathmotif problem with $k$ combinations of $q$
colors. In total we need at most ${q}\choose{k}$ queries to decide the existence
of a \rainbowpath. So the overall complexity to decide the existence of a
\rainbowpath is $\bigO(q^k 2^k k (nt+m))$.

\smallskip
Note that the algorithm is practical for small values of $q$.

\subsection{Solving the \ecpathmotif problems}

Recall that in the \ecpathmotif problem we are given a temporal-graph
$G^\tau=(V,E^\tau)$, a tuple $T=(j_1,\dots,j_{k-1})$ of timestamps such that
$j_i < j_{i+1}$ for all $i \in [k-2]$ and a
multiset $M$ of colors, and the goal is to decide whether there exists a
temporal path
matching $T$ in the specified order with vertex colors agreeing with $M$. We
construct a generating polynomial for \ecpathmotif such that

\[ P_{u,1,j_0} = x_u, \]

for each $u \in V$ and $j_0 = 1$; and

\[
P_{u,\ell,j_{\ell-1}} = x_u \sum_{(u,v,j_{\ell-1}) \in E^\tau} 
                        y_{uv,\ell-1,j_{\ell-1}} P_{v,\ell-1,j_{\ell-2}}
\]

for each $u \in V$, $\ell \in \{2,\dots,k\}$.

\smallskip
From Lemmas~\ref{lemma:ktemppath:2} and \ref{lemma:motif:1}, it follows that
existence of a multilinear monomial in the polynomial
$\mathcal{P}_{k,j_{k-1}} = \sum_{u \in V} P_{u,k,j_{k-1}}$ would
imply the existence of a \ecpathmotif of length $k-1$ and satisfying the
constraints.
The algorithm requires $\bigO(2^k(nk+m))$ time and $\bigO(n)$ space. In more
detail, computing $\mathcal{P}_{\ell,j_{\ell-1}}$ requires $n+m_{\ell-1}$ 
additions and multiplications. Furthermore, for all $\ell \in [k-1]$ we require
$n(k-1) + \sum_{\ell=1}^{k-1} m_{j_{\ell}} = \bigO(nk+m)$ additions and
multiplications. So the overall time complexity of the algorithm is 
$\bigO(2^k (nk+m))$. For computing $\mathcal{P}_{\ell, j_{\ell-1}}$, we only 
need $\mathcal{P}_{\ell-1, j_{\ell-2}}$. So the
space complexity is $\bigO(n)$.

\smallskip
Note that \ectemppath problem is a special case of \ecpathmotif problem such
that the motif query has a single color i.e, $M=\{1^k\}$. Likewise, we use
the algorithm described above to solve the \ectemppath problem and the
complexity results of the algorithm follow from \ecpathmotif problem.


\subsection{Solving the \vcpathmotif problem}
Recall that in the \vcpathmotif problem we are given a temporal
graph $G^\tau=(V,E^\tau)$, a coloring function $c:V \rightarrow [q]$, a tuple
$M=(c_1,\dots,c_k)$ of colors, and the goal is to decide whether there exists a
temporal path $P^\tau$ such that the vertex colors of $P^\tau$ match with $M$ 
in the specified order. We construct a generating polynomial for the
\vcpathmotif problem such that

\[
P_{u,1,i} = x_u,
\]

for all $u \in V$, $i \in [t]$ and $c(u) = c_1$; and

\[
P_{u,\ell,i}  = x_u \sum_{v \in N_i(u)} y_{uv,\ell-1,i} P_{v,\ell-1,i-1} +
                P_{u,\ell,i-1}
\]

for all $u \in V$, $\ell \in \{2,\dots,k\}$, $i\in [t]$, and $c(u) = c_\ell$,
$c(v) = c_{\ell-1}$.

\smallskip
From Lemmas~\ref{lemma:ktemppath:2} and \ref{lemma:motif:1}, it follows that
the existence of a multilinear monomial in the polynomial generated above would
imply the existence of a \ecpathmotif and vice versa.
The runtime and space complexity results follow.

\subsection{Solving the \vccolorfulpath problem}
\label{sec:algorithm-ext:vccolorfulpath}
The \vccolorfulpath problem in temporal graphs is solvable in polynomial time.
Indeed, we present a dynamic programming algorithm for solving the problem in
$\bigO(mt)$ time and $\bigO(nt)$ space.

\smallskip
Again, recall that in the \vccolorfulpath problem we are given a temporal graph
$G^\tau=(V,E^\tau)$, an integer $k \leq n$, a coloring function 
$c: V \rightarrow [k]$, and a tuple $M$ of $k$ different colors. The problem asks
us to decide whether there exists a temporal path $P^\tau$ of length $k-1$ such
that the vertex colors of $P^\tau$ match with $M$ in the specified order.

\smallskip
Without loss of generality, let $M=(1,\dots,k)$. Let $V_i \subseteq V$ denote the set of
vertices with color $i \in [k]$.
Define an indicator variable $I_{u,i}$ for each $u \in V$ and 
$i \in [t]$ to indicate the existence of an \vccolorfulpath. Our
dynamic programming recursion works as follows:
for all $u \in V_1$ and $i \in [t]$, initialize $I_{u,i}=1$.
Iterate over $\ell \in \{2,\dots,k\}$.
For each $u \in V_{\ell}$ set $I_{u,j} = 1$ for all $j \in \{i+1,\dots,t\}$ if there exists an edge $(u,v,i)
\in E^\tau$ such that $v \in V_{\ell-1}$ and $I_{v,i}=1$. Finally, if $I_{u,i} = 1$
for any vertex $u \in V_k$ and timestamp $i \in [t]$, it implies that there
exists a \vccolorfulpath ending at vertex $u$ at time $i$.

\smallskip
As we have a total of $nt$ indicator variables, the space complexity is $\bigO(nt)$. 
For each edge $e \in E^\tau$ we update at most $t$ indicator variables, so the
algorithm takes $\bigO(mt)$ time.

\subsection{Path motif problem with delays}
In a transport network a transition between any two locations may
involve a \emph{transition time} and a minimum \emph{delay time} at a location
before continuing the journey, for example a minimum time
to visit a museum. In this
section, we introduce a problem setting with transition and delay
times,
and present generating polynomials to solve the problems.

\smallskip
For a temporal graph $G^\tau=(V,E^\tau)$, an edge $e \in E^\tau$ is a
tuple $(u,v,i,\epsilon)$ where $u,v \in V$, $i \in [t]$ is the edge
timestamp and $\epsilon \in \mathbb{Z}_{+}$ is the transition time from $u$ to $v$.
Additionally, each vertex has a delay time $\delta:V \rightarrow
\mathbb{Z}_{\geq 0}$.

\smallskip
We consider the following cases.

\smallskip
{Encoding only with delay}:
\[
P_{u,\ell,i} = x_u \, \sum_{(u,v,i,\epsilon) \in E} y_{uv,\ell-1,i}
              P_{v,\ell-1,i-\delta(v) }
              + P_{u,\ell,i-1}.
\]

{Encoding only with transition time}:
\[
P_{u,\ell,i+\epsilon} = x_u \, \sum_{(u,v,i,\epsilon) \in E} y_{uv,\ell-1,i}
                  P_{v,\ell-1,i-1}
                    + P_{u,\ell,i-1}.
\]

{Encoding with transition time and delay}:
%
\[
P_{u,\ell,i+\epsilon} = x_u \, \sum_{(u,v,i,\epsilon) \in E} y_{uv,\ell-1,i}
              P_{v,\ell-1,i-\delta(v)} + P_{u,\ell,i-1}.
\]

\smallskip

From Lemmas~\ref{lemma:ktemppath:2} and \ref{lemma:motif:1}, it follows that
existence of a multilinear monomial in the polynomial generated above would
imply the existence of a \pathmotif (or \ktemppath).
The time and space complexity of the algorithm follow.

\subsection{Including wildcard entries}

We now describe how to extend our framework to include wildcard entries.
We can extend the \pathmotif problem so that each vertex is colored
with a set of colors instead of a single color. This extension enables us to
have wildcard entry matches, which can be desirable in scenarios where vertices
are not immediately adjacent but rather some uncertain distance away.


\smallskip
Let $c:V \rightarrow [q]$ be the color mapping defined in the original problem
instance. We introduce an additional color $q+1$ and associate it with each
vertex in the graph.
More precisely, $c':V \rightarrow \{c(v), q+1\}$. In other words, each vertex $v
\in V$ is mapped to two colors; one of which is defined by $c$ and the other is
$q+1$.
In the modified problem instance each vertex can either match to the original
color defined by $c$ or the color $q+1$.


\smallskip
We discuss the approach for solving the \pathmotif problem with wildcard
entries, however, for other variants of the problem the method is similar. 
Given a \pathmotif instance
$(G^\tau,c,M,k)$ with $c: V \rightarrow [q]$, we construct an instance
$(G^\tau,c',M',k')$ such that $c':V \rightarrow \{c(v), q'\}$, $M'=M$, $k'=k$
$q'=q+1$ and query the decision oracle. 
Now if the oracle returns a YES, we report the
existence of a match. Otherwise, we increment $k'=k'+1$, add a color $q'$ to
the multiset i.e, $M'=M'\cup\{q'\}$ and query the decision oracle with instance
$(G^\tau,c',M',k')$. 
Again, if the decision oracle returns a YES, we report the
existence of a match with $k'- k$ wildcard entries, otherwise we repeat the
procedure of adding color $q'$ to $M'$, incrementing $k'$ and querying the
decision oracle up to some maximum value of $k'$.

%

\section{Implementation}
\label{sec:implementation}
We use the design of Bj{\"o}rklund et al.~\cite{BjorklundKKL2015} as a starting
point for our implementation, in particular we make use of their fast
finite-field arithmetic implementation.

\mpara{Intuition.} A high-level intuition of the approach is as follows: we
assign each variable in a monomial a value in the field $\mathit{GF}(2^b)$. The
multiplication between any two field variables is defined as a XOR operation.
Likewise, if we multiply two variables with the same value they cancel out each
other and the resulting monomial has a zero value. So, even though the generating
polynomial has monomials that are not $x$-multilinear, the contributions from
such monomials will cancel out during the evaluation. It is to be noted that
the actual implementation is not identical to this description, but the
high-level idea is similar. For the implementation details of the finite-field
arithmetic, we refer the reader to the work of Bj{\"o}rklund et al.~\cite[\S\,3.3]{BjorklundKKL2015}.

\smallskip
A natural challenge in implementation engineering is to saturate the arithmetic
bandwidth of the hardware, while simultaneously keeping the memory pipeline
busy. Modern computing architectures have high memory bandwidth, however, the
increase in bandwidth comes at the cost of latency. More precisely, after the
processor issues a memory fetch instruction it takes many clock cycles to fetch
the data from the main memory and make it available on the registers. Often the
memory latency is orders of magnitude greater than the latency of arithmetic
operations. Now, the challenge is to keep the processer busy with enough
arithmetic instructions for computation meanwhile the memory pipeline is busy
fetching the data for subsequent computations. 
The memory interface can be effectively utilized using coalesced memory accesses,
by arranging the memory layout such that the data used for consecutive
computation is available in consecutive memory addresses. 
The arithmetic bandwidth can be saturated by enabling parallel executions of
the same arithmetic operations, which are enabled using vector extensions. 
More precisely, if
we are executing the same arithmetic operation on different operands, 
then we can group the operands using vector extensions to execute arithmetic operations
in parallel, thereby increasing the arithmetic throughput.
The combination of memory coalescence and vector extensions 
are often used to speedup the computation in the algorithm-engineering
community~\cite{BellG2008}.

\smallskip
Our engineering effort boils down to implementing the generating
function~(\ref{eq:gen-polynomial:3}) and evaluating the recurrence at $2^k$
random points. Specifically, we introduce a domain variable $x_v$ for each $v
\in V$ and a support variable $y_{uv,\ell,i}$ for each $\ell \in [k]$ and
$(u,v,i) \in E^\tau$. In total, there are $\bigO(n)$ domain variables and
$\bigO(mk)$ support variables. 
The values of variables $x_v$ are computed using
Equation~(\ref{eq:motif:1}) and
the values of variables $y_{uv,\ell,i}$ are assigned
uniformly at random using a pseudorandom number generator, on the fly without
storing in memory. 
Observe that the variables $x_v$ and
$y_{uv,\ell,i}$ are used exactly once during the computation of
$\mathcal{P}_{k,t}$. Recall that in theory our algorithm has a false
negative probability of $\frac{2k-1}{2^b}$, however, in practice the
false negative probability depend on the quality of the random number
generator.
Our implementation of the recurrence in Equation~(\ref{eq:gen-polynomial:3})
loops over three variables: the outermost loop is over $[k]$, second loop over
$[t]$ and final loop over $V$. So we compute the value of $P_{u,\ell,i}$ for all
$u \in V$ with $\ell$ and $i$ fixed. Precisely, for each iteration of inner most
loop we compute the value of ${P}_{u, \ell, i}$.

\smallskip
The implementation borrowed from our earlier work \cite{conf-paper} uses
$\bigO(ntk)$ memory, nevertheless, in our current work we improve the memory
footprint of the implementation to $\bigO(nt)$ memory, as claimed in theory.
For simplicity, we refer our implementations of using $\bigO(ntk)$
and $\bigO(nt)$ memory as \emph{genf-1} and \emph{genf-2}, respectively. 
Recall that computing the $\ell$-th polynomial encoding $\mathcal{P}_{\ell,i}$ would
require only the $(\ell-1)$-th polynomial encodings $\mathcal{P}_{\ell-1,i-1}$ and
$P_{u,\ell,i-1}$. 
Storing $\mathcal{P}_{\ell,i}$ and $\mathcal{P}_{\ell-1,i}$
for all $i \in [t]$ requires $2 nt$ memory. However, computing
$\mathcal{P}_{\ell+1,i}$ for all $i \in [t]$ does not require
$\mathcal{P}_{\ell-1, i}$, so we can reuse the memory used
to store $\mathcal{P}_{\ell-1, i}$. This is achieved by swapping the array
pointers and initializing the array at each iteration of the outermost loop over
$[k]$.

\smallskip

\smallskip
If we observe carefully, the dynamic programming recursion~(\ref{eq:gen-polynomial:3}),
computing $P_{u,\ell,i}$ is independent for each vertex.
More precisely, computing $P_{u,\ell,i}$ is independent
for each $u \in V$, provided that we fix the values of $\ell$ and $i$. 
So the algorithm is thread-parallelizable up to $n$ threads. 
Likewise, we employ OpenMP API using the ``\texttt{\small omp parallel for}''
construct with \texttt{\small default} scheduling over the vertices in $V$ to
achieve thread-level parallelism. 
Furthermore,
we need to evaluate at $2^k$ random points, that is, we need $2^k$~ran\-dom
substitutions of the variables in $\mathbf{z}$ to evaluate $\mathcal{P}_{k,t}$
and these substitutions are independent of each other. 
More precisely, vector parallelization is achieved on $2^k$
random points of evaluation, since we perform the same arithmetic operation
on different set of data. 
So the algorithm is
vector-parallelizable up to $2^k n$. 
In order to reduce the memory access latency we organize our memory layout as $k
\times t \times n$ and $t \times n$ for \emph{genf-1} and \emph{genf-2}, respectively;
furthermore, we employ hardware pre-fetching \cite[\S\,3.6]{BjorklundKKL2015} to
saturate the memory bandwidth.

\smallskip
Our implementation is written in the C programming language with OpenMP constructs
to achieve thread-level parallelism.
The source code is compiled using
\texttt{\small -march=native} and \texttt{\small -O5} optimization flags to enable
architecture-specific optimization and instruction set extensions. The running
time is measured using OpenMP time interface \texttt{\small omp\_get\_wtime} and memory
usage is tracked using wrapper functions around standard C memory allocation
subroutines \texttt{\small malloc} and \texttt{\small free}.
We support extracting an optimum solution using two
approaches: the first approach uses
self-reducibility of the interval oracles, based on the work of Bj{\"o}rklund et
al.~\cite{BjorklundKK-esa2014}; and our second approach makes use of
vertex-localization and temporal-DFS.
We support both directed and undirected graphs.

\subsection{Preprocessing}
\label{sec:implementation:preproc}
We take advantage of two preprocessing steps to reduce the size of the input
graph: ($i$) we remove all vertices whose vertex color do not match with the
multiset colors; ($ii$) we merge the temporal graph instance to a non-temporal
graph, build a vertex-localized sieve on the non-temporal instance and reconstruct a
temporal graph using the list of vertices, which are incident to at least one match in
the vertex-localized sieve.

\smallskip
Recall that in the vertex-localized sieving instead of working with polynomial encoding of
all walks for all vertices, we localize the generating polynomial to each
vertex. More precisely, we construct polynomial encoding of all walks ending at
each vertex, independently and perform $2^k$ random evaluations for the
localized polynomial. If the sieve evaluates to a non-zero value, it indicates
that there exists at least one path ending at vertex satisfying the constraints.
However, in case of non-temporal graphs, the sieve evaluates to a non-zero if the
vertex is incident to (part of) at least one match, this helps us to obtain a
list of vertices that are incident at least one match, furthermore helps us to
remove all vertices that are not incident to a match. We suggest the reader to refer the work of
Kaski et al.~\cite{kaskiLT2018} for a detailed discussion of
vertex-localization in non-temporal graphs.

\smallskip
For preprocessing using vertex-localization, we engineer an implementation of the
generating function~(\ref{eq:gen-polynomial:1}) and evaluate the recurrence at
$2^k$ random points. Specifically, we compute a set of polynomials
\{$P_{u,k}$: for each $u \in V$\} and evaluate them independently.
Finally, we obtain a set of vertices for which the polynomial $P_{u,k}$
evaluates to a non-zero term.
We introduce a domain variable $x_v$ for each
$v \in V$ and a support variable $y_{uv,\ell}$ for each $\ell \in [k]$ and
$(u,v) \in E$. In total, there are $\bigO(n)$ domain variables and $\bigO(m)$
support variables. 
The values of variables $x_v$ are computed using
Equation~(\ref{eq:motif:1}) and
the values of variables $y_{uv,\ell}$ are assigned
uniformly at random using a pseudorandom number generator, on the fly without
storing in memory. 
%
%
Our implementation of the recurrence in
Equation~(\ref{eq:gen-polynomial:1}) loops over two variables: the outer loop is
over $[k]$ and the inner loop over $V$. We compute the value of $P_{u,\ell}$
for all $u \in V$ in the inner loop for each iteration of
$\ell \in [k]$. Precisely, for each iteration of the inner
most loop we evaluate a family of polynomials {${P}_{u, \ell}$: for each $u \in V$},
independently.
The generating function ~(\ref{eq:gen-polynomial:1}) has an asymptotic memory
complexity of $\bigO(n)$. More precisely, we
need $2 \cdot n$ memory to store $\mathcal{P}_\ell$ and $\mathcal{P}_{\ell-1}$.

\smallskip
The first preprocessing step is trivial, however, the approach is ineffective
for \colorfulpath, \sdcolorfulpath problem and also if a vertex matches to more
than one color (for example, while using wildcard entries). 
The second approach is non-trivial and unique to our generating
function construction, which can be employed
to further reduce the graph size even after applying the trivial first step. 
We make use of the fact that if there exists a temporal
path in a temporal graph it implies that there exist a path in the corresponding
non-temporal instance of the graph, however, the vice versa is not always true.
To take advantage of this fact, we build a vertex-localized sieve on the
non-temporal instance, which return a list of vertices that are incident to at
least one match. Furthermore, we reconstruct a temporal graph using the vertices
that are incident to at least one match on the static graph.  In practice, it is
expected that there most likely will not be many matches, as a consequence, the
size of the input graph could reduce significantly depending on the number of
incident vertices. The runtime of our algorithm rely on the number of target
matches in the graph only if the preprocessing step is applied.  Recall that our
algorithm is randomized and has a per-vertex false negative probability of
$\frac{2k-1}{2^{b}}$ \cite{kaskiLT2018}, for most of our experiments we use
field size of $\mathit{GF}(2^{64})$, which makes the false negative probability
very close to zero and a single run of vertex-localized sieve construction is
sufficient. However, for smaller field size multiple repetition of the
experiment is required to avoid false negatives.

\smallskip
Our software is available as open source \cite{conf-code, journal-code}.
\section{Experimental setup}
\label{sec:experimental-setup}
In this section we discuss our experimental setup.

\subsection{Baseline}
\label{sec:experimental-setup:baseline}
For the problems considered in this paper we are not aware
of any known baselines that provide an exact or approximate solution.
Thus, to compare our techniques with non-algebraic methods, 
we implement two naive baselines:
\squishlist
\item[($i$)] an \emph{exhaustive-search} algorithm using temporal DFS, and
\item[($ii$)] a \emph{brute-force} algorithm based on random walks.
\squishend

\mpara{Exhaustive-search.} 
In this technique, we pick a vertex $v \in V$ and perform
temporal DFS starting from $v$ by restricting the depth of the DFS to $k$.
Every time we reach depth $k$, we check if the set of vertices in the path satisfies
the multiset colors. If it does, we update the solution to our optimal
solution only if the maximum timestamp in the temporal path is less than the
current solution, otherwise we continue. We repeat the process for
each vertex. 

\smallskip
The runtime of the exhaustive-search algorithm is bounded by $\bigO(n\Delta^k)$,
where $\Delta$ is the maximum degree of the graph. The runtime analysis is as
follows: as a worst case let us assume that each vertex has degree $\Delta$. While
performing temporal DFS each vertex has one incoming edge and $\Delta-1$ outgoing edges. 
So the temporal DFS tree of depth $k$ has at most $\bigO(\Delta^k)$ vertices to
visit.\footnote{More precisely, the number of vertices in the temporal DFS tree is
$\frac{(\Delta-1)^{k-1}+\Delta-3}{\Delta-2}$.}
Furthermore we perform temporal DFS starting from each vertex in the graph
making the overall complexity of the algorithm $\bigO(n \Delta^k)$.
Observe carefully that a temporal DFS on each vertex is independent
and it can be parallelized up to $n$-threads. So we use thread-level parallelism
to speedup the computation. Note that the runtime bound of exhaustive search is
loose, in practice the algorithm performs much faster than the theoretical
bounds. A more tighter bound can be obtained with an assumption on degree and timestamp
distribution, for example, $d$-regular or power-law distribution for graph
degrees.

\mpara{Random-walk.}
In this approach we pick a vertex $v$ uniformly at random and perform
a random temporal walk by restricting its length to $k-1$. We check if the
random walk is a path and the vertex colors of the walk matches the multiset.
If it does, we update the solution to the optimal solution provided that the maximum timestamp
in the current solution is less than the optimal solution, otherwise we continue to the next random
walk. The runtime of the algorithm is bounded by the number of random walk
iterations. This approach is also thread parallelizable since each random walk
is independent. 

\smallskip
It is to be noted that, our random-walk implementation failed to report an
optimal solution even for small graph instances with $m=10^4$ and $k=5$ even after
hundred million random iterations. For this reason, we experiment only with the
exhaustive-search baseline.

\smallskip
Our baselines are implemented in the C programming language.
Furthermore, they are optimized for architecture-specific instruction set 
and parallelized to achieve thread-level parallelism.

\subsection{Hardware}
To evaluate our algorithm implementations we experiment with two hardware configurations.

\spara{\em Workstation}:~A Fujitsu Esprimo E920 with 1$\times${3.2}\,GHz Intel Core
i5-4570 CPU, 4 cores, 16\,GB memory, Ubuntu, 
and $\texttt{\small gcc}$\,v\,5.4.0.

\spara{\em Computenode}:~A Dell PowerEdge C4130 with 2$\times${2.5}\,GHz
Intel Xeon\,2680\,V3 CPU, 24 cores, 12 cores/CPU, 128\,GB memory, Red Hat,
and $\texttt{\small gcc}$\,v\,6.3.0.

\smallskip
Our executions make use of all the cores, advanced vector extensions (AVX-2) and
PCMULQDQ instruction set for fast finite field arithmetic.  All the experiments are
executed on the workstation configuration with the only exception for the
experiments with {\it scaling to large graphs}, which are executed on the computenode.

\subsection{Input graphs}
\label{sec:inputgraphs}
We evaluate our methods using both synthetic and real-world graphs.

\mpara{Synthetic graphs.}
We use two types of synthetic graphs: (a) random $d$-regular graphs and (b)
power-law graphs. The regular graphs are generated using the \emph{configuration
model}~\cite[\S\,2.4]{Bollobas01}. The configuration model for power-law graphs
is as follows: given non-negative integers $D$, $n$, $w$, and $\alpha < 0$, we
generate an $n$-vertex graph such the following properties roughly hold: ($i$) the sum of
vertex degrees is~$Dn$; ($ii$) the distribution of degrees is supported at $w$
distinct values with geometric spacing; and ($iii$) the frequency of vertices
with degree $d$ is proportional to~$d^\alpha$. The edge timestamps are assigned
uniformly at random in range~$\{1,\dots,t\}$.
Both directed and undirected graphs are generated using the same
configuration model, however, for directed graphs the orientation is preserved.
We ensure that the graph generator produces identical graph
instances in all the hardware configurations.
Our graph generator is available as open source~\cite{conf-code}.

\mpara{Real-world graphs.}
For the real-world graphs, we make use of temporal graphs from Koblenz
Network Collection \cite{koblenz}, SNAP \cite{snapnets}, and transport networks
of Helsinki and Madrid \cite{conf-code}.

\smallskip
{\em Transport.} We use the bus, interurban bus, metro, train, tram
networks of Madrid and bus network of Helsinki.
In these datasets, each row is a temporal edge between two locations, i.e., an
unique identifier describing origin and destination, starting time and duration
of travel; origin and destination locations.

\smallskip
{\em Koblenz.} We use 
Chess \href{http://konect.uni-koblenz.de/networks/chess}{({\tt chess})}, 
DNC emails \href{http://konect.uni-koblenz.de/networks/dnc-temporalGraph}{({\tt dnc-emails})},
Elections \href{http://konect.uni-koblenz.de/networks/elec}{({\tt wikipedia-elections})}, and
Epinions \href{http://konect.uni-koblenz.de/networks/epinions}{({\tt epinions-trust})}
temporal graphs from Koblenz Network Collection.

\smallskip
{\em SNAP.} We make use of
Bitcoin alpha \href{http://snap.stanford.edu/data/soc-sign-bitcoin-otc.html}{({\tt bitcoin-alpha})},
Bitcoin otc \href{http://snap.stanford.edu/data/soc-sign-bitcoin-alpha.html}{({\tt bitcoin-otc})},
College messages \href{http://snap.stanford.edu/data/CollegeMsg.html}{({\tt college-msg})}, and
Email EU core \href{http://snap.stanford.edu/data/email-Eu-core-temporal.html}{({\tt email-eu})}
temporal graphs from SNAP.

\smallskip
For a detailed description of the datasets used from Koblenz~\cite{koblenz}and
SNAP~\cite{snapnets}, we
refer the reader check the webpages of the corresponding dataset collections.
We preprocess the datasets to generate a graph by renaming the location
identifiers (or vertices) in the range from $1$ to the maximum number of
locations (or vertices) available in the dataset. We assign an unique identifier for each
discrete timestamp beginning with $1$ and incrementing the identifier by one
for each next available timestamp. By doing so we are avoiding the
timestamps for which there are no temporal edges, thereby reducing the maximum
timestamp value. If the time values in dataset are unix timestamps, then we
approximate the value to the closest (floor) second before assigning an unique
timestamp identifier. The vertex colors are assigned uniformly at random in
range $\{1,\dots,30\}$ and the multiset colors are chosen uniformly at random in
the range $\{1,\dots,30\}$. 

\section{Experimental evaluation}
\label{sec:experimental-eval}

We will now describe our results and key findings.
We define \emph{decision time} to be the time required to decide 
the existence of one solution, and \emph{extraction time} the time
required to extract such a solution.  As discussed previously, extracting a
solution requires multiple calls to the decision~oracle. Our baseline and
scalability experiments are performed on the \colorfulpath problem. 
Recall that in the \colorfulpath problem the set of vertex colors is equal
to the colors in the query multiset. 
As a result, the trivial preprocessing step of removing vertices
with colors not matching the motif query cannot be employed.

\subsection{Baseline} 
Our first set of experiments compares the \emph{extraction} time to obtain an
optimum solution using our algebraic algorithm and the exhaustive search
baseline. In Table~\ref{table:baseline:1}, we report extraction times
of the algebraic algorithm and the baseline for extracting an optimum solution
for:
($i$) five independent
\emph{$d$-regular} random graphs
with $n=10^2,\dots,10^5$ and fixed values of $d=20$, $t=100$, $k=5$;
($ii$) five independent
\emph{power-law} graphs with $n=10^2\dots,10^5$, $D=20$, $w=100$, $k=5$, $\alpha=-0.5$; and
($iii$) 
same as the previous setting but $\alpha=-1.0$.
Vertex colors are assigned randomly in the range $\{1,\dots,k\}$ and the query
multiset is $\{1,\dots,k\}$.
Each graph has at least ten target occurrences 
agreeing with the query multiset colors with different timestamps chosen uniformly at random.
For the baseline we report the \emph{minimum} time of five independent runs, however,
for the algebraic algorithm we report the \emph{maximum}.
\emph{Speedup} is the ratio of the runtime of the baseline 
by the runtime of the algebraic algorithm.
The experiments are executed on the {\em workstation} configuration using all cores.
All runtimes are shown in seconds.

\begin{table*}[t]
\caption{Comparison of extraction time for baseline and algebraic algorithms.}
\label{table:baseline:1}
\centering
\footnotesize
\setlength{\tabcolsep}{0.1cm}
\begin{tabular}{r r r r r r r r r r}
\toprule
\multirow{2}{*}{\shortstack{No. of edges \\($m$)}} &
\multicolumn{3}{c}{Regular} & \multicolumn{3}{c}{Powlaw $d^{-0.5}$} &
\multicolumn{3}{c}{Powlaw $d^{-1.0}$}\\ \cmidrule{2-10}
& Baseline& Algebraic & Speedup
& Baseline & Algebraic & Speedup
& Baseline & Algebraic & Speedup \\
\midrule
     1\,040 &  0.05\,s &  0.04\,s & 1.2 &   0.05\,s &  0.04\,s &    1.2 &       0.05\,s &  0.04\,s &    1.2\\
    10\,040 &  0.48\,s &  0.12\,s & 4.1 &   1.03\,s &  0.11\,s &    9.4 &      10.82\,s &  0.10\,s & {\bf 103.6}\\
   100\,040 &  5.62\,s &  1.06\,s & 5.3 &  30.38\,s &  1.07\,s &   28.5 & 20\,430.16\,s &  0.92\,s & {\bf 22\,306.1}\\
1\,000\,040 & 74.01\,s & 12.02\,s & 6.2 & 808.24\,s & 11.18\,s & {\bf 72.3} &       -- & 10.03\,s &     --\\
\bottomrule
\end{tabular}
\end{table*}

\begin{figure*}[t]
\centering
\setlength{\tabcolsep}{0.01cm}
\begin{tabular}{c c c}
  \includegraphics[width=0.33\linewidth]{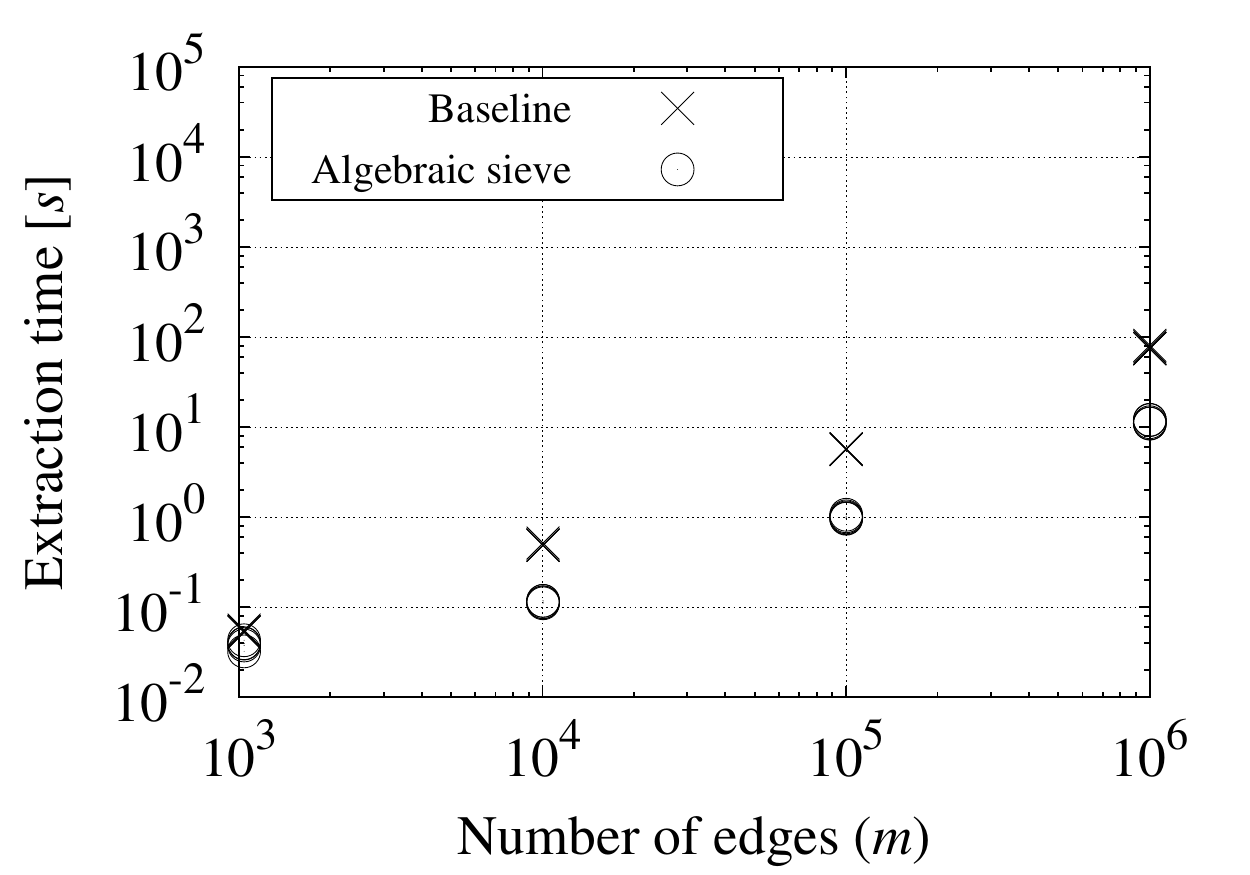} &
  \includegraphics[width=0.33\linewidth]{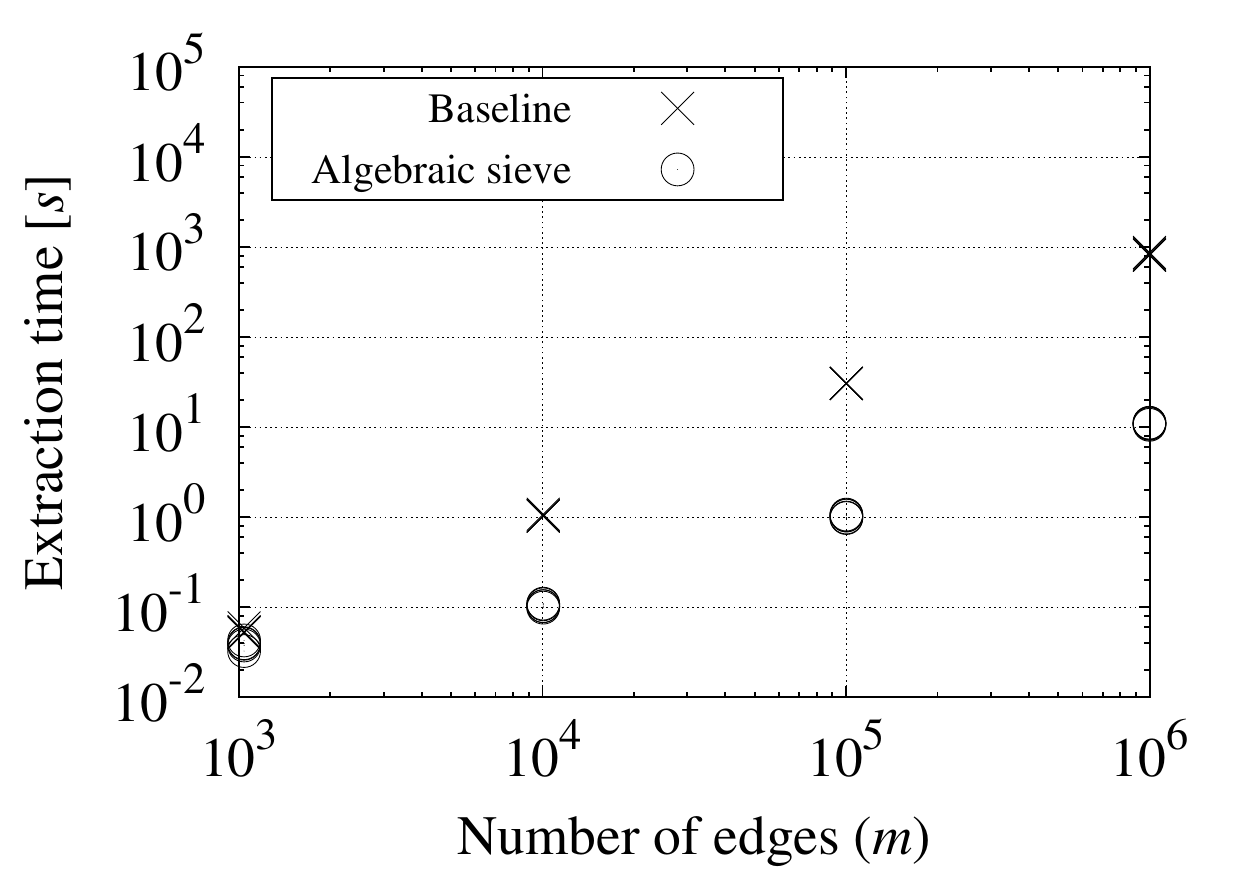} &
  \includegraphics[width=0.33\linewidth]{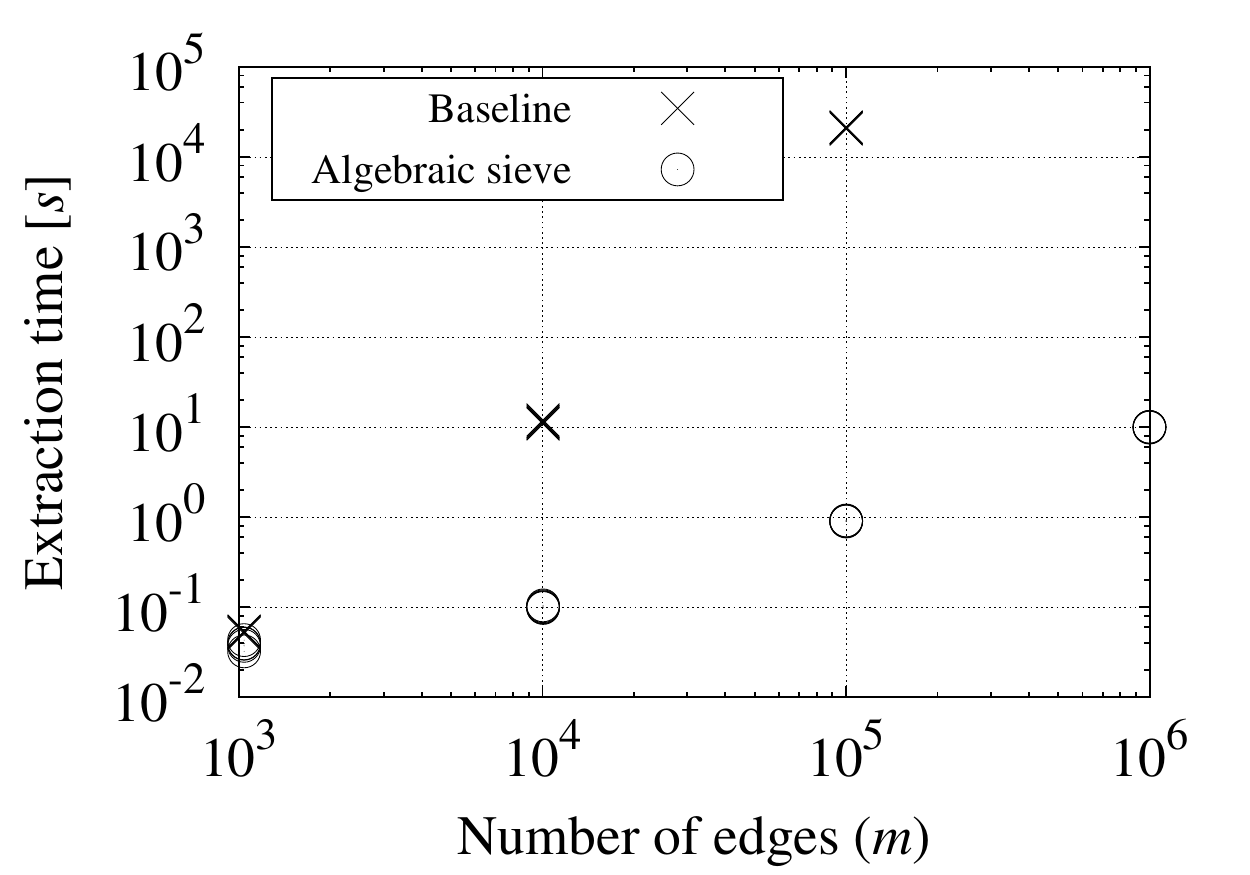}
\end{tabular}
\caption{\label{fig:baseline}%
Comparison of extraction time of baseline and algebraic algorithms. We report
the decision time as a function of the number of edges for $d$-regular random
graph (left), power-law graph with $d=-0.5$ (center), and power-law graph with
$d=-1.0$ (right).%
}
\end{figure*}

\smallskip
Surprisingly, the baseline can compete with the algebraic
algorithm in the case of $d$-regular random graphs, however, the runtimes have high
variance across different graph topologies.
On the other hand, the algebraic algorithm is very stable.
For the power-law graphs
with $m=10^5$ edges and query multiset size $k=5$, the algebraic algorithm is at least
\emph{twenty thousand} times faster than the baseline.
Our exhaustive-search implementation fails to report a solution in small graphs
$m=10^3$ with query multiset size $k=10$.

\smallskip
In Figure~\ref{fig:baseline}, we report the extraction time for baseline and algebraic
algorithms to obtain an
optimum solution for five independent \emph{$d$-regular} random graphs with
$n=10^2,\dots,10^5$ and fixed values of $d=20$, $t=100$, $k=5$ (left); five
independent \emph{power-law} graphs with $n=10^2\dots,10^5$, $D=20$, $w=100$,
$t=100$, $k=5$, $\alpha=-0.5$ (center); and $\alpha=-1.0$ (right). The vertex
colors are assigned uniformly at random in the range $\{1,\dots,k\}$ and the query multiset is
$\{1,\dots,k\}$. Each graph instance has at least ten target instances satisfying the
query multiset colors with different timestamps chosen uniformly at random. The
experiments are executed on the \emph{workstation} configuration.

\smallskip
The runtimes of both the baseline and the algebraic algorithms are consistent with
very little variance across independent graph inputs of the same graph topology.
However, the exhaustive-search baseline has high variance in runtime depending on 
the graph topology. For instance, for a power-law graph of size $m=10^5$ with $\alpha=-1.0$, 
the runtime of the exhaustive-search approach is at least three thousand times greater 
than the runtime on a $d$-regular graph of same size.

\subsection{Scalability}
Our second set of experiments studies scalability with
respect to: ($i$) number of edges, ($ii$) query multiset size, ($iii$) number of
timestamps, and ($iv$) vertex degree.

\begin{figure*}[t]
\centering
\setlength{\tabcolsep}{0.5cm}
\begin{tabular}{c c}
  \includegraphics[width=0.35\linewidth]{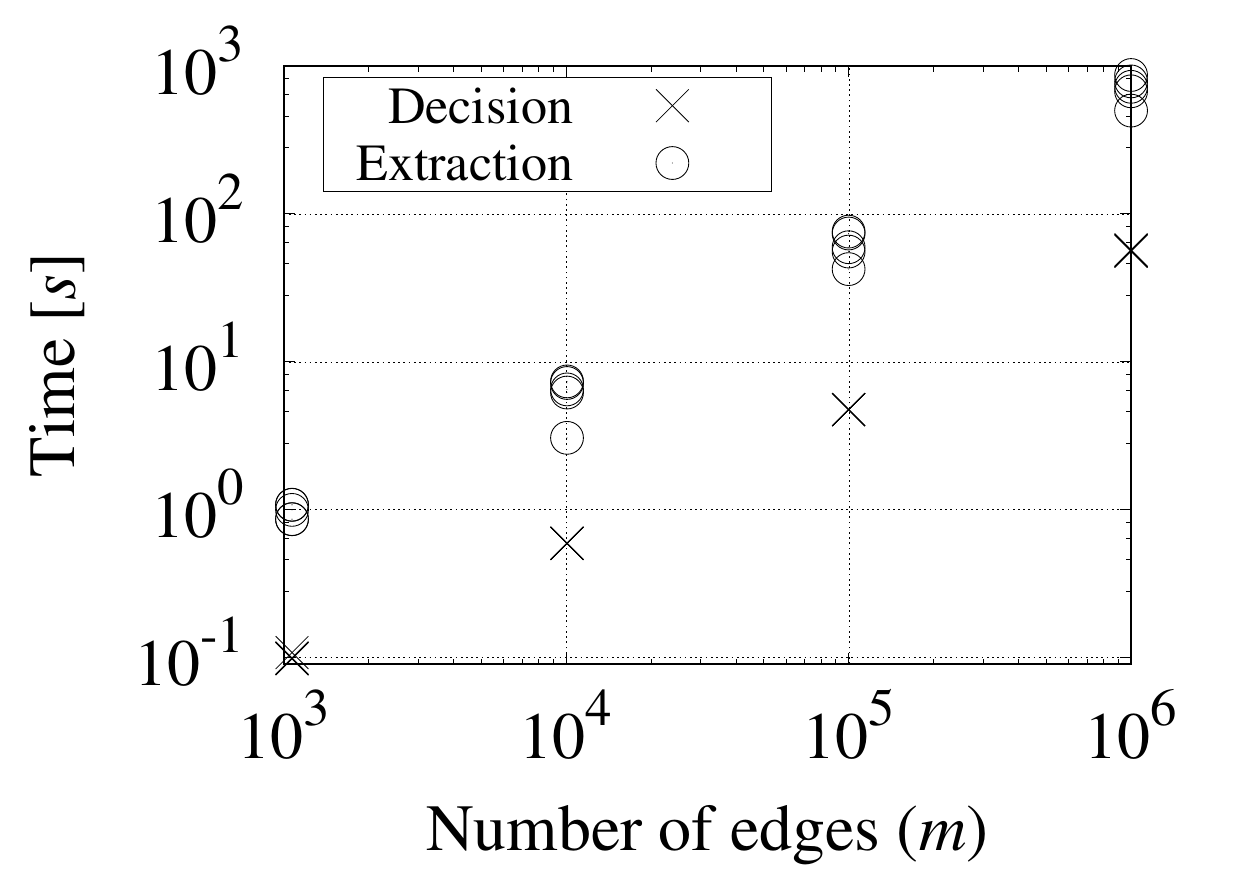}&
  \includegraphics[width=0.35\linewidth]{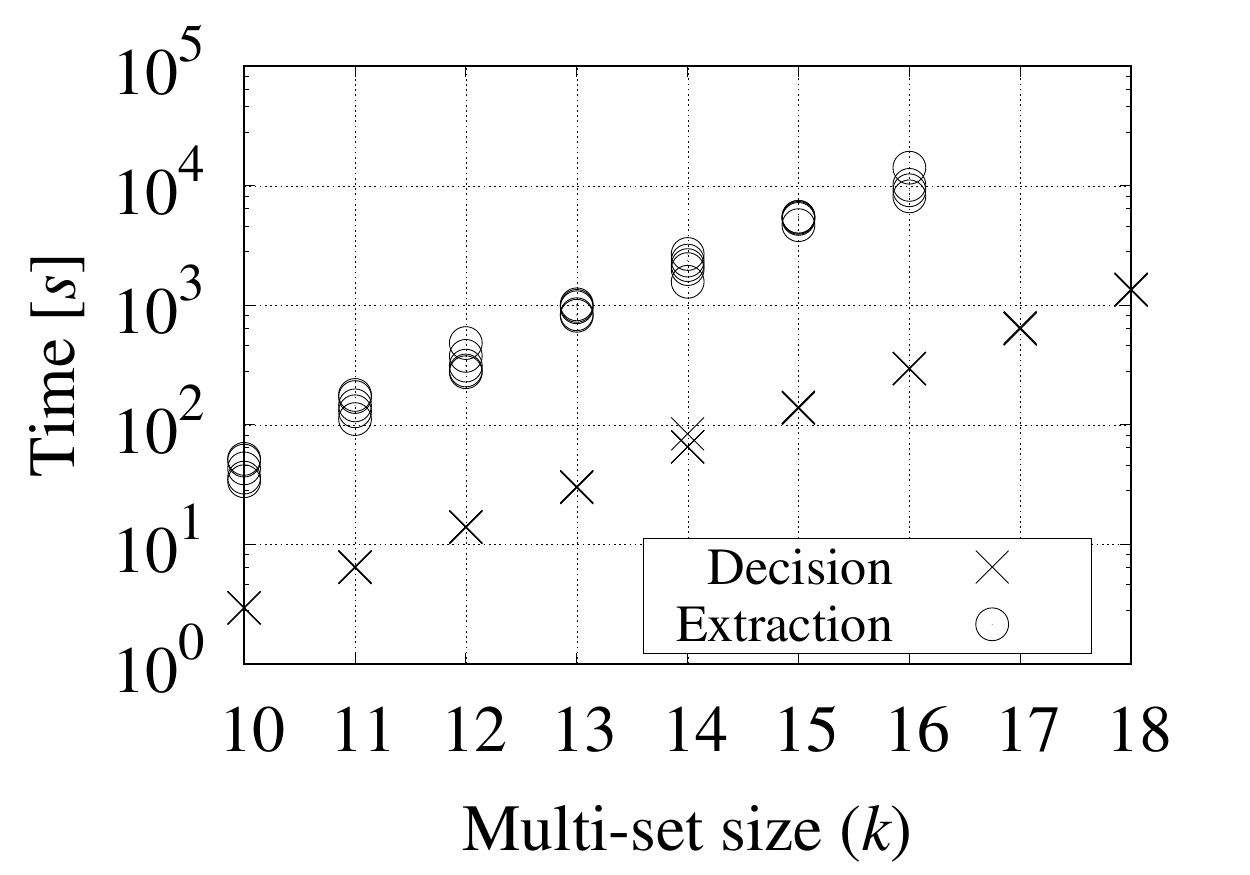}\\
  \includegraphics[width=0.35\linewidth]{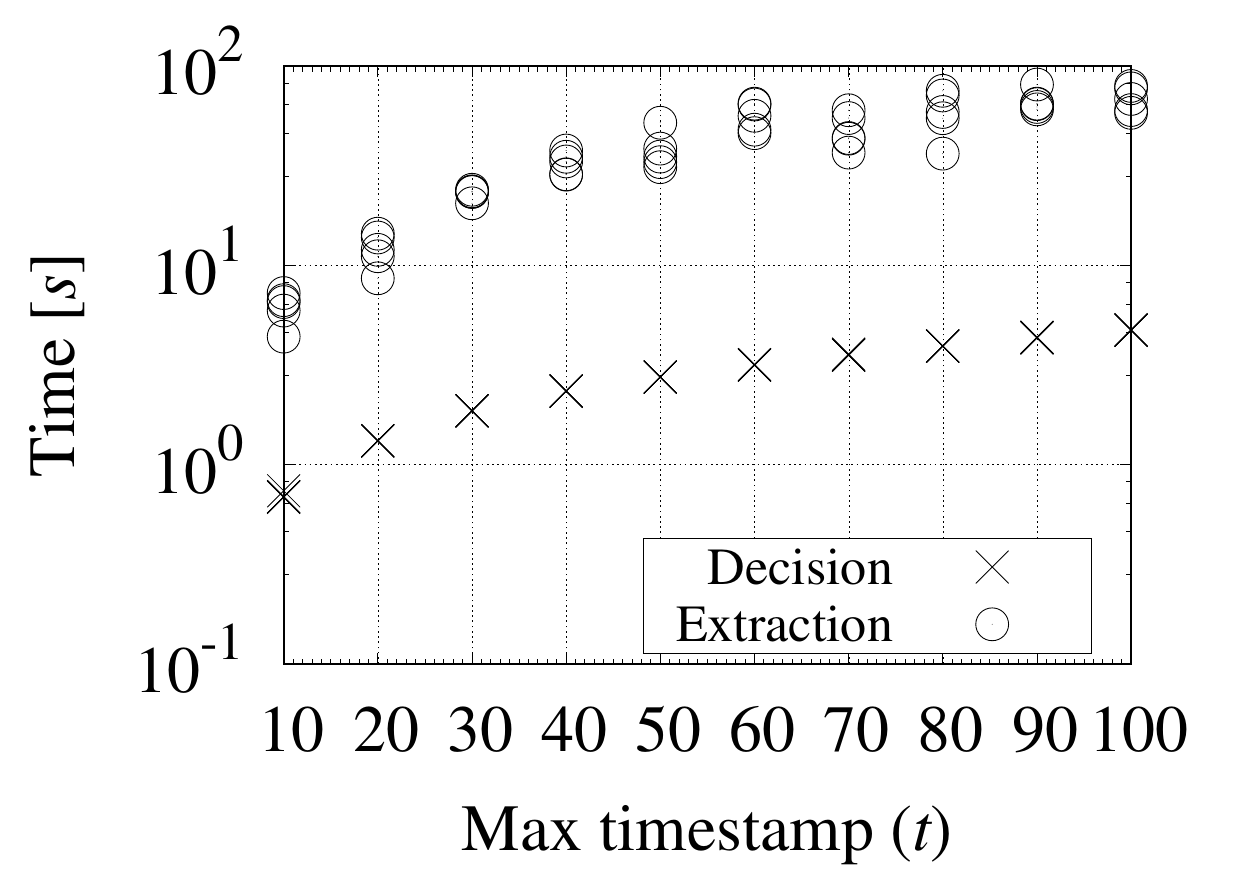}&
  \includegraphics[width=0.35\linewidth]{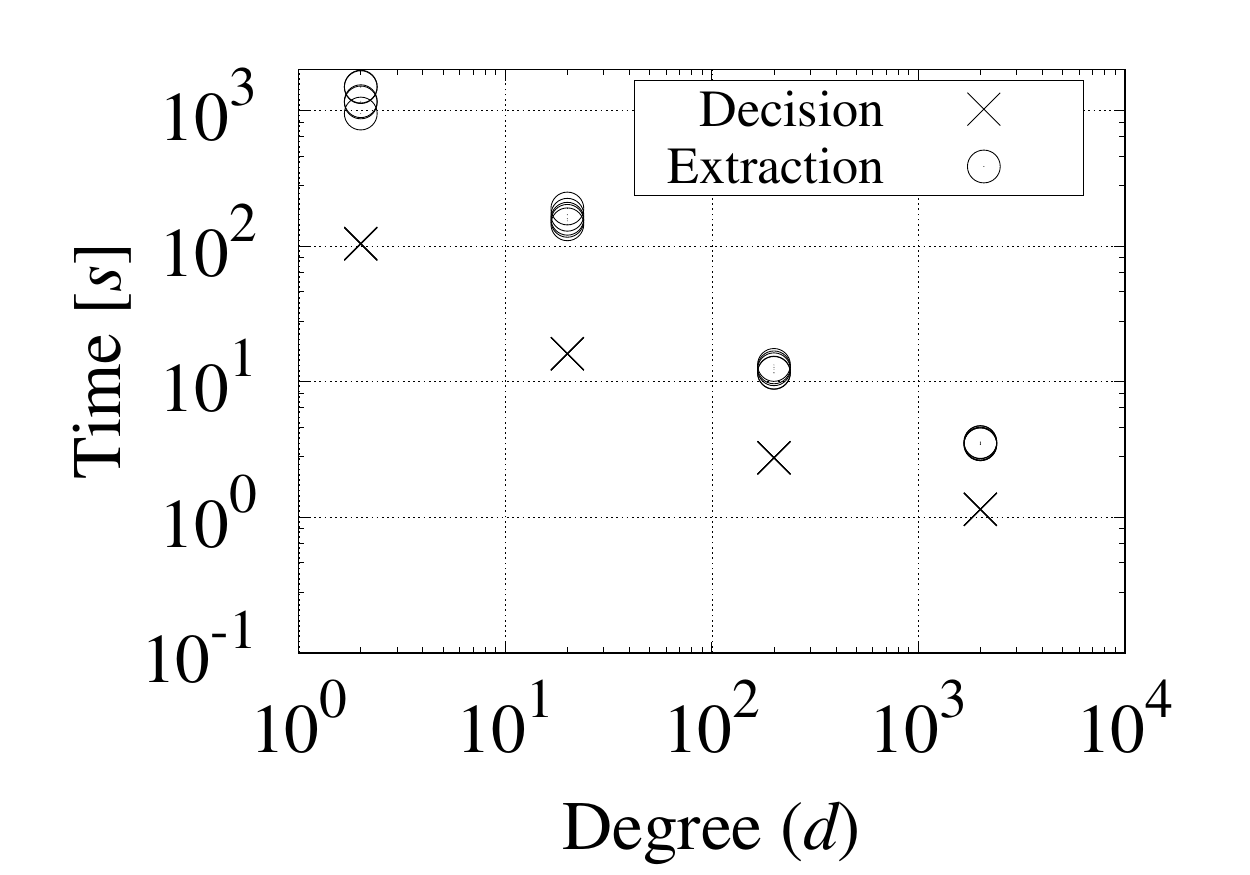}
\end{tabular}
\caption{\label{fig:scalability}%
Scalability results. Runtime as a function of the number of
edges (top-left); query multiset size (top-right); number of timestamps (bottom-left);
and degree (bottom-right).%
}
\end{figure*}

\smallskip
Figure~\ref{fig:scalability}\,(top-left) reports decision and extraction times for
$d$-regular random graphs with $n=10^2,\dots,10^5$ and fixed values of $d=20$, $k=8$, $t=100$.
Figure~\ref{fig:scalability}\,(top-right) shows decision time for
$d$-regular random graphs with $k=10,\dots,18$ and fixed values of $n=10^3$, $d=20$, $t=100$.
Vertex colors are assigned randomly in the range $\{1,\dots,k\}$ and the query
multiset is $\{1,\dots,k\}$.
We observe linear scaling with increasing the number of edges and
exponential scaling with increasing the query multiset size,
as expected by the theory.
The variance in decision time is very small for different inputs, however,
it is higher for extraction time.
The algorithm is able to decide the existence of a solution in less than two minutes for graphs up to
one million edges with query multiset size $k=8$ and extract a solution in less than
sixteen minutes.

\smallskip
%
Next we study the effect of graph density on scalability.
Figure~\ref{fig:scalability}\,(bottom-left) shows decision and extraction
times for 
$d$-regular random graphs with $t=10,\dots,100$ and
fixed values of $n=10^4$, $d=20$, $k=8$.
Figure~\ref{fig:scalability}\,(bottom-right) shows decision and extraction times for
$d$-regular random graphs with $d=2,20,200,2000$ and
corresponding values of $n=10^6,\dots,10^3$, with fixed $m=10^6$ and $t=100$.
We observe that
the algebraic algorithm performs better for dense graphs.
A possible explanation is that for sparse graphs there is not
enough work to keep both the arithmetic and the memory pipeline busy,
simultaneously.

\smallskip
All experiments are executed on the \emph{workstation} configuration using all
cores with undirected graphs. Additionally, we make sure that each input
instance has at least ten solutions agreeing with the multiset of colors, with different
timestamps chosen uniformly at random. We also verified the correctness of our
implementation with graph instances having an unique solution and no solution.

\subsection{Using thread-level parallelism}

\begin{figure*}[t]
\centering
\setlength{\tabcolsep}{0.5cm}
\begin{tabular}{c c}
  \includegraphics[width=0.35\linewidth]{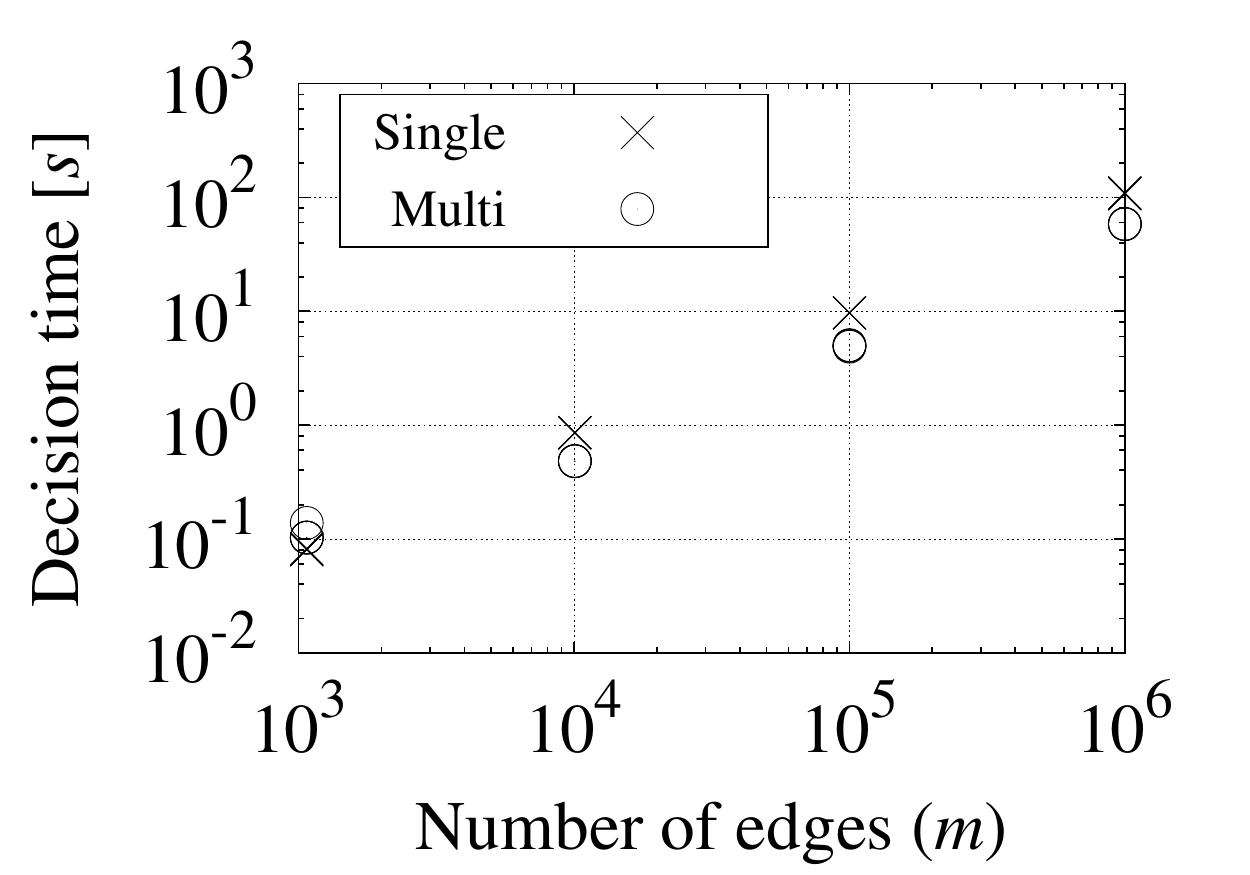}&
  \includegraphics[width=0.35\linewidth]{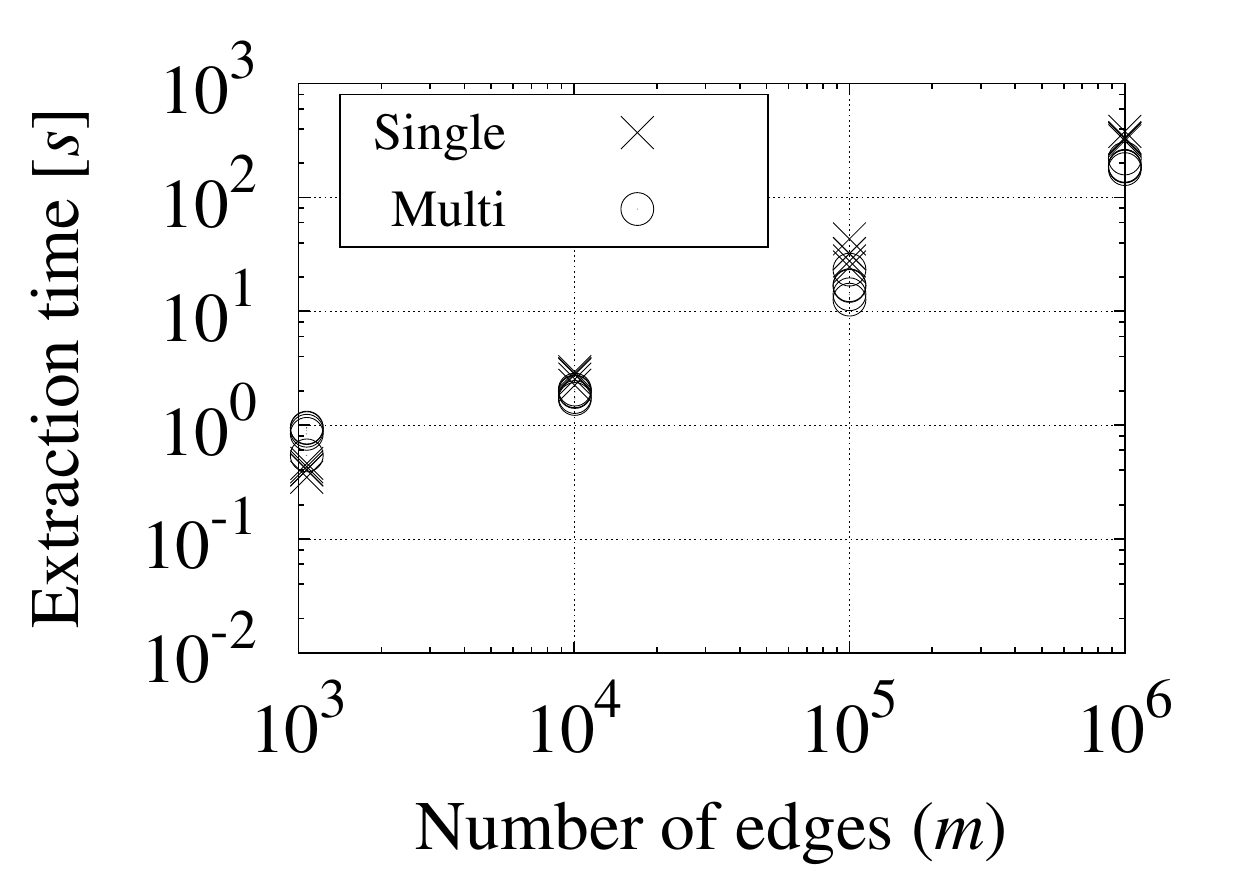}
\end{tabular}
\caption{\label{fig:multithreading:1}%
Using thread-level parallelism. Decision time (left) and extraction time (right) as a function of
number of edges for sparse graphs with degree $d=20$.
}
\end{figure*}

\begin{figure*}[t]
\centering
\setlength{\tabcolsep}{0.5cm}
\begin{tabular}{c c}
  \includegraphics[width=0.35\linewidth]{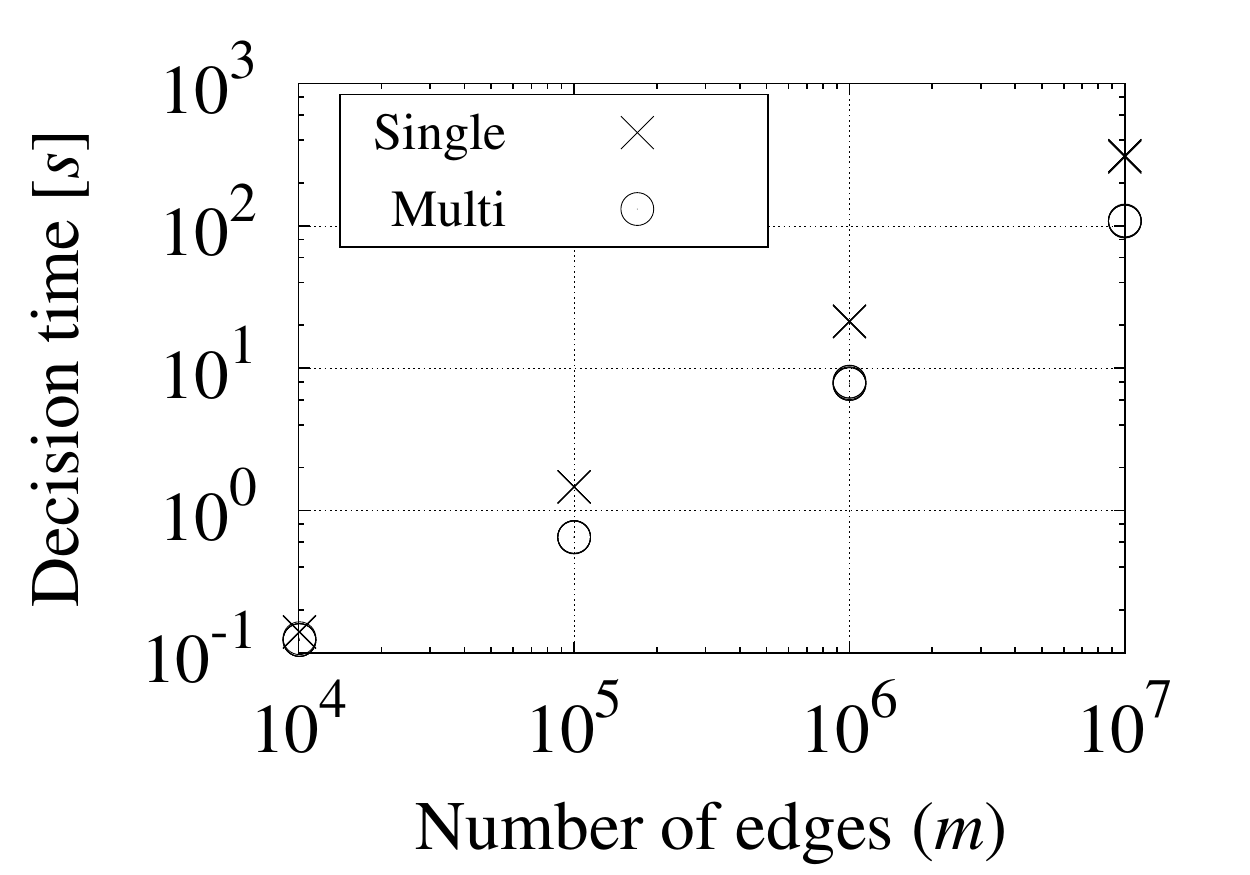}&
  \includegraphics[width=0.35\linewidth]{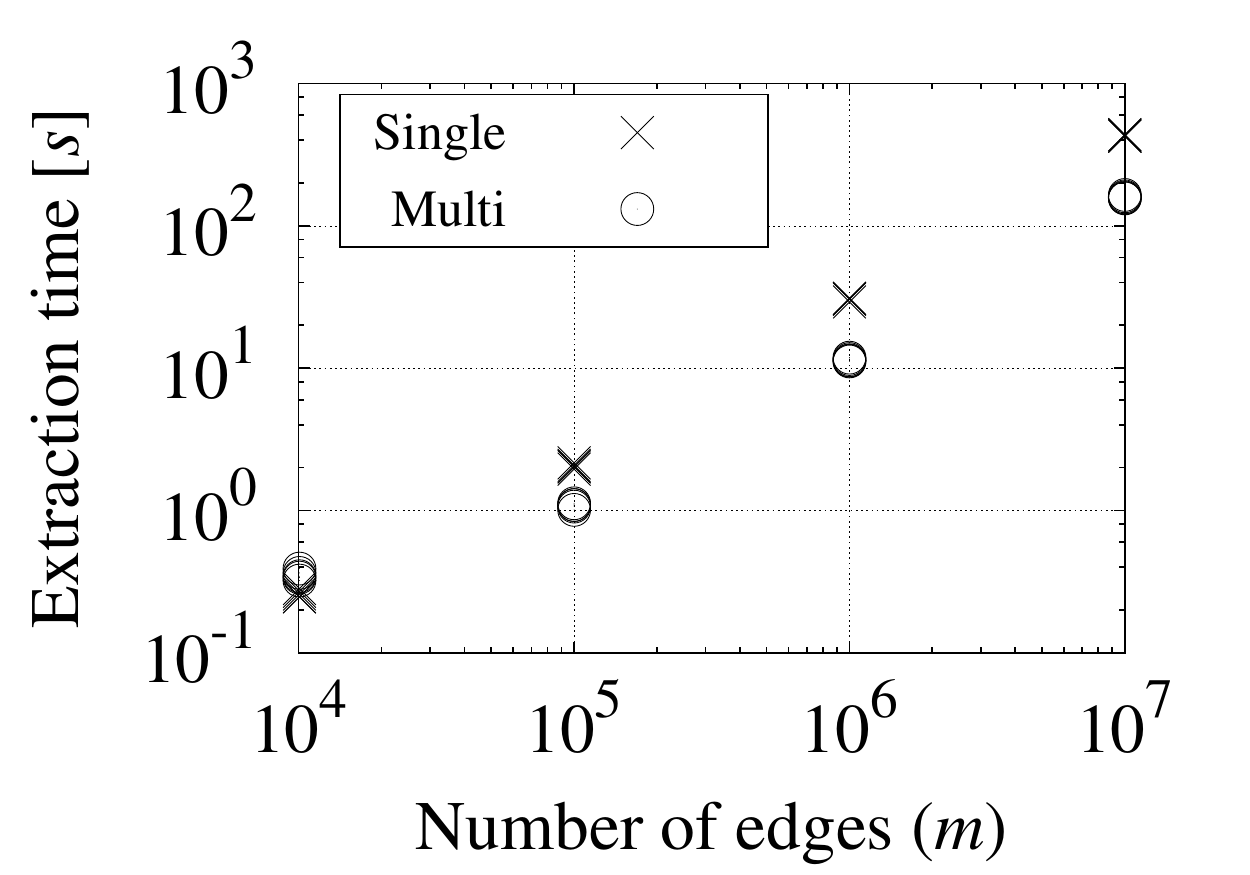}
\end{tabular}
\caption{\label{fig:multithreading:2}%
Using thread-level parallelism. Decision time (left) and extraction time (right) as a function of
number of edges ($m$) for dense graphs with degree $d=200$.
}
\end{figure*}

In our third set of experiments we compare the decision and extraction times 
for single and multi-threaded variants of our implementation. 
We display the decision and extraction times for
five independent $d$-regular random graph instances for each configuration of
$n=10^3,\dots,10^5$, $t=100$, $k=8$ with $d=20$ (sparse graphs) in
Figure~\ref{fig:multithreading:1} and $d=200$ (dense graphs) in
Figure~\ref{fig:multithreading:2}.
Table~\ref{table:multithreading} reports the decision and extraction time along
with the \emph{speedup}, which is the ratio of single and multi-threaded
runtimes.The reported runtime is the \emph{maximum} of five independent
repetitions and all runtimes are shown in seconds.
When scaled to four threads in the \emph{workstation} configuration, we achieved up to
factor 2.7 speedup for the dense graphs. However, for the sparse graphs the
speedup is modest with up to factor 1.9 improvement in the runtime.

\smallskip
Note that
for the dense graphs in our experiments,  
for each vertex there exists at least one edge at
every timestamp with high probability. Conversely, there exists a set of
vertices with no edges at each timestamp. However, in our
current implementation it is not guaranteed that the vertices with no edges are
distributed evenly across the threads while scheduled using the \texttt{OpenMP}
default scheduling. A possible explanation for not achieving a perfect speedup
of factor four is due to the inefficiency in the load-balancing mechanism. 
This presents us with a challenge to achieve perfect load balancing by dynamically
redistributing the set of vertices with no edges uniformly across the threads
at each timestamp. 
Additionally, since the graph instances are sparse, there is not enough
load on the threads to keep the arithmetic pipeline busy.


%

\begin{table}[t]
\caption{\label{table:multithreading} %
Using thread-level parallelism.}
\centering
\footnotesize
\setlength{\tabcolsep}{0.1cm}
\begin{tabular}{r r r r r r r r}
\toprule
\multirow{2}{*}{\shortstack{No. of edges \\($m$)}} &
\multicolumn{3}{c}{Decision} & \multicolumn{3}{c}{Extraction}\\ \cmidrule{2-7}
& Single & Multi & Speedup
& Single & Multi & Speedup \\
\midrule
{\it Sparse graphs}\\
      1\,070 &     0.08\,s &     0.14\,s &    0.6 &     0.46\,s &     0.95\,s &    0.5\\
     10\,070 &     0.86\,s &     0.48\,s &    1.8 &     2.94\,s &     2.05\,s &    1.4\\
    100\,070 &     9.72\,s &     5.03\,s &    1.9 &    43.22\,s &    23.15\,s &    {\bf 1.9}\\
 1\,000\,070 &   109.95\,s &    58.38\,s &    1.9 &   379.31\,s &   219.33\,s &    {\bf 1.7}\\
\midrule
{\it Dense graphs}\\
     10\,070 &     0.14\,s &     0.13\,s &    1.1 &     0.28\,s &     0.39\,s &    0.7\\
    100\,070 &     1.47\,s &     0.65\,s &    2.3 &     2.16\,s &     1.12\,s &    1.9\\
 1\,000\,070 &    21.42\,s &     8.04\,s &    2.7 &    31.02\,s &    11.85\,s & {\bf 2.6}\\
10\,000\,070 &   311.51\,s &   109.04\,s &    2.9 &   437.53\,s &   164.75\,s & {\bf 2.7}\\
\bottomrule
\end{tabular}
\end{table}

\subsection{Memory footprint}
Recall from \S\ref{sec:implementation} that we have presented two variants 
of the generating-function implementation:
($i$) an implementation borrowed from our earlier work \cite{conf-paper}, which uses
$\bigO(ntk)$ memory (\emph{genf-1}); and ($ii$) a memory-efficient
implementation, which uses $\bigO(nt)$ working memory (\emph{genf-2}).

\smallskip
Our next set of experiments demonstrates 
that the generating-function implementation \emph{genf-2} 
is more memory efficient than \emph{genf-1}, 
without significant change in the runtime.
%
%
Figure~\ref{fig:memory} displays the extraction
time (left) and peak-memory usage (right) for $d$-regular random graphs 
with $n=10^2,\dots,10^5$ with fixed values of $t=100$, $k=8$,  and $d=20$. 
We observe no significant change in the extraction time between \emph{genf-1}
and \emph{genf-2}. However, the reduction in working memory is
significant, for example, in large graphs with $m=10^6$ and $k=8$, \emph{genf-1} uses
at least three times as much memory as \emph{genf-2}. Note that the reported
peak-memory also includes the memory used to store the input graph, which occupy
a significant portion of the working memory. The experiments are executed on the
\emph{workstation} configuration using all cores.

\begin{figure*}[t]
\centering
\setlength{\tabcolsep}{0.5cm}
\begin{tabular}{c c}
  \includegraphics[width=0.35\linewidth]{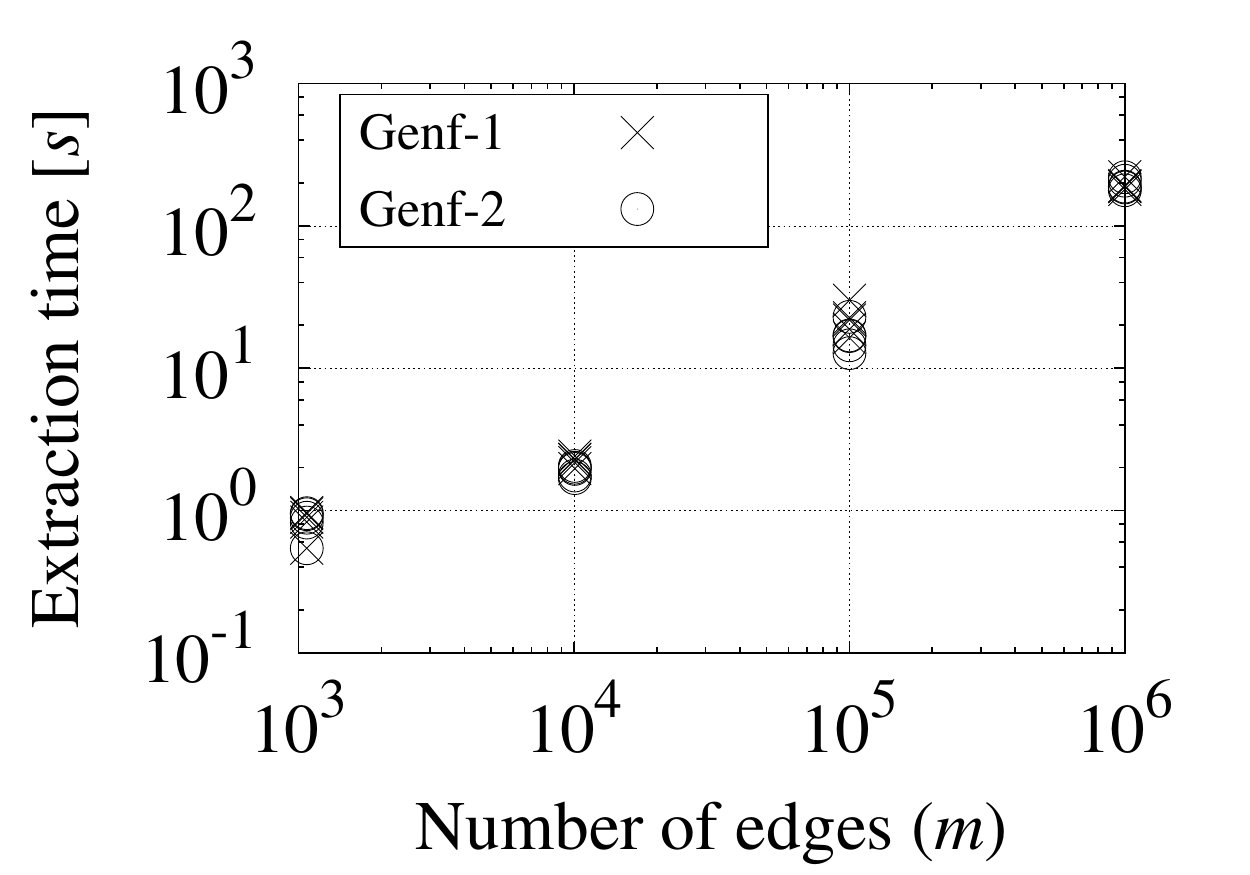}&
  \includegraphics[width=0.35\linewidth]{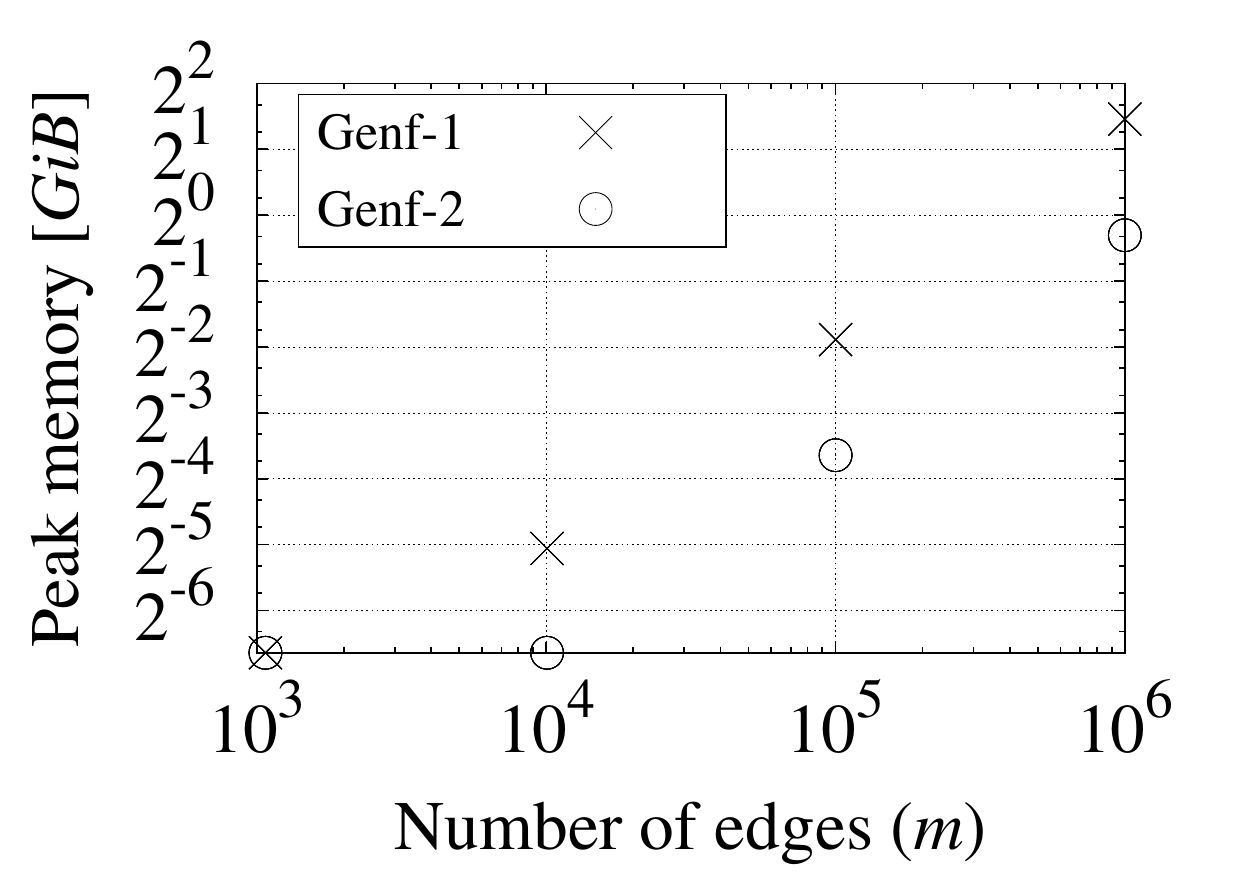}
\end{tabular}
\caption{\label{fig:memory}%
Comparing the memory footprint. Extraction time (left) and peak-memory usage (right)
as a function of the number of edges for \emph{genf-1} and \emph{genf-2}.
}
\end{figure*}

\subsection{Preprocessing and vertex-localization}
Next we demonstrate the effectiveness of
preprocessing and vertex-localization for improving the performance of the
algorithm.

\begin{table}[t]
\caption{\label{table:preproc}%
Preprocessing and vertex-localization.}
\centering
\footnotesize
\setlength{\tabcolsep}{0.1cm}
\begin{tabular}{r r r r r r r r}
\toprule
\shortstack{No. of edges \\($m$)} &
\shortstack{Algebraic \\~} &
\shortstack{Algebraic \\(pre)} &
\shortstack{Speedup-1\\~} &
\shortstack{Algebraic \\(vloc)} &
\shortstack{Speedup-2\\~} &
\shortstack{Algebraic \\(pre + vloc)} &
\shortstack{Speedup-3\\~}\\
\midrule
{\pathmotif}\\
       1\,070 &     1.11\,s &     1.06\,s &    1.0 &     0.27\,s &    4.1 &     0.25\,s &    4.4\\
      10\,070 &     5.75\,s &     3.11\,s &    1.8 &     1.25\,s &    4.6 &     0.56\,s & {\bf 10.2}\\
     100\,070 &    47.86\,s &    11.73\,s &    4.1 &    13.95\,s &    3.4 &     3.36\,s & {\bf 14.2}\\
  1\,000\,070 &   493.65\,s &   111.96\,s &    4.4 &   163.64\,s &    3.0 &    41.57\,s & {\bf 11.9}\\
\midrule
{\colorfulpath}\\
       1\,070 &     0.87\,s &     0.90\,s &    1.0 &     0.26\,s &    3.3 &     0.26\,s &    3.4\\
      10\,070 &     1.96\,s &     2.03\,s &    1.0 &     0.58\,s &    3.4 &     0.55\,s &    3.6\\
     100\,070 &    16.60\,s &    16.90\,s &    1.0 &     7.56\,s &    2.2 &     7.98\,s &    2.1\\
  1\,000\,070 &   195.99\,s &   196.79\,s &    1.0 &    86.66\,s &    2.3 &    90.38\,s &    2.2\\
\bottomrule
\end{tabular}
\end{table}

\begin{figure}[t]
\centering
\setlength{\tabcolsep}{0.5cm}
\begin{tabular}{c c}
  \includegraphics[width=0.35\linewidth]{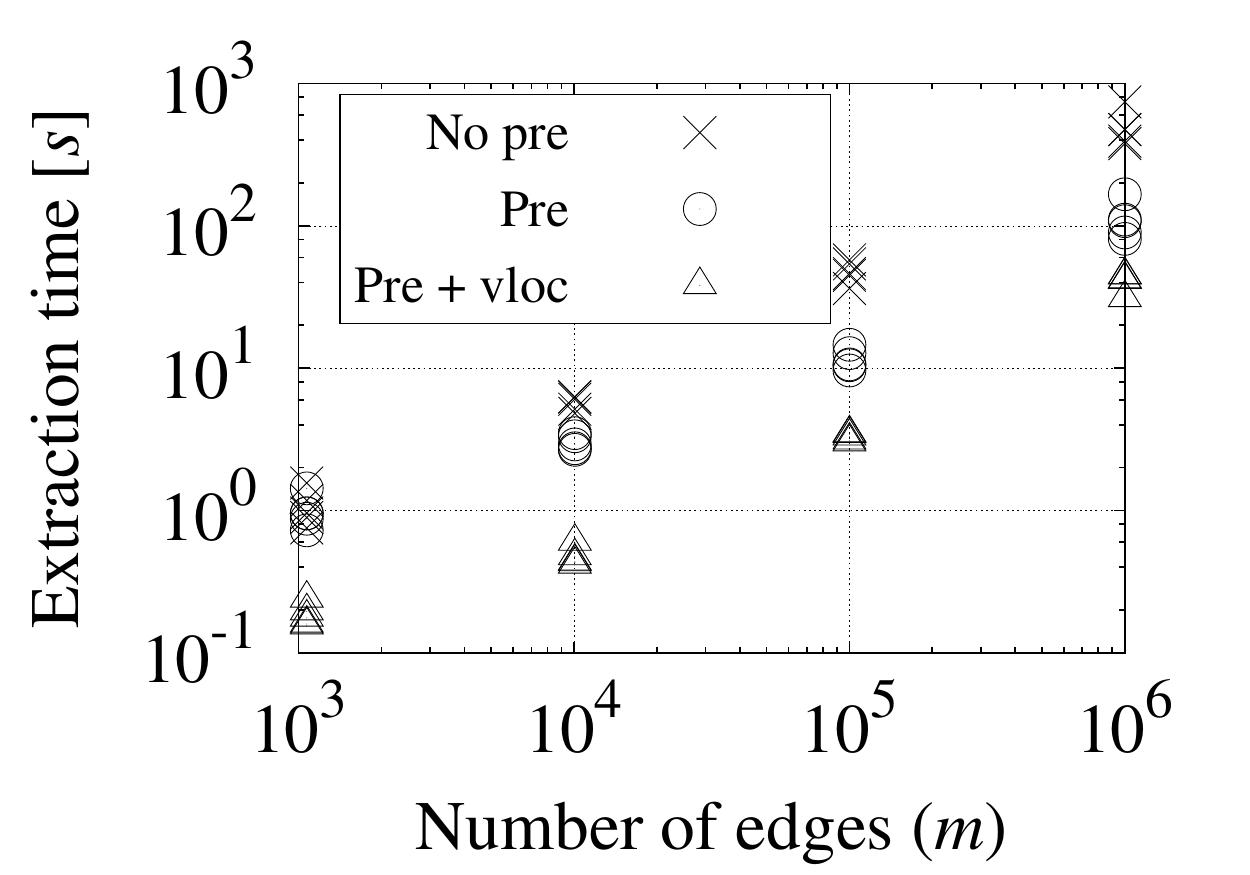}
  &
  \includegraphics[width=0.35\linewidth]{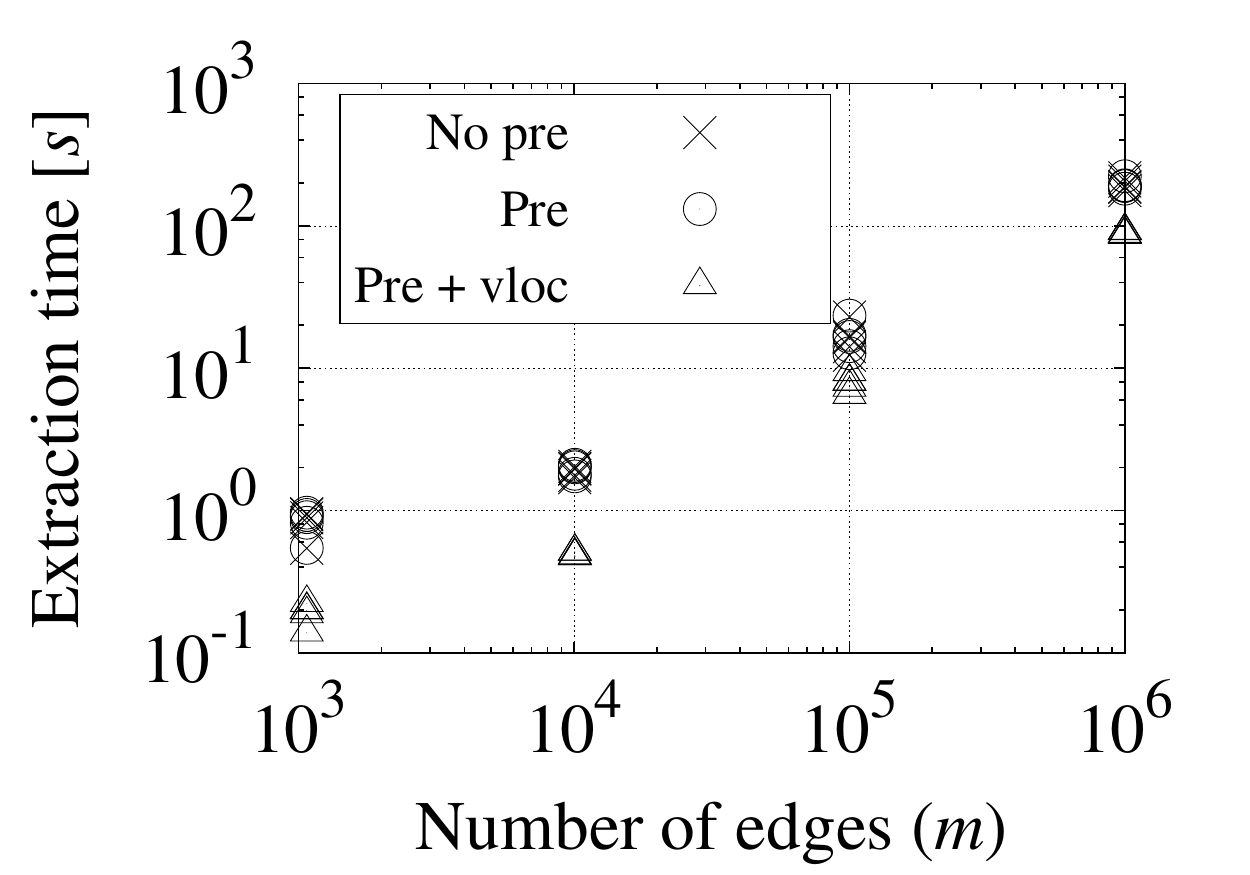}
\end{tabular}
\caption{\label{fig:preproc}%
Preprocessing and vertex-localization. Extraction time as a function of the number
of edges $m$ for the \pathmotif problem (left) and the \colorfulpath problem (right) problem.
}
\end{figure}

Recall that we have implemented two preprocessing techniques
[cf.~\S\ref{sec:implementation:preproc}]: ($i$) remove vertices
with colors not matching the query multiset colors and edges incident to them; and 
($ii$) remove vertices (and edges incident to them), which are not incident to a
match in the corresponding non-temporal instance. In the
following experiments we only make use of the second preprocessing~step.

\smallskip
In Table~\ref{table:preproc} we report the extraction time: ($i$) without
preprocessing (algebraic), ($ii$) with preprocessing (\emph{pre}), ($iii$)
with vertex-localization (\emph{vloc}), and ($iv$) with preprocessing and vertex
localization (\emph{pre + vloc}).
The experiments are performed on five independent instances of $d$-regular
random graphs for each configuration of
$n=10^2,\dots,10^5$ and fixed values of $d=20$, $t=100$, $k=8$. 
In the \pathmotif problem the vertex colors are chosen uniformly in range
$\{1,\ldots,30\}$ and the
query multiset is chosen randomly. In the \colorfulpath problem the vertices are
colored uniformly in range $\{1,\dots,k\}$ and the query multiset is
$\{1,\dots,,k\}$. The reported
runtimes are the average of five independent executions and all runtimes are shown in
seconds. 
Additionally, we display the extraction time as a function of the
number of edges $m$ for the \pathmotif problem (left) and the \colorfulpath problem
(right) in Figure~\ref{fig:preproc}.

\smallskip
We observe that preprocessing and vertex-localization are very effective in the
\pathmotif instances compared to the \colorfulpath instances. In
\pathmotif instances we obtain up to factor-fourteen speedup in extraction time for
large graphs. However, for the \colorfulpath instances the speedup is rather
modest with up to factor $3.6$ improvement in runtime. We also observe high variance 
in the extraction time for the \pathmotif
problem as compared to the \colorfulpath problem.
The experiments are performed on the \emph{workstation} configuration using all cores.

\smallskip
Note that, the execution time of our algorithm vary depending on the reduction
in the graph size obtained after preprocessing. The high variance in the runtime is
a consequence of the variation in the graph size after preprocessing, as
observed in Figure~\ref{fig:preproc}.
Additionally, we observed a significant reduction in the graph size after preprocessing in
\pathmotif instances compared to \colorfulpath instances. As a consequence, we obtain
better speedup in computation for \pathmotif instances.

\subsection{Scaling to large graphs}
Next we demonstrate the scalability of the algebraic algorithm to graphs with up to
one billion edges. 

\smallskip
Figure~\ref{fig:pushing-limits:1} shows extraction time
(left) and peak-memory usage (right) for
$d$-regular random graphs with $n=10^3,\dots,10^7$, $d=200$, $t=100$ with $k=5$.
In graphs with one billion edges, our algebraic algorithm using
preprocessing and vertex-localization can extract an optimum solution in less
than thirteen minutes for small query (multiset) size with $k=5$, while making use of
less than one-hundred gigabytes of memory. It is important to note that
more than half of the working memory is used for processing the input graph.

\begin{figure}[t]
\centering
\setlength{\tabcolsep}{0.5cm}
\begin{tabular}{c c}
  \includegraphics[width=0.35\linewidth]{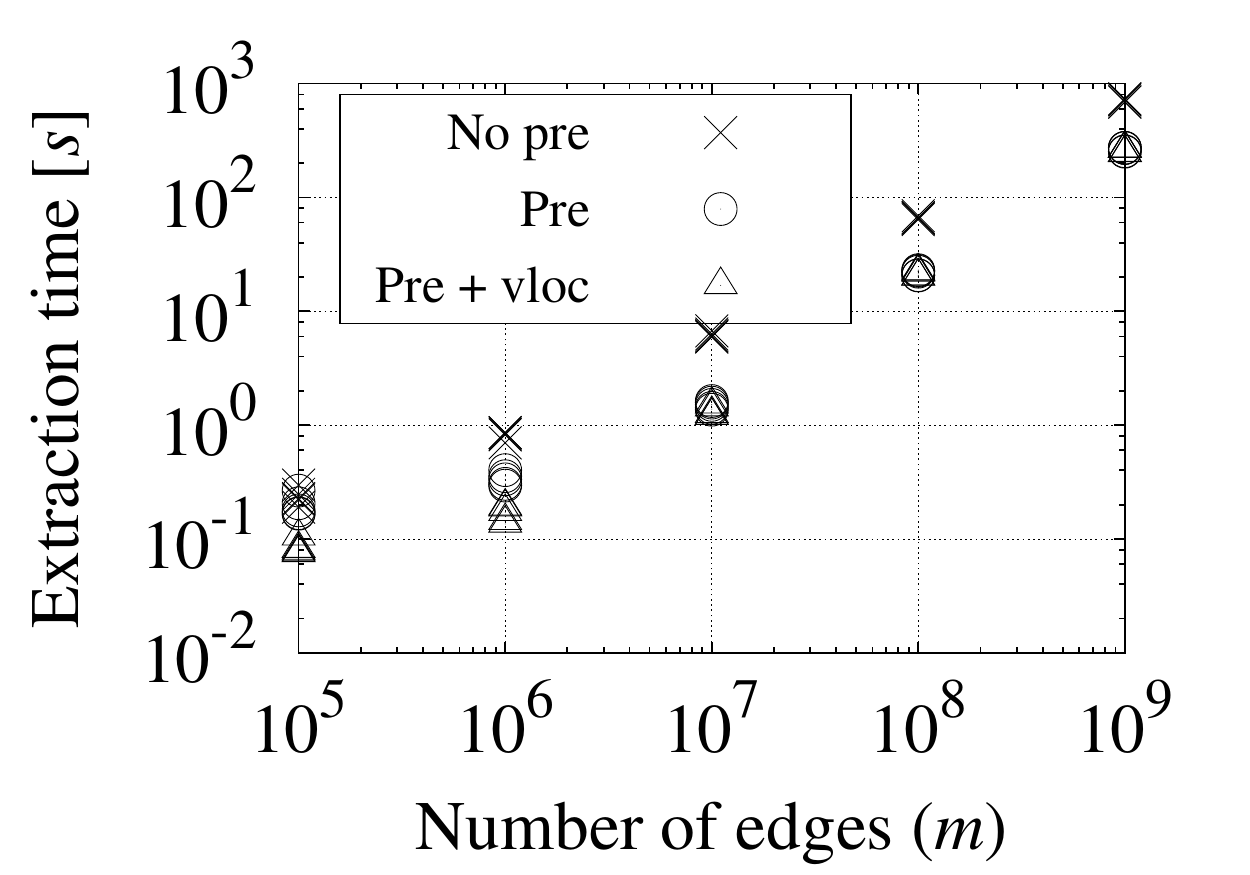} &
  \includegraphics[width=0.35\linewidth]{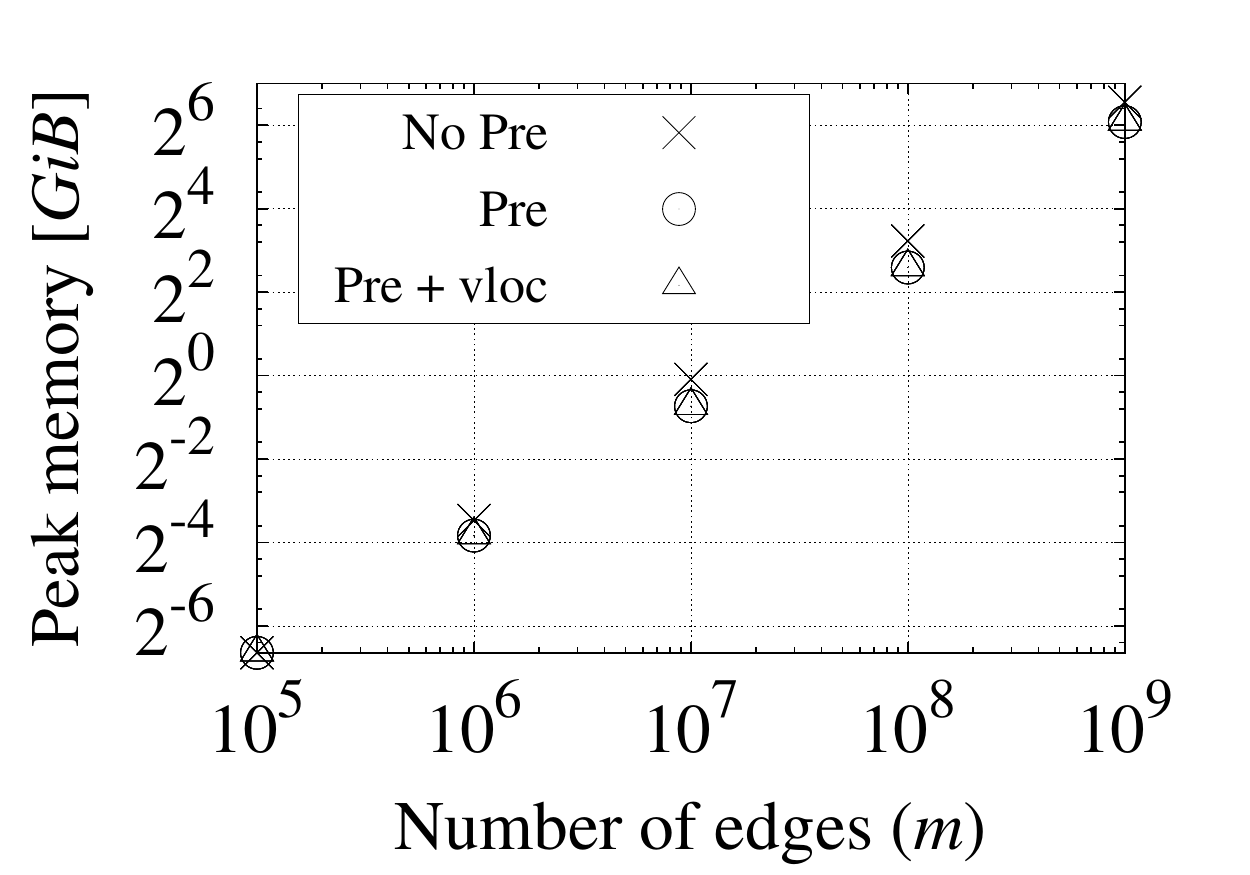}
\end{tabular}
\caption{\label{fig:pushing-limits:1}%
Scaling to a billion edges with small query multiset size. Extraction time (left) and peak-memory usage (right)
as a function of the number of edges with query (multiset) size $k=5$.}
\end{figure}

\smallskip
Our next set of experiments studies the scaling of the algorithm for graphs up to
hundred million edges for large query multiset size with $k=10$.
Figure~\ref{fig:pushing-limits:2} shows the extraction time (left) and
peak-memory usage (right) of the algorithm for five independent $d$-regular
random graphs with $n=10^3,\dots,10^6$ and $d=200$, $k=10$, $t=200$, fixed. In
graphs with one hundred million edges and query (multiset) size $k=10$, our
implementation using preprocessing and vertex-localization can extract a
solution in less than thirty five minutes while
using of at most ten gigabytes of working memory.

\smallskip
The experiments are executed on \pathmotif instances. The vertex colors are
assigned uniformly at random in range $\{1,\dots,30\}$ and
the query multiset is chosen uniformly at random. We ensure that each graph instance has at
least ten target instances agreeing with the query multiset colors.
All experiments are performed on the \emph{computenode} configuration using all cores with
undirected graphs and we employ the second preprocessing step
(removing vertices that are not incident to a match in the corresponding non-temporal instance).

\begin{figure}[t]
\centering
\setlength{\tabcolsep}{0.5cm}
\begin{tabular}{c c}
  \includegraphics[width=0.35\linewidth]{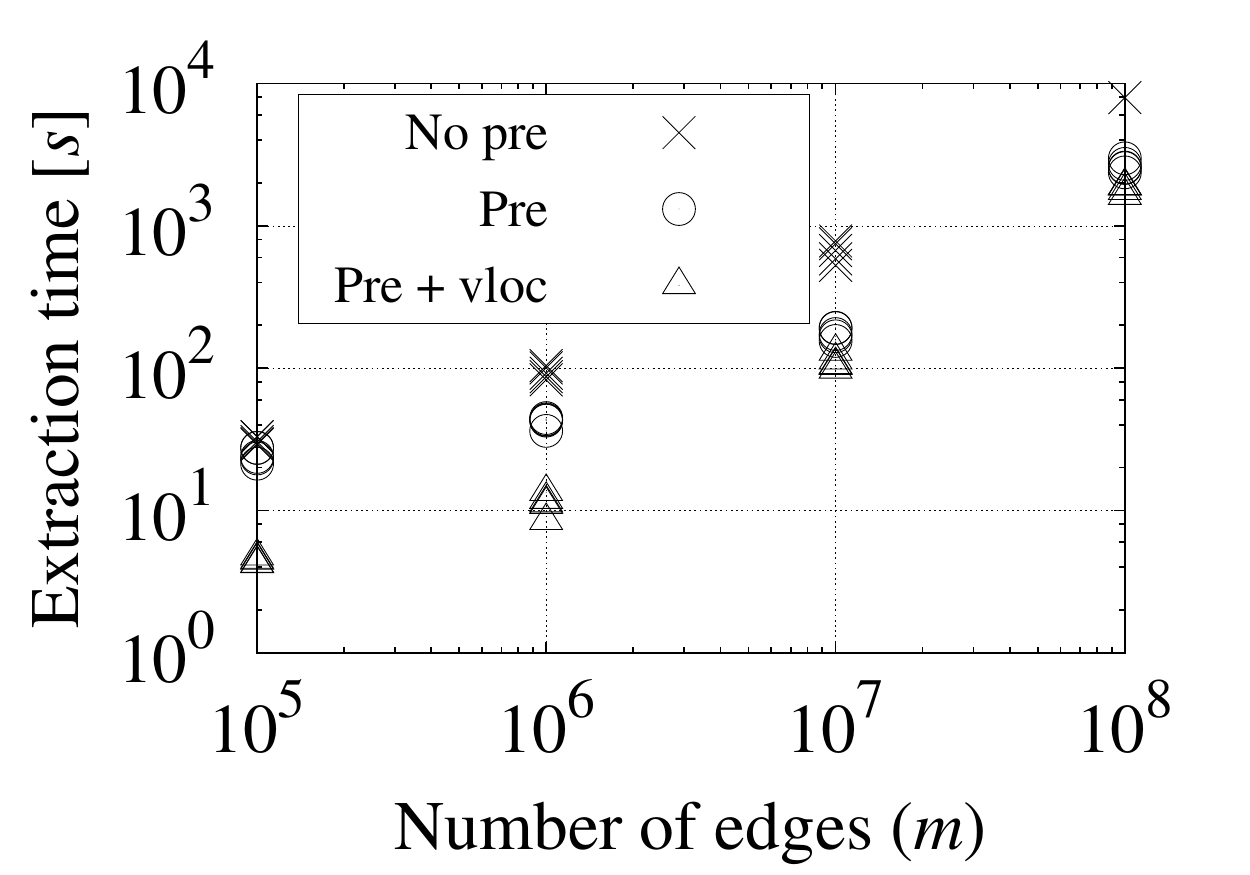} &
  \includegraphics[width=0.35\linewidth]{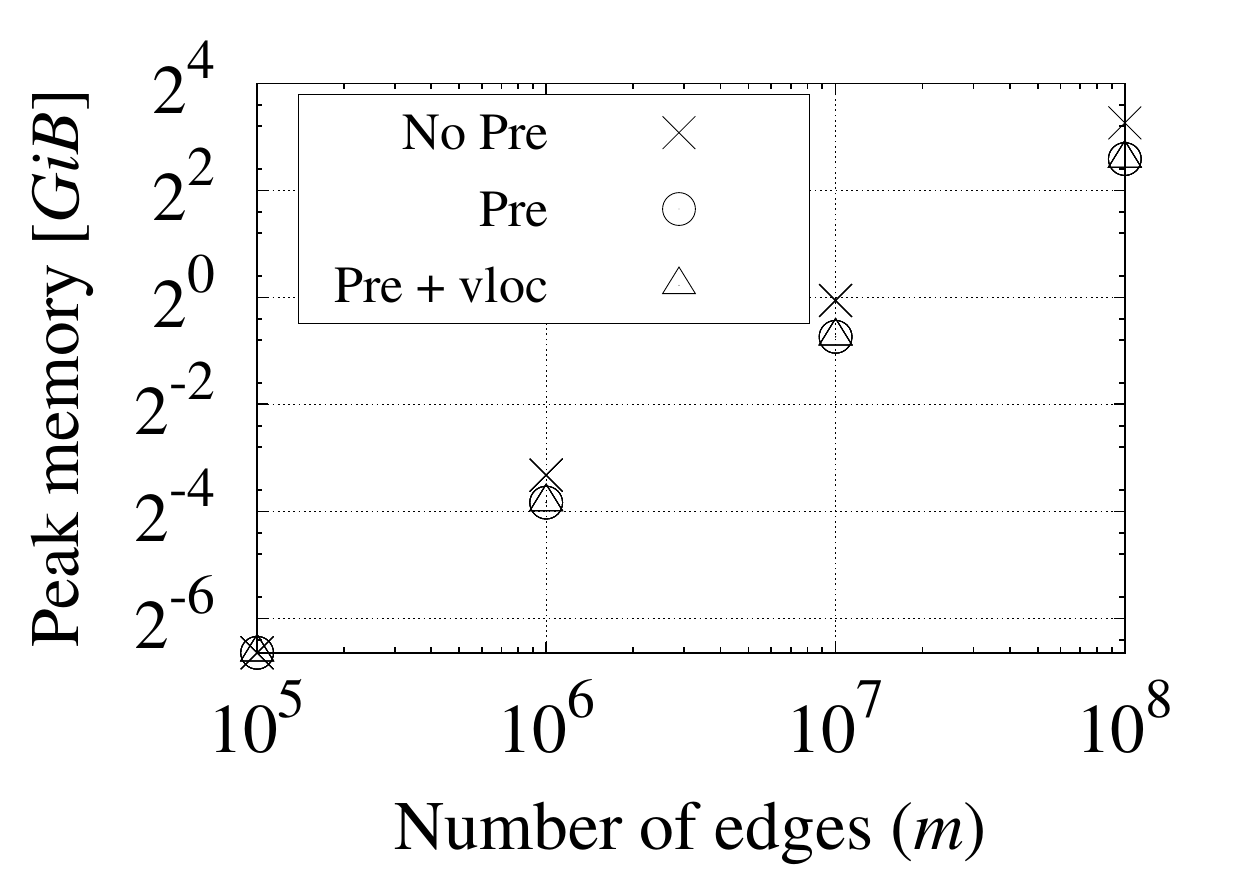}
\end{tabular}
\caption{\label{fig:pushing-limits:2}%
Scaling to large graphs with large query multiset size. Extraction time (left) and peak-memory usage (right)
as a function of the number of edges with query (multiset) size $k=10$.
}
\end{figure}

\subsection{Experiments with real-world graphs}

Finally, we evaluate our algebraic algorithm on real-world data, 
using the datasets described in Section~\ref{sec:inputgraphs}. 
For this set of experiments we focus on extraction time. 
Table~\ref{table:dataset:1} compares the extraction time (shown in seconds) 
for the baseline and algebraic algorithms  on the transport datasets. 
For each dataset we report the \emph{maximum} time for the algebraic algorithm and 
the \emph{minimum} time for the baseline algorithm, over five independent
executions by choosing the query multiset of colors at random. 
For query multiset size $k=5$, the extraction time is at most thirteen seconds.  
For larger query multiset size $k=10$, the extraction time is at most eight minutes 
in all the datasets.  
For the baseline algorithm we preprocess the graphs by
removing vertices whose colors do not match with the query multiset colors. The
reported running time for the algebraic algorithm include the second preprocessing
step with vertex-localization.
Note that for the baseline we timeout the experiments which consume
more than two hours.

\smallskip
Table~\ref{table:dataset:2} reports the extraction time (shown in seconds) of
the algebraic algorithm for the experiments on real-world datasets. For each
dataset we report the \emph{maximum} time among five independent executions by 
choosing the multiset colors at random with multiset query size $k=5$. 
Speedup is the ratio of the runtime of the algebraic algorithm with and without vertex-localization. 
Memory is the peak-memory usage in gigabytes. 
The algorithm can extract a solution in less than ninety seconds in all the datasets 
using less than twelve gigabytes of working memory.

\smallskip
All experiments are performed on the \emph{workstation} configuration using all cores with
undirected graphs and we employ the second preprocessing step
(removing vertices that are not incident to a match in the corresponding
non-temporal instance). 

\begin{table*}[t]
\caption{Experimental results on transport networks.}
\label{table:dataset:1}
\centering
\footnotesize
\begin{tabular}{r r r r r r r r}
\toprule
 & & & & \multicolumn{2}{c}{$k=5$} & \multicolumn{2}{c}{$k=10$}\\
\cmidrule(lr){5-6} 
\cmidrule(lr){7-8} 
Dataset & $n$ & $m$ & $t$ & Baseline & Algebraic & Baseline & Algebraic \\
\midrule
          Madrid tram &     70 &      35\,144 & 1\,265 &      1.37\,s &  0.12\,s & 1\,337.98\,s &  28.05\,s \\
         Madrid train &     91 &      43\,677 & 1\,181 &     40.01\,s &  0.12\,s & 1\,634.56\,s &  24.12\,s \\
Madrid interurban bus & 7\,543 &  1\,495\,055 & 1\,440 &    744.79\,s &  1.11\,s &           -- & 325.51\,s \\
           Madrid bus & 4\,597 &  2\,254\,993 & 1\,440 & 6\,337.89\,s &  1.40\,s &           -- & 278.91\,s \\
         Helsinki bus & 7\,959 &  6\,403\,785 & 1\,440 &           -- &  3.52\,s &           -- & 444.66\,s \\
         Madrid metro &    467 & 37\,565\,706 & 1\,440 &           -- & 12.87\,s &           -- &  98.69\,s \\
\bottomrule
\end{tabular}
\end{table*}

\begin{table*}[t]
\caption{Experimental results on real-world graphs.}
\label{table:dataset:2}
\centering
\footnotesize
\begin{tabular}{r r r r r r r r}
\toprule
Dataset & $n$ & $m$ & $t$ & No vloc & vloc & Speedup & Memory (GB)\\
\midrule
         Bitcoin alpha &    3\,783 &      24\,190 &     1\,647 &     0.69\,s &     0.36\,s &    1.9 &   0.10 \\
           Madrid tram &        70 &      35\,139 &     1\,265 &     0.20\,s &     0.12\,s &    1.7 &   0.00 \\
           Bitcoin otc &    5\,881 &      35\,596 &    31\,467 &    22.19\,s &    13.27\,s &    1.7 &   2.95 \\
            DNC emails &    1\,891 &      39\,268 &    19\,383 &     4.63\,s &     2.83\,s &    1.6 &   0.58 \\
          Madrid train &        91 &      43\,672 &     1\,181 &     0.19\,s &     0.12\,s &    1.7 &   0.00 \\
           College msg &    1\,899 &      58\,975 &    35\,913 &    12.52\,s &     7.14\,s &    1.8 &   1.11 \\
                 Chess &    7\,301 &      64\,962 &        100 &     0.12\,s &     0.10\,s &    1.1 &   0.01 \\
             Elections &    7\,118 &     103\,679 &    98\,026 &    85.85\,s &    53.32\,s &    1.6 &  11.40 \\
        Emails EU core &       986 &     327\,228 &   139\,649 &    44.15\,s &    23.93\,s &    1.8 &   2.26 \\
              Epinions &  131\,828 &     841\,376 &        939 &     5.31\,s &     4.63\,s &    1.1 &   1.97 \\
 Madrid interurban bus &    7\,543 &  1\,495\,050 &     1\,440 &     1.44\,s &     1.11\,s &    1.3 &   0.22 \\
            Madrid bus &    4\,597 &  2\,254\,988 &     1\,440 &     1.77\,s &     1.40\,s &    1.3 &   0.21 \\
          Helsinki bus &    7\,959 &  6\,403\,780 &     1\,440 &     3.50\,s &     3.52\,s &    1.0 &   0.41 \\
          Madrid metro &       467 & 37\,565\,706 &     1\,195 &    12.87\,s &    12.87\,s &    1.0 &   1.76 \\
\bottomrule
\end{tabular}
\end{table*}

\section{Conclusions and future work}
\label{sec:conclusion}

We introduced several pattern detection problems that arise in the
context of mining large temporal graphs. In particular, we presented
both complexity results and designed exact algebraic algorithms based
on constrained multilinear sieving for the problems.
As a highlight, our publicly available implementation can scale to large graphs
with up to one billion edges despite the studied problems being \np-hard. 
We presented extensive experimental results that validate our scalability claims.

Note that the application of our framework is not limited to
temporal paths but rather
can be extended to a wide range of pattern detection problems where we search
for information cascades, temporal arborescences, and temporal subgraphs.
Our algebraic approach makes use of $\bigO(nt)$ memory, which limits the
scalability of the algorithm for large values of $t$. A possible direction to
explore would be to design algorithms with space complexity independent of $t$.
Also we would like to study if it is possible to trade time for space.

\section{Acknowledgements}
\label{sec:acknowledgements}

This research was supported by the Academy of Finland
project ``Adaptive and Intelligent Data (AIDA)'' (317085), 
the EC H2020 RIA project ``SoBigData++'' (871042), 
and the Wallenberg AI, Autonomous Systems and Software Program (WASP)
funded by Knut and Alice Wallenberg Foundation.  
We acknowledge the use of computational
resources funded by the project ``Science-IT'' at Aalto University, Finland.

\bibliographystyle{siam}
\bibliography{paper}


\end{document}

%% file: macros.tex
\usepackage[charter]{mathdesign}
\usepackage{fullpage}
\usepackage{amsmath}
\usepackage{verbatim}
\usepackage{ifthen}
\usepackage{hyperref}
\usepackage{xspace}
\usepackage{subcaption}
\usepackage{amsthm}
\usepackage{calrsfs}

\usepackage[linesnumbered, ruled, vlined]{algorithm2e}

\newcommand{\squishlist}{
 \begin{list}{$\bullet$}
  {  \setlength{\itemsep}{0pt}
     \setlength{\parsep}{3pt}
     \setlength{\topsep}{3pt}
     \setlength{\partopsep}{0pt}
     \setlength{\leftmargin}{2em}
     \setlength{\labelwidth}{1.5em}
     \setlength{\labelsep}{0.5em}
} }
\newcommand{\squishlisttight}{
 \begin{list}{$\bullet$}
  { \setlength{\itemsep}{0pt}
    \setlength{\parsep}{0pt}
    \setlength{\topsep}{0pt}
    \setlength{\partopsep}{0pt}
    \setlength{\leftmargin}{2em}
    \setlength{\labelwidth}{1.5em}
    \setlength{\labelsep}{0.5em}
} }

\newcommand{\squishdesc}{
 \begin{list}{}
  {  \setlength{\itemsep}{0pt}
     \setlength{\parsep}{3pt}
     \setlength{\topsep}{3pt}
     \setlength{\partopsep}{0pt}
     \setlength{\leftmargin}{1em}
     \setlength{\labelwidth}{1.5em}
     \setlength{\labelsep}{0.5em}
} }

\newcommand{\squishend}{
  \end{list}
}

\newcommand{\spara}[1]{\smallskip\noindent{\bf #1}}
\newcommand{\mpara}[1]{\medskip\noindent{\bf #1}}

\newtheorem{theorem}{Theorem}[section]

\newtheorem{lemma}[theorem]{Lemma}

\newcommand{\fpr}[1]{\mathopen{}\left(#1\right)}

\newcommand{\dispfunc}[2]{%
  \ensuremath{%
    \ifthenelse{\equal{\noexpand#2}{}}%
	     {#1}%
		      {{#1}\fpr{#2}}}}

\DeclareMathAlphabet{\pazocal}{OMS}{zplm}{m}{n}

\newcommand{\bigO}{\ensuremath{\mathcal{O}}\xspace}
\newcommand{\np}{\ensuremath{\mathbf{NP}}\xspace}

\newcommand{\kpath}{\ensuremath{k\text{\sc{-Path}}}\xspace}
\newcommand{\ktemppath}{\ensuremath{k\text{\sc{-Temp}\-Path}}\xspace}
\newcommand{\temppath}{\ensuremath{\text{\sc Temp\-Path}}\xspace}
\newcommand{\pathmotif}{\ensuremath{\text{\sc Path\-Motif}}\xspace}
\newcommand{\rainbowpath}{\ensuremath{\text{\sc Rain\-bow\-Path}}\xspace}
\newcommand{\sdcolorfulpath}{\ensuremath{(s,d)\text{\sc-Color\-ful\-Path}}\xspace}
\newcommand{\colorfulpath}{\ensuremath{\text{\sc Color\-ful\-Path}}\xspace}
\newcommand{\ectemppath}{\ensuremath{\text{\sc EC-Temp\-Path}}\xspace}
\newcommand{\ecpathmotif}{\ensuremath{\text{\sc EC-Path\-Motif}}\xspace}
\newcommand{\vcpathmotif}{\ensuremath{\text{\sc VC-Path\-Motif}}\xspace}
\newcommand{\vccolorfulpath}{\ensuremath{\text{\sc VC-Color\-ful\-Path}}\xspace}

%% file: tikzstuff.tex
\newcommand{\tikzscale}{{0.7}}

\usepackage{tikz,pgfplots,pgfplotstable}
\usetikzlibrary{positioning,fit,shapes}
\usetikzlibrary{decorations.pathreplacing}
\usetikzlibrary{arrows}

\pgfdeclarelayer{background}
\pgfdeclarelayer{foreground}
\pgfsetlayers{background,main,foreground}

\makeatletter
\tikzset{multicircle/.style  args={#1, #2}{%
 alias=tmp@name, %
  postaction={%
    insert path={
     \pgfextra{%
     \pgfpointdiff{\pgfpointanchor{\pgf@node@name}{center}}%
                  {\pgfpointanchor{\pgf@node@name}{east}}%
     \pgfmathsetmacro\insiderad{\pgf@x}%
        \fill[white] (\pgf@node@name.center)  circle (\insiderad-\pgflinewidth);%
        \draw[#2] (\pgf@node@name.center)  circle (\insiderad-\pgflinewidth);%
        \fill[#2] (\pgf@node@name.center)  -- ++(0:\insiderad-\pgflinewidth) arc (0:#1:\insiderad-\pgflinewidth)--cycle;%
        }}}}}
\makeatother

\definecolor{yafaxiscolor}{rgb}{0.3, 0.3, 0.3}
\definecolor{yafcolor1}{rgb}{0.4, 0.165, 0.553}
\definecolor{yafcolor2}{rgb}{0.949, 0.482, 0.216}
\definecolor{yafcolor3}{rgb}{0.47, 0.549, 0.306}
\definecolor{yafcolor4}{rgb}{0.925, 0.165, 0.224}
\definecolor{yafcolor5}{rgb}{0.141, 0.345, 0.643}
\definecolor{yafcolor6}{rgb}{0.965, 0.933, 0.267}
\definecolor{yafcolor7}{rgb}{0.627, 0.118, 0.165}
\definecolor{yafcolor8}{rgb}{0.878, 0.475, 0.686}
\definecolor{yafcolor9}{rgb}{0.965, 0.733, 0.767}

\newlength{\yafaxispad}
\newlength{\yaftlpad}
\newlength{\yaflabelpad}
\newlength{\yafaxiswidth}
\newlength{\yafticklen}
\makeatletter
\def\pgfplots@drawtickgridlines@INSTALLCLIP@onorientedsurf#1{}
\makeatother

\newcommand{\yafdrawxaxis}[2]{
  \pgfplotstransformcoordinatex{#1}\let\xmincoord=\pgfmathresult 
  \pgfplotstransformcoordinatex{#2}\let\xmaxcoord=\pgfmathresult 
  \pgfsetlinewidth{\yafaxiswidth} 
  \pgfsetcolor{yafaxiscolor}
  \pgfpathmoveto{\pgfpointadd{\pgfpointadd{\pgfplotspointrelaxisxy{0}{0}}{\pgfqpointxy{\xmincoord}{0}}}{\pgfqpoint{-0.5\yafaxiswidth}{\yafaxispad}}}
  \pgfpathlineto{\pgfpointadd{\pgfpointadd{\pgfplotspointrelaxisxy{0}{0}}{\pgfqpointxy{\xmaxcoord}{0}}}{\pgfqpoint{0.5\yafaxiswidth}{\yafaxispad}}}
  \pgfusepath{stroke}

}
\newcommand{\yafdrawyaxis}[2]{
  \pgfplotstransformcoordinatey{#1}\let\ymincoord=\pgfmathresult 
  \pgfplotstransformcoordinatey{#2}\let\ymaxcoord=\pgfmathresult 
  \pgfsetlinewidth{\yafaxiswidth} 
  \pgfsetcolor{yafaxiscolor}
  \pgfpathmoveto{\pgfpointadd{\pgfpointadd{\pgfplotspointrelaxisxy{0}{0}}{\pgfqpointxy{0}{\ymincoord}}}{\pgfqpoint{\yafaxispad}{-0.5\yafaxiswidth}}}
  \pgfpathlineto{\pgfpointadd{\pgfpointadd{\pgfplotspointrelaxisxy{0}{0}}{\pgfqpointxy{0}{\ymaxcoord}}}{\pgfqpoint{\yafaxispad}{0.5\yafaxiswidth}}}
  \pgfusepath{stroke}
}

\pgfplotscreateplotcyclelist{yaf}{%
{yafcolor1,mark options={scale=0.75},mark=o}, 
{yafcolor2,mark options={scale=0.75},mark=square},
{yafcolor3,mark options={scale=0.75},mark=triangle},
{yafcolor4,mark options={scale=0.75},mark=o},
{yafcolor5,mark options={scale=0.75},mark=o},
{yafcolor6,mark options={scale=0.75},mark=o},
{yafcolor7,mark options={scale=0.75},mark=o},
{yafcolor8,mark options={scale=0.75},mark=o}} 

\pgfplotsset{axis y line=left, axis x line=bottom,
  tick align=outside,
  compat = 1.3,
  tickwidth=\yafticklen,
  clip = false,
  every axis title shift = 0pt,
    x axis line style= {-, line width = 0pt, opacity = 0},
    y axis line style= {-, line width = 0pt, opacity = 0},
    x tick style= {line width = \yafaxiswidth, color=yafaxiscolor, yshift = \yafaxispad},
    y tick style= {line width = \yafaxiswidth, color=yafaxiscolor, xshift = \yafaxispad},
    x tick label style = {font=\scriptsize, yshift = \yaftlpad},
    y tick label style = {font=\scriptsize, xshift = \yaftlpad},
    every axis y label/.style = {at = {(ticklabel cs:0.5)}, rotate=90, anchor=center, font=\scriptsize, yshift = -\yaflabelpad},
    every axis x label/.style = {at = {(ticklabel cs:0.5)}, anchor=center, font=\scriptsize, yshift = \yaflabelpad},
    x tick label style = {font=\scriptsize, yshift = 1pt},
    grid = major,
    major grid style  = {dash pattern = on 1pt off 3 pt},
  every axis plot post/.append style= {line width=\yafaxiswidth} ,
  legend cell align = left,
  legend style = {inner sep = 1pt, cells = {font=\scriptsize}},
  legend image code/.code={%
    \draw[mark repeat=2,mark phase=2,#1] 
    plot coordinates { (0cm,0cm) (0.15cm,0cm) (0.3cm,0cm) };%
  } 
}

%% file: tikz-np-hardness-1.tex
\begin{tikzpicture}[scale=\tikzscale,every node/.style={scale=\tikzscale}]]

\input{tikz-defs}

\node[exnode] (u1) at ( 0,   2.7) {$u_1$};
\node[exnode] (u2) at ( 2,   3.1) {$u_2$};
\node[exnode] (u3) at ( 3,   1.5) {$u_3$};
\node[exnode] (u4) at ( 1.3, 0  ) {$u_4$};
\node[exnode] (u5) at (-0.5, 1  ) {$u_5$};

\draw (u1) edge[-, exedge, bend left  = 10] (u2);
\draw (u2) edge[-, exedge, bend left  = 10] (u3);
\draw (u3) edge[-, exedge, bend left  = 10] (u4);
\draw (u4) edge[-, exedge, bend left  = 10] (u5);
\draw (u5) edge[-, exedge, bend left  = 10] (u1);

\draw (u5) edge[-, exedge] (u2);
\draw (u5) edge[-, exedge] (u3);

\node[fill=white] at (0,-1) {\large Graph $G$};

\end{tikzpicture}%

%% file: tikz-np-hardness-2.tex
\begin{tikzpicture}[scale=\tikzscale,every node/.style={scale=\tikzscale}]]

\input{tikz-defs}

\node[exnode] (u1) at ( 0,   2.7) {$u_1$};
\node[exnode] (u2) at ( 2,   3.1) {$u_2$};
\node[exnode] (u3) at ( 3,   1.5) {$u_3$};
\node[exnode] (u4) at ( 1.3, 0  ) {$u_4$};
\node[exnode] (u5) at (-0.5, 1  ) {$u_5$};

\draw (u1) edge[-, exedge, bend left  = 10] node[exlabel, pos = 0.7, inner sep=-3pt] {\scriptsize $\{1,\ldots,k-1\}$} (u2);
\draw (u2) edge[-, exedge, bend left  = 10] node[exlabel, pos = 0.8] {\scriptsize $\{1,\ldots,k-1\}$} (u3);
\draw (u3) edge[-, exedge, bend left  = 10] node[exlabel, pos = 0.3] {\scriptsize $\{1,\ldots,k-1\}$} (u4);
\draw (u4) edge[-, exedge, bend left  = 10] node[exlabel, pos = 0.7, inner sep=-1pt] {\scriptsize $\{1,\ldots,k-1\}$} (u5);
\draw (u5) edge[-, exedge, bend left  = 10] node[exlabel, pos = 0.1] {\scriptsize $\{1,\ldots,k-1\}\quad$} (u1);

\draw (u2) edge[-, exedge] node[exlabel, pos = 0.3,  inner sep=-5pt] {\scriptsize $\{1,\ldots,k-1\}\quad$} (u5);
\draw (u3) edge[-, exedge] node[exlabel, pos = 0.6,  inner sep=-15pt] {\scriptsize $\{1,\ldots,k-1\}\quad$} (u5);

\node[fill=white] at (0,-1) {\large Temporal graph $G^{\tau}$};

\end{tikzpicture}%

%% file: tikz-pathmotif-1.tex
\begin{tikzpicture}[scale=\tikzscale,every node/.style={scale=\tikzscale}]]

\input{tikz-defs}

\node[mnode, fill=red]   (m1) at ( -2, 2   ) {};
\node[mnode, fill=red]   (m2) at ( -2, 1.5 ) {};
\node[mnode, fill=green] (m3) at ( -2, 1   ) {};
\node[mnode, fill=blue]  (m4) at ( -2, 0.5 ) {};

\node[node, fill=red]    (u1) at (   0, 2.7) {};
\node[node, fill=blue]   (u2) at (   2, 3.1) {};
\node[node, fill=yellow] (u3) at (   3, 1.5) {};
\node[node, fill=red]    (u4) at ( 1.3, 0  ) {};
\node[node, fill=green]  (u5) at (-0.5, 1  ) {};

\draw (u1) edge[-, exedge, bend left  = 10] node[exlabel, pos = 0.5, inner sep=-8pt] {\small $3$} (u2);
\draw (u2) edge[-, exedge, bend left  = 10] node[exlabel, pos = 0.5] {\small $1$} (u3);
\draw (u3) edge[-, exedge, bend left  = 10] node[exlabel, pos = 0.5] {\small $3$} (u4);
\draw (u4) edge[-, exedge, bend left  = 10] node[exlabel, pos = 0.5] {\small $1$} (u5);
\draw (u5) edge[-, exedge, bend left  = 10] node[exlabel, pos = 0.5] {\small $2$} (u1);

\draw (u2) edge[-, exedge] node[exlabel, pos = 0.5,  inner sep=-7pt] {\small $4$} (u5);
\draw (u3) edge[-, exedge] node[exlabel, pos = 0.5,  inner sep=-9pt] {\small $5$} (u5);

\node[fill=white] at (0,-1) {\large (a) A motif query and a temporal graph};

\end{tikzpicture}%

%% file: tikz-pathmotif-2.tex
\begin{tikzpicture}[scale=\tikzscale,every node/.style={scale=\tikzscale}]]

\input{tikz-defs}

\node[node, fill=red,   line width = 1.25mm]    (u1) at ( 0,   2.7) {};
\node[node, fill=blue,  line width = 1.25mm]   (u2) at ( 2,   3.1) {};
\node[node, fill=yellow                    ] (u3) at ( 3,   1.5) {};
\node[node, fill=red,   line width = 1.25mm]    (u4) at ( 1.3, 0  ) {};
\node[node, fill=green, line width = 1.5 mm]  (u5) at (-0.5, 1  ) {};

\draw (u1) edge[-, exedge, bend left  = 10, line width=1.25mm] node[exlabel, pos = 0.5, inner sep=-8pt] {\small $3$} (u2);
\draw (u2) edge[-, exedge, bend left  = 10                   ] node[exlabel, pos = 0.5] {\small $1$} (u3);
\draw (u3) edge[-, exedge, bend left  = 10                   ] node[exlabel, pos = 0.5] {\small $3$} (u4);
\draw (u4) edge[-, exedge, bend left  = 10, line width=1.25mm] node[exlabel, pos = 0.5] {\small $1$} (u5);
\draw (u5) edge[-, exedge, bend left  = 10, line width=1.25mm] node[exlabel, pos = 0.5] {\small $2$} (u1);

\draw (u2) edge[-, exedge] node[exlabel, pos = 0.5,  inner sep=-7pt] {\small $4$} (u5);
\draw (u3) edge[-, exedge] node[exlabel, pos = 0.5,  inner sep=-9pt] {\small $5$} (u5);

\node[fill=white] at (1,-1) {(b) A \pathmotif};

\end{tikzpicture}%

%% file: tikz-colorfulpath-1.tex
\begin{tikzpicture}[scale=\tikzscale,every node/.style={scale=\tikzscale}]]

\input{tikz-defs}

\node[mnode, fill=red]     (m1) at ( -2, 2   ) {};
\node[mnode, fill=green]   (m2) at ( -2, 1.5 ) {};
\node[mnode, fill=blue]    (m3) at ( -2, 1   ) {};
\node[mnode, fill=yellow]  (m4) at ( -2, 0.5 ) {};

\node[node, fill=red]    (u1) at (   0, 2.7) {};
\node[node, fill=blue]   (u2) at (   2, 3.1) {};
\node[node, fill=yellow] (u3) at (   3, 1.5) {};
\node[node, fill=red]    (u4) at ( 1.3, 0  ) {};
\node[node, fill=green]  (u5) at (-0.5, 1  ) {};

\draw (u1) edge[-, exedge, bend left  = 10] node[exlabel, pos = 0.5, inner sep=-8pt] {\small $3$} (u2);
\draw (u2) edge[-, exedge, bend left  = 10] node[exlabel, pos = 0.5] {\small $1$} (u3);
\draw (u3) edge[-, exedge, bend left  = 10] node[exlabel, pos = 0.5] {\small $3$} (u4);
\draw (u4) edge[-, exedge, bend left  = 10] node[exlabel, pos = 0.5] {\small $1$} (u5);
\draw (u5) edge[-, exedge, bend left  = 10] node[exlabel, pos = 0.5] {\small $2$} (u1);

\draw (u2) edge[-, exedge] node[exlabel, pos = 0.5,  inner sep=-7pt] {\small $4$} (u5);
\draw (u3) edge[-, exedge] node[exlabel, pos = 0.5,  inner sep=-9pt] {\small $5$} (u5);

\node[fill=white] at (0,-1) {\large (a) A motif query and a temporal graph};

\end{tikzpicture}%

%% file: tikz-colorfulpath-2.tex
\begin{tikzpicture}[scale=\tikzscale,every node/.style={scale=\tikzscale}]]

\input{tikz-defs}

\node[node, fill=red,    line width = 1.25mm]  (u1) at ( 0,   2.7) {};
\node[node, fill=blue,   line width = 1.25mm]  (u2) at ( 2,   3.1) {};
\node[node, fill=yellow, line width = 1.25mm]  (u3) at ( 3,   1.5) {};
\node[node, fill=red                       ]  (u4) at ( 1.3, 0  ) {};
\node[node, fill=green,  line width = 1.5 mm]  (u5) at (-0.5, 1  ) {};

\draw (u1) edge[-, exedge, bend left  = 10, line width=1.25mm] node[exlabel, pos = 0.5, inner sep=-8pt] {\small $3$} (u2);
\draw (u2) edge[-, exedge, bend left  = 10                   ] node[exlabel, pos = 0.5] {\small $1$} (u3);
\draw (u3) edge[-, exedge, bend left  = 10                   ] node[exlabel, pos = 0.5] {\small $3$} (u4);
\draw (u4) edge[-, exedge, bend left  = 10                   ] node[exlabel, pos = 0.5, inner sep=-1pt] {\small $1$} (u5);
\draw (u5) edge[-, exedge, bend left  = 10                   ] node[exlabel, pos = 0.5] {\small $2$} (u1);

\draw (u2) edge[-, exedge, line width = 1.25mm] node[exlabel, pos = 0.5,  inner sep=-7pt] {\small $4$} (u5);
\draw (u3) edge[-, exedge, line width = 1.25mm] node[exlabel, pos = 0.5,  inner sep=-9pt] {\small $5$} (u5);

\node[fill=white] at (1,-1) {(b) A \rainbowpath };

\end{tikzpicture}%

%% file: tikz-rainbowpath-hardness.tex
\begin{tikzpicture}[scale=\tikzscale,every node/.style={scale=\tikzscale}]]

\input{tikz-defs}


\node[label] at    (    0,  7.3) {\large $v^1$};

\node[node1] (u11) at ( 0,  6.5) {};
\node[label] at    ( -0.6,  6.5) {\large $u^1_1$};

\node[node2] (u12) at ( 0,  5.5) {};
\node[label] at    ( -0.6,  5.5) {\large $u^1_2$};

\node[node3] (u13) at ( 0,  4.5) {};
\node[label] at    ( -0.6,  4.5) {\large $u^1_3$};

\node[label] at   ( 0,  3.85) {\large $\vdots$};

\node[node4] (u1n) at ( 0,  3) {};
\node[label] at    ( -0.6,  3) {\large $u^1_n$};

\node[node5] (u1n1) at ( 0,  2) {};
\node[label] at     ( -0.6,  2) {\large $u^1_{n+1}$};

\node[node7] (u1n2) at ( 0,  1) {};
\node[label] at     ( -0.6,  1) {\large $u^1_{n+2}$};


\node[label] at    (    1.5,  7.3) {\large $v^2$};

\node[node1] (u21) at ( 1.5,  6.5) {};
\node[label] at       ( 0.9,  6.5) {\large $u^2_1$};

\node[node2] (u22) at ( 1.5,  5.5) {};
\node[label] at       ( 0.9,  5.5) {\large $u^2_2$};

\node[node3] (u23) at ( 1.5,  4.5) {};
\node[label] at       ( 0.9,  4.5) {\large $u^2_3$};

\node[label] at       ( 1.5,  3.85) {\large $\vdots$};

\node[node4] (u2n) at ( 1.5,  3) {};
\node[label] at       ( 0.9,  3) {\large $u^2_n$};

\node[node5] (u2n1) at ( 1.5, 2) {};
\node[label] at        ( 0.9, 2) {\large $u^2_{n+1}$};

\node[node7] (u2n2) at ( 1.5,  1) {};
\node[label] at        ( 0.9,  1) {\large $u^2_{n+2}$};


\node[label] at    (    3.0,  7.3) {\large $v^3$};

\node[node1] (u31) at ( 3.0,  6.5) {};
\node[label] at       ( 2.4,  6.5) {\large $u^3_1$};

\node[node2] (u32) at ( 3.0,  5.5) {};
\node[label] at       ( 2.4,  5.5) {\large $u^3_2$};

\node[node3] (u33) at ( 3.0,  4.5) {};
\node[label] at       ( 2.4,  4.5) {\large $u^3_3$};

\node[label] at       ( 3.0,  3.85) {\large $\vdots$};

\node[node4] (u3n) at ( 3.0,  3) {};
\node[label] at       ( 2.4,  3) {\large $u^3_n$};

\node[node5] (u3n1) at ( 3.0, 2) {};
\node[label] at        ( 2.4, 2) {\large $u^3_{n+1}$};

\node[node7] (u3n2) at ( 3.0,  1) {};
\node[label] at        ( 2.4,  1) {\large $u^3_{n+2}$};


\node[label] at ( 4.0,  6.5) {\large $\ldots$};
\node[label] at ( 4.0,  5.5) {\large $\ldots$};
\node[label] at ( 4.0,  4.5) {\large $\ldots$};
\node[label] at ( 4.0,  3) {\large $\ldots$};
\node[label] at ( 4.0,  2) {\large $\ldots$};
\node[label] at ( 4.0,  1) {\large $\ldots$};


\node[label] at       (  5.2,  7.3) {\large $v^{k+1}$};

\node[node1] (uk11) at ( 5.0,  6.5) {};
\node[node2] (uk12) at ( 5.0,  5.5) {};
\node[node3] (uk13) at ( 5.0,  4.5) {};
\node[label] at        ( 5.0,  3.85) {\large $\vdots$};
\node[node4] (uk1n) at ( 5.0,  3) {};
\node[node5] (uk1n1) at ( 5.0, 2) {};
\node[node7] (uk1n2) at ( 5.0,  1) {};


\node[label] at       (  6.7,  7.3) {\large $v^{k+2}$};

\node[node1] (uk21) at ( 6.5,  6.5) {};
\node[node2] (uk22) at ( 6.5,  5.5) {};
\node[node3] (uk23) at ( 6.5,  4.5) {};
\node[label] at        ( 6.5,  3.85) {\large $\vdots$};
\node[node4] (uk2n) at ( 6.5,  3) {};
\node[node5] (uk2n1) at ( 6.5, 2) {};
\node[node7] (uk2n2) at ( 6.5,  1) {};


\node[label] at       (  8.2,  7.3) {\large $v^{k+3}$};

\node[node1] (uk31) at ( 8.0,  6.5) {};
\node[label] at        ( 8.7,  6.5) {\large $u^{k+3}_1$};

\node[node2] (uk32) at ( 8.0,  5.5) {};
\node[label] at        ( 8.7,  5.5) {\large $u^{k+3}_2$};

\node[node3] (uk33) at ( 8.0,  4.5) {};
\node[label] at        ( 8.7,  4.5) {\large $u^{k+3}_3$};

\node[label] at        ( 8.0,  3.85) {\large $\vdots$};

\node[node4] (uk3n) at ( 8.0,  3) {};
\node[label] at        ( 8.7,  3) {\large $u^{k+3}_n$};

\node[node5] (uk3n1) at ( 8.0, 2) {};
\node[label] at         ( 8.7, 2) {\large $u^{k+3}_{n+1}$};

\node[node7] (uk3n2) at ( 8.0,  1) {};
\node[label] at         ( 8.7,  1) {\large $u^{k+3}_{n+2}$};


\draw (u1n1) edge[-] (u21);
\draw (u1n1) edge[-] (u22);
\draw (u1n1) edge[-] (u23);
\draw (u1n1) edge[-] (u2n);

\draw (u21) edge[-] (u32);
\draw (u22) edge[-] (u31);
\draw (u22) edge[-] (u33);

\draw (uk11) edge[-] (uk22);
\draw (uk12) edge[-] (uk21);
\draw (uk12) edge[-] (uk23);

\draw (uk21) edge[-] (uk3n2);
\draw (uk22) edge[-] (uk3n2);
\draw (uk23) edge[-] (uk3n2);
\draw (uk2n) edge[-] (uk3n2);

\end{tikzpicture}%

%% file: tikz-polynomial-encoding-1.tex
\begin{tikzpicture}[scale=\tikzscale,every node/.style={scale=\tikzscale}]]

\input{tikz-defs}

\node[exnode] (v2) at (  -1,  1.75) {$v_2$};
\node[exnode] (v3) at (  -1,  0) {$v_3$};
\node[exnode] (v4) at (  -1, -1.75) {$v_4$};
\node[exnode] (v1) at ( 2,    0) {$v_1$};

\node[fill=white] at ( -1.5,    1) {\large $P_{v_2, \ell-1}$};
\node[fill=white] at ( -1.5,    -0.75) {\large $P_{v_3, \ell-1}$};
\node[fill=white] at ( -1.5,   -2.5) {\large $P_{v_4, \ell-1}$};

\node[fill=white, text width=5cm] at (  4.5,    -1.25) {\large
				$P_{v_1,\ell}  =  x_{v_1} y_{v_1v_2,\ell-1} \, P_{v_2,\ell-1} \,+ $\\
        $~~~~~~~~~~ x_{v_1} y_{v_1v_3,\ell-1} \, P_{v_3,\ell-1} \,+ \notag$\\
        $~~~~~~~~~~ x_{v_1} y_{v_1v_4,\ell-1} \, P_{v_4,\ell-1} $};

\draw (v1) edge[-, exedge] (v2);
\draw (v1) edge[-, exedge] (v3);
\draw (v1) edge[-, exedge] (v4);

\node [fill=white] at ( 0.25,   1.75) {$y_{v_1v_2,\ell-1}$};
\node [fill=white] at ( 0.25,   0.3) {$y_{v_1v_3,\ell-1}$};
\node [fill=white] at ( 0.25,   -1.75) {$y_{v_1v_4,\ell-1}$};

\end{tikzpicture}%

%% file: tikz-polynomial-encoding-demo-1.tex
\begin{tikzpicture}[scale=\tikzscale,every node/.style={scale=\tikzscale}]]

\input{tikz-defs}

\node[exnode] (v1) at (  -1,  1) {$v_1$};
\node[exnode] (v2) at ( 2,    0) {$v_2$};
\node[exnode] (v3) at (  -1, -1) {$v_3$};

\node[fill=white] at ( -0.5,    1.75) {\large 
        $P_{v_1, 1} = \mathbf{x_{v_1}}$};
\node[fill=white] at ( -0.5,   -1.75) {\large 
        $P_{v_3, 1} = \mathbf{x_{v_3}}$};

\node[fill=white] at (  2.5,    -1) {\large
				$P_{v_2, 1} = \mathbf{x_{v_2}}$};

\draw (v1) edge[-, exedge] (v2);
\draw (v3) edge[-, exedge] (v2);


\node[fill=white] at (0, -3.5) {\Large (a) $\ell = 1$.};
\end{tikzpicture}%

%% file: tikz-polynomial-encoding-demo-2.tex
\begin{tikzpicture}[scale=\tikzscale,every node/.style={scale=\tikzscale}]]

\input{tikz-defs}

\node[exnode] (v1) at (  -1,  1) {$v_1$};
\node[exnode] (v2) at ( 2,    0) {$v_2$};
\node[exnode] (v3) at (  -1, -1) {$v_3$};

\node[fill=white] at ( -0.25,    1.75) {\large 
        $P_{v_1, 2} = \mathbf{x_{v_1} y_{v_1v_2,1} x_{v_2}}$};
\node[fill=white] at ( -0.25,   -1.75) {\large 
        $P_{v_3, 2} = \mathbf{x_{v_3} y_{v_3v_2,1} x_{v_2}}$};

\node[fill=white, text width=4cm] at (  3.5,    -1.25) {\large
				$P_{v_2,2}  =  \mathbf{x_{v_2} y_{v_2v_1,1} x_{v_1}}\,+ $\\
        $~~~~~~~~~~ \mathbf{x_{v_2} y_{v_2v_3,1} x_{v_3}}$};

\draw (v1) edge[-, exedge] (v2);
\draw (v3) edge[-, exedge] (v2);


\node[fill=white] at (0, -3.5) {\Large (b) $\ell = 2$.};
\end{tikzpicture}%

%% file: tikz-polynomial-encoding-demo-3.tex
\begin{tikzpicture}[scale=\tikzscale,every node/.style={scale=\tikzscale}]]

\input{tikz-defs}

\node[exnode] (v1) at (  -1,  1) {$v_1$};
\node[exnode] (v2) at ( 2,    0) {$v_2$};
\node[exnode] (v3) at (  -1, -1) {$v_3$};

\node[fill=white, text width=5.5cm] at (  -0.5,    2) {\large
				$P_{v_1,3}  =  x_{v_1} y_{v_1v_2,2} x_{v_2} y_{v_2v_1,1} x_{v_1}\,+ $\\
        $~~~~~~~~~~ \mathbf{x_{v_1} y_{v_1v_2,2} x_{v_2} y_{v_2v_3,1} x_{v_3}}$};

\node[fill=white, text width=5.5cm] at (  4,    -1) {\large
        $P_{v_2, 3} = x_{v_2} y_{v_2v_1,2} x_{v_1} y_{v_1v_2,1} x_{v_2}\,+ $\\
        $~~~~~~~~~~ x_{v_2} y_{v_2v_3,2} x_{v_3} y_{v_3v_2,1} x_{v_2}$};

\node[fill=white, text width=5.5cm] at ( -0.5,  -2) {\large 
        $P_{v_3, 3} = \mathbf{x_{v_3} y_{v_3v_2,2} x_{v_2} y_{v_2v_1,1} x_{v_1}}\,+ $\\
        $~~~~~~~~~~ x_{v_3} y_{v_3v_2,2} x_{v_2} y_{v_2v_3,1} x_{v_3}$};

\draw (v1) edge[-, exedge] (v2);
\draw (v3) edge[-, exedge] (v2);


\node[fill=white] at (0, -3.5) {\Large (c) $\ell = 3$.};
\end{tikzpicture}%

%% file: tikz-polynomial-encoding-2.tex
\begin{tikzpicture}[scale=\tikzscale,every node/.style={scale=\tikzscale}]]

\input{tikz-defs}

\node[exnode] (v2) at (  -1,  1.75) {$v_2$};
\node[exnode] (v3) at (  -1,  0) {$v_3$};
\node[exnode] (v4) at (  -1, -1.75) {$v_4$};
\node[exnode] (v1) at ( 2,    0) {$v_1$};

\node[fill=white] at ( -1.5,    1) {\large $P_{v_2, \ell-1, i-1}$};
\node[fill=white] at ( -1.5,    -0.75) {\large $P_{v_3, \ell-1, i-1}$};
\node[fill=white] at ( -1.5,   -2.5) {\large $P_{v_4, \ell-1, i-1}$};
\node[fill=white, text width=7cm] at (  5.5,    -1.25) {\large 
				$P_{v_1,\ell,i}  =  x_{v_1} y_{v_1v_2,\ell-1,i} \, P_{v_2,\ell-1,i-1} \,+ $\\
                $~~~~~~~~~~~~~ x_{v_1} y_{v_1v_3,\ell-1,i} \, P_{v_3,\ell-1,i-1} \,+ \notag$\\
                $~~~~~~~~~~~~~ x_{v_1} y_{v_1v_4,\ell-1,i} \, P_{v_4,\ell-1,i-1} + P_{v_1,\ell,i-1}$};

\draw (v1) edge[-, exedge] (v2);
\draw (v1) edge[-, exedge] (v3);
\draw (v1) edge[-, exedge] (v4);

\node [fill=white] at ( 0.25,   1.75) {$y_{v_1v_2,\ell-1,i}$};
\node [fill=white] at ( 0.25,   0.3) {$y_{v_1v_3,\ell-1,i}$};
\node [fill=white] at ( 0.25,   -1.75) {$y_{v_1v_4,\ell-1,i}$};

\end{tikzpicture}%

%% file: tikz-polynomial-encoding-temp-11.tex
\begin{tikzpicture}[scale=\tikzscale,every node/.style={scale=\tikzscale}]]

\input{tikz-defs}

\node[exnode] (v1) at (  -1,  1) {$v_1$};
\node[exnode] (v2) at ( 2,    0) {$v_2$};
\node[exnode] (v3) at (  -1, -1) {$v_3$};

\node[fill=white] at ( -1.5,    1.75) {\large 
        $P_{v_1, 1, 1} = \mathbf{x_{v_1}}$};
\node[fill=white, text width=5cm] at (  4,    -1) {\large
				$P_{v_2, 1, 1} = \mathbf{x_{v_2}}$};
\node[fill=white] at ( -1.5,   -1.75) {\large 
        $P_{v_3, 1, 1} = \mathbf{x_{v_3}}$};

\draw (v1) edge[-, exedge] (v2);
\draw (v3) edge[-, exedge] (v2);

\node [fill=white] at ( 0.5,   1) {$\{1,2\}$};
\node [fill=white] at ( 0.5,  -1) {$\{1,2\}$};

\node[fill=white] at (0, -3) {\Large (a) $\ell = 1, i = 1$.};
\end{tikzpicture}%

%% file: tikz-polynomial-encoding-temp-12.tex
\begin{tikzpicture}[scale=\tikzscale,every node/.style={scale=\tikzscale}]]

\input{tikz-defs}

\node[exnode] (v1) at (  -1,  1) {$v_1$};
\node[exnode] (v2) at ( 2,    0) {$v_2$};
\node[exnode] (v3) at (  -1, -1) {$v_3$};

\node[fill=white] at ( -1.5,    1.75) {\large 
        $P_{v_1, 1, 2} = \mathbf{x_{v_1}}$};
\node[fill=white, text width=4cm] at (  4,    -1) {\large
				$P_{v_2, 1, 2} = \mathbf{x_{v_2}}$};
\node[fill=white] at ( -1.5,   -1.75) {\large 
        $P_{v_3, 1, 2} = \mathbf{x_{v_3}}$};

\draw (v1) edge[-, exedge] (v2);
\draw (v3) edge[-, exedge] (v2);

\node [fill=white] at ( 0.5,   1) {$\{1,2\}$};
\node [fill=white] at ( 0.5,  -1) {$\{1,2\}$};

\node[fill=white] at (0, -3) {\Large (b) $\ell = 1, i = 2$.};
\end{tikzpicture}%

%% file: tikz-polynomial-encoding-temp-21.tex
\begin{tikzpicture}[scale=\tikzscale,every node/.style={scale=\tikzscale}]]

\input{tikz-defs}

\node[exnode] (v1) at (  -1,  1) {$v_1$};
\node[exnode] (v2) at ( 2,    0) {$v_2$};
\node[exnode] (v3) at (  -1, -1) {$v_3$};

\node[fill=white] at ( -1.5,    2) {\large 
        $P_{v_1, 2, 1} = \mathbf{x_{v_1} y_{v_1v_1,1,1} x_{v_2}}$};
\node[fill=white, text width=5.5cm] at ( 4.5,   -1.25) {\large 
        $P_{v_2, 2, 1} = \mathbf{x_{v_2} y_{v_2v_1,1,1} x_{v_1}}\,+ $\\
        $~~~~~~~~~~~~ \mathbf{x_{v_2} y_{v_2v_3,1,1} x_{v_3}}$};
\node[fill=white] at ( -1.5,   -2) {\large 
				$P_{v_2, 2, 1} = \mathbf{x_{v_3} y_{v_3v_2,1,1} x_{v_2}}$};

\draw (v1) edge[-, exedge] (v2);
\draw (v3) edge[-, exedge] (v2);

\node [fill=white] at ( 0.5,   1) {$\{1,2\}$};
\node [fill=white] at ( 0.5,  -1) {$\{1,2\}$};

\node[fill=white] at (0, -3.5) {\Large (c) $\ell = 2, i = 1$.};
\end{tikzpicture}%

%% file: tikz-polynomial-encoding-temp-22.tex
\begin{tikzpicture}[scale=\tikzscale,every node/.style={scale=\tikzscale}]]

\input{tikz-defs}

\node[exnode] (v1) at (  -1,  1) {$v_1$};
\node[exnode] (v2) at ( 2,    0) {$v_2$};
\node[exnode] (v3) at (  -1, -1) {$v_3$};

\node[fill=white, text width=5.5cm] at ( -1.5,    2) {\large 
        $P_{v_1, 2, 2} = \mathbf{x_{v_1} y_{v_1v_2,1,1} x_{v_2}}\,+ $\\
        $~~~~~~~~~~~~ \mathbf{x_{v_1} y_{v_1v_2,1,2} x_{v_2}}$};
\node[fill=white, text width=5.5cm] at ( 4.5,   -1.75) {\large 
        $P_{v_2, 2, 2} = \mathbf{x_{v_2} y_{v_2v_1,1,1} x_{v_1}}\, + $\\
        $~~~~~~~~~~~~ \mathbf{x_{v_2} y_{v_2v_3,1,1} x_{v_3}}\, +$\\
        $~~~~~~~~~~~~ \mathbf{x_{v_2} y_{v_2v_1,1,2} x_{v_1}}\, +$\\
        $~~~~~~~~~~~~ \mathbf{x_{v_2} y_{v_2v_3,1,2} x_{v_3}}$};
\node[fill=white, text width=5.5cm] at ( -1.5,   -2) {\large 
				$P_{v_3, 2, 2} = \mathbf{x_{v_3} y_{v_3v_2,1,1} x_{v_2}}\, +$\\
        $~~~~~~~~~~~~ \mathbf{x_{v_3} y_{v_3v_2,1,2} x_{v_2}}$};

\draw (v1) edge[-, exedge] (v2);
\draw (v3) edge[-, exedge] (v2);

\node [fill=white] at ( 0.5,   1) {$\{1,2\}$};
\node [fill=white] at ( 0.5,  -1) {$\{1,2\}$};

\node[fill=white] at (0, -3.5) {\Large (d) $\ell = 2, i = 2$.};
\end{tikzpicture}%

%% file: tikz-polynomial-encoding-temp-31.tex
\begin{tikzpicture}[scale=\tikzscale,every node/.style={scale=\tikzscale}]]

\input{tikz-defs}

\node[exnode] (v1) at (  -1,  1) {$v_1$};
\node[exnode] (v2) at ( 2,    0) {$v_2$};
\node[exnode] (v3) at (  -1, -1) {$v_3$};

\node[fill=white] at ( -1.5,    1.75) {\large 
        $P_{v_1, 3, 1} = \emptyset$};
\node[fill=white, text width=5cm] at (  4,    -1) {\large
				$P_{v_2, 3, 1} = \emptyset$};
\node[fill=white] at ( -1.5,   -1.75) {\large 
        $P_{v_3, 3, 1} = \emptyset$};

\draw (v1) edge[-, exedge] (v2);
\draw (v3) edge[-, exedge] (v2);

\node [fill=white] at ( 0.5,   1) {$\{1,2\}$};
\node [fill=white] at ( 0.5,  -1) {$\{1,2\}$};

\node[fill=white] at (0, -3) {\Large (e) $\ell = 3, i = 1$.};
\end{tikzpicture}%

%% file: tikz-polynomial-encoding-temp-32.tex
\begin{tikzpicture}[scale=\tikzscale,every node/.style={scale=\tikzscale}]]

\input{tikz-defs}

\node[exnode] (v1) at (  -1,  1) {$v_1$};
\node[exnode] (v2) at ( 2,    0) {$v_2$};
\node[exnode] (v3) at (  -1, -1) {$v_3$};

\node[fill=white, text width=6cm] at ( -1.5,    2) {\large 
        $P_{v_1, 3, 2} = x_{v_1} y_{v_1v_2,2,2} x_{v_2} y_{v_2v_1,1,1} x_{v_1}\,+ $\\
        $~~~~~~~~~~~~ \mathbf{x_{v_1} y_{v_1v_2,2,2} x_{v_2} y_{v_2v_3,1,1} x_{v_3}}$};
\node[fill=white, text width=6cm] at ( 4.5,   -1) {\large 
        $P_{v_2, 3, 2} = x_{v_2} y_{v_2v_1,2,2} x_{v_1} y_{v_1v_2,1,1} x_{v_2}\, + $\\
        $~~~~~~~~~~~~ x_{v_2} y_{v_2v_3,2,2} x_{v_3} y_{v_3,v_2,1,1} x_{v_2}$};
\node[fill=white, text width=6cm] at ( -1.5,   -2) {\large 
				$P_{v_3, 3, 2} = x_{v_3} y_{v_3v_2,2,2} x_{v_2} y_{v_2v_1,1,1} x_{v_1}\, +$\\
        $~~~~~~~~~~~~ \mathbf{x_{v_3} y_{v_3v_2,2,2} x_{v_2} y_{v_2v_3,1,1} x_{v_3}}$};

\draw (v1) edge[-, exedge] (v2);
\draw (v3) edge[-, exedge] (v2);

\node [fill=white] at ( 0.5,   1) {$\{1,2\}$};
\node [fill=white] at ( 0.5,  -1) {$\{1,2\}$};

\node[fill=white] at (0, -3.5) {\Large (f) $\ell = 3, i = 2$.};
\end{tikzpicture}%

%% file: paper.bbl
\begin{thebibliography}{10}

\bibitem{AlonPIFC2008}
{\sc N.~Alon, P.~Dao, I.~Hajirasouliha, F.~Hormozdiari, and S.~C. Sahinalp},
  {\em Biomolecular network motif counting and discovery by color coding},
  Bioinformatics, 24 (2008), pp.~241--249.

\bibitem{AslayMFGG2018}
{\sc C.~Aslay, A.~Nasir, G.~De~Francisci~Morales, and A.~Gionis}, {\em Mining
  frequent patterns in evolving graphs}, in CIKM, 2018, pp.~923--932.

\bibitem{BellG2008}
{\sc N.~Bell and M.~Garland}, {\em Efficient sparse matrix-vector
  multiplication on {CUDA}}, {NVIDIA} tech. rep., {NVIDIA} Corp., 2008.

\bibitem{BensonGL2016}
{\sc A.~Benson, D.~Gleich, and J.~Leskovec}, {\em Higher-order organization of
  complex networks}, Science, 353 (2016), pp.~163--166.

\bibitem{BjorklundHKK2017}
{\sc A.~Bj{\"{o}}rklund, T.~Husfeldt, P.~Kaski, and M.~Koivisto}, {\em Narrow
  sieves for parameterized paths and packings}, JCSS, 87 (2017), pp.~119--139.

\bibitem{Bjorklund2014}
{\sc A.~Bj{\"{o}}rklund, P.~Kaski, and {\L}.~Kowalik}, {\em Determinant sums
  for undirected {H}amiltonicity}, {SIAM} J. Comput., 43 (2014), pp.~280--299.

\bibitem{BjorklundKK-esa2014}
\leavevmode\vrule height 2pt depth -1.6pt width 23pt, {\em Fast witness
  extraction using a decision oracle}, in ESA, 2014, pp.~149--160.

\bibitem{BjorklundKK2016}
\leavevmode\vrule height 2pt depth -1.6pt width 23pt, {\em Constrained
  multilinear detection and generalized graph motifs}, Algorithmica, 74 (2016),
  pp.~947--967.

\bibitem{BjorklundKKL2015}
{\sc A.~Bj{\"{o}}rklund, P.~Kaski, {\L}.~Kowalik, and J.~Lauri}, {\em
  Engineering motif search for large graphs}, in ALENEX, 2015, pp.~104--118.

\bibitem{Bollobas01}
{\sc B.~Bollob{\'a}s}, {\em Random Graphs}, Cambridge UP, second~ed., 2001.

\bibitem{BressanLP2019}
{\sc M.~Bressan, S.~Leucci, and A.~Panconesi}, {\em Motivo: Fast motif counting
  via succinct color coding and adaptive sampling}, {PVLDB}, 12 (2019),
  pp.~1651--1663.

\bibitem{CasteigtsHMZ2019}
{\sc A.~Casteigts, A.~Himmel, H.~Molter, and P.~Zschoche}, {\em The
  computational complexity of finding temporal paths under waiting time
  constraints}, CoRR, abs/1909.06437 (2019).

\bibitem{Chen2011}
{\sc L.~Chen, X.~Li, and Y.~Shi}, {\em The complexity of determining the
  rainbow vertex-connection of a graph}, Theoretical Computer Science, 412
  (2011), pp.~4531--4535.

\bibitem{cicaleseGGLLRT13}
{\sc F.~Cicalese, T.~Gagie, E.~Giaquinta, E.~S. Laber, Z.~Lipt{\'{a}}k,
  R.~Rizzi, and A.~I. Tomescu}, {\em Indexes for jumbled pattern matching in
  strings, trees and graphs}, in SPIRE, 2013, pp.~56--63.

\bibitem{ColettoGGL2017}
{\sc M.~Coletto, K.~Garimella, A.~Gionis, and C.~Lucchese}, {\em Automatic
  controversy detection in social media: {A} content-independent motif-based
  approach}, Online Social Networks and Media, 3-4 (2017), pp.~22--31.

\bibitem{colettoKGL2017}
{\sc M.~Coletto, K.~Garimella, A.~Gionis, and C.~Lucchese}, {\em A motif-based
  approach for identifying controversy}, in Eleventh International AAAI
  Conference on Web and Social Media, 2017.

\bibitem{CyganFKLMPPS2015}
{\sc M.~{Cygan}, F.~V. {Fomin}, {\L}.~{Kowalik}, D.~{Lokshtanov}, D.~{Marx},
  M.~{Pilipczuk}, M.~{Pilipczuk}, and S.~{Saurabh}}, {\em {Parameterized
  algorithms}}, 2015.

\bibitem{DeFAGLY2010}
{\sc M.~DeChoudhury, M.~Feldman, S.~Amer-Yahia, N.~Golbandi, R.~Lempel, and
  C.~Yu}, {\em Automatic construction of travel itineraries using social
  breadcrumbs}, in HT, 2010, pp.~35--44.

\bibitem{Dechter1991}
{\sc R.~Dechter, I.~Meiri, and J.~Pearl}, {\em Temporal constraint networks},
  Artificial intelligence, 49 (1991), pp.~61--95.

\bibitem{DellLM2020}
{\sc H.~Dell, J.~Lapinskas, and K.~Meeks}, {\em Approximately counting and
  sampling small witnesses using a colourful decision oracle}, in SODA, 2020,
  pp.~2201--2211.

\bibitem{Fomin2016efficient}
{\sc F.~V. Fomin, D.~Lokshtanov, F.~Panolan, and S.~Saurabh}, {\em Efficient
  computation of representative families with applications in parameterized and
  exact algorithms}, J. ACM, 63 (2016).

\bibitem{GagieHLW13}
{\sc T.~Gagie, D.~Hermelin, G.~M. Landau, and O.~Weimann}, {\em Binary jumbled
  pattern matching on trees and tree-like structures}, in ESA, 2013.

\bibitem{garey2002computers}
{\sc M.~R. Garey and D.~S. Johnson}, {\em Computers and intractability},
  vol.~29, W. H. Freeman and Co., 2002.

\bibitem{GeorgeKS2007}
{\sc B.~George, S.~Kim, and S.~Shekhar}, {\em Spatio-temporal network databases
  and routing algorithms: A summary of results}, in International Symposium on
  Spatial and Temporal Databases, 2007, pp.~460--477.

\bibitem{GiaquintaG13}
{\sc E.~Giaquinta and S.~Grabowski}, {\em New algorithms for binary jumbled
  pattern matching}, IPL, 113 (2013), pp.~538--542.

\bibitem{GionisLPT2014}
{\sc A.~Gionis, T.~Lappas, K.~Pelechrinis, and E.~Terzi}, {\em Customized tour
  recommendations in urban areas}, WSDM, 2014, pp.~313--322.

\bibitem{GuptaAH2011}
{\sc M.~Gupta, C.~C. Aggarwal, and J.~Han}, {\em Finding top-k shortest path
  distance changes in an evolutionary network}, in SSTD, 2011, pp.~130--148.

\bibitem{Holme2015}
{\sc P.~Holme}, {\em Modern temporal network theory: a colloquium}, European
  Physical Journal B, 88 (2015), p.~234.

\bibitem{holme2012temporal}
{\sc P.~Holme and J.~Saram{\"a}ki}, {\em Temporal networks}, Physics reports,
  519 (2012), pp.~97--125.

\bibitem{Holmes2012}
\leavevmode\vrule height 2pt depth -1.6pt width 23pt, {\em Temporal networks},
  Physics reports, 519 (2012), pp.~97--125.

\bibitem{honey2007network}
{\sc C.~J. Honey, R.~K{\"o}tter, M.~Breakspear, and O.~Sporns}, {\em Network
  structure of cerebral cortex shapes functional connectivity on multiple time
  scales}, PNAS, 104 (2007), pp.~10240--10245.

\bibitem{kaskiLT2018}
{\sc P.~Kaski, J.~Lauri, and S.~Thejaswi}, {\em {Engineering Motif Search for
  Large Motifs}}, in SEA, 2018, pp.~1--19.

\bibitem{Kostakos2009}
{\sc V.~Kostakos}, {\em Temporal graphs}, Physica A: Statistical Mechanics and
  its Applications, 388 (2009), pp.~1007--1023.

\bibitem{koutis-icalp08}
{\sc I.~Koutis}, {\em Faster algebraic algorithms for path and packing
  problems}, in ICALP, 2008.

\bibitem{koutis-dagstuhl}
\leavevmode\vrule height 2pt depth -1.6pt width 23pt, {\em The power of group
  algebras for constrained multilinear monomial detection}, Dagstuhl meeting
  10441,  (2010).

\bibitem{koutis-ipl}
\leavevmode\vrule height 2pt depth -1.6pt width 23pt, {\em Constrained
  multilinear detection for faster functional motif discovery}, IPL, 112
  (2012), pp.~889--892.

\bibitem{koutis-williams-icalp09}
{\sc I.~Koutis and R.~Williams}, {\em Limits and applications of group algebras
  for parameterized problems}, in ICALP (1), 2009.

\bibitem{koutisW2016}
{\sc I.~Koutis and R.~Williams}, {\em Algebraic fingerprints for faster
  algorithms}, Comm. of the ACM, 59 (2016), pp.~98--105.

\bibitem{KovanenKKKS2011}
{\sc L.~Kovanen, M.~Karsai, K.~Kaski, J.~Kert{\'e}sz, and J.~Saram{\"a}ki},
  {\em Temporal motifs in time-dependent networks}, Journal of Statistical
  Mechanics: Theory and Experiment, 2011 (2011), p.~P11005.

\bibitem{KowalikL2016}
{\sc {\L}.~Kowalik and J.~Lauri}, {\em On finding rainbow and colorful paths},
  TCS, 628 (2016), pp.~110 -- 114.

\bibitem{koblenz}
{\sc J.~Kunegis}, {\em {KONECT:} the {K}oblenz network collection}, in WWW,
  2013, pp.~1343--1350.
\newblock \url{http://konect.uni-koblenz.de/networks/}.

\bibitem{lacroix2006motif}
{\sc V.~Lacroix, C.~G. Fernandes, and M.-F. Sagot}, {\em Motif search in
  graphs: application to metabolic networks}, IEEE Transactions on
  Computational Biology and Bioinformatics (TCBB), 3 (2006), pp.~360--368.

\bibitem{latapy2018stream}
{\sc M.~Latapy, T.~Viard, and C.~Magnien}, {\em Stream graphs and link streams
  for the modeling of interactions over time}, Social Network Analysis and
  Mining, 8 (2018).

\bibitem{snapnets}
{\sc J.~Leskovec and A.~Krevl}, {\em {SNAP Datasets}: {Stanford} large network
  dataset collection}.
\newblock \url{http://snap.stanford.edu/data}, June 2014.

\bibitem{LinAHC2016}
{\sc S.-J. Lin, T.~Y. Al-Naffouri, Y.~S. Han, and W.-H. Chung}, {\em Novel
  polynomial basis with fast {F}ourier transform and its application to
  {R}eed-{S}olomon erasure codes}, ITIT, 62 (2016).

\bibitem{liu2010mining}
{\sc L.~Liu, J.~Tang, J.~Han, M.~Jiang, and S.~Yang}, {\em Mining topic-level
  influence in heterogeneous networks}, in CIKM, 2010, pp.~199--208.

\bibitem{LiuBC2019}
{\sc P.~Liu, A.~Benson, and M.~Charikar}, {\em Sampling methods for counting
  temporal motifs}, in WSDM, 2019, pp.~294--302.

\bibitem{MiloSKC2002}
{\sc R.~Milo, S.~Shen-Orr, S.~Itzkovitz, N.~Kashtan, D.~Chklovskii, and
  U.~Alon}, {\em Network motifs: Simple building blocks of complex networks},
  Science, 298 (2002), pp.~824--827.

\bibitem{ParanjapeBL2017}
{\sc A.~Paranjape, A.~Benson, and J.~Leskovec}, {\em Motifs in temporal
  networks}, WSDM, 2017, pp.~601--610.

\bibitem{schwartz}
{\sc J.~T. Schwartz}, {\em Fast probabilistic algorithms for verification of
  polynomial identities}, J. ACM, 27 (1980), pp.~701--717.

\bibitem{conf-code}
{\sc S.~Thejaswi and A.~Gionis}, 2019.
\newblock \url{https://github.com/suhastheju/temporal-patterns}.

\bibitem{conf-paper}
{\sc S.~Thejaswi and A.~Gionis}, {\em Pattern detection in large temporal
  graphs using algebraic fingerprints}, in SDM, 2020, pp.~1--10.

\bibitem{journal-code}
{\sc S.~Thejaswi, A.~Gionis, and J.~Lauri}, 2020.
\newblock \url{https://github.com/suhastheju/temporal-patterns-mk2}.

\bibitem{Uchizawa2013}
{\sc K.~Uchizawa, T.~Aoki, T.~Ito, A.~Suzuki, and X.~Zhou}, {\em On the rainbow
  connectivity of graphs: complexity and {FPT} algorithms}, Algorithmica, 67
  (2013), pp.~161--179.

\bibitem{VansteenwegenSO2011}
{\sc P.~Vansteenwegen, W.~Souffriau, and D.~V. Oudheusden}, {\em The
  orienteering problem: A survey}, EJOR, 209 (2011), pp.~1 -- 10.

\bibitem{wackersreuther2010frequent}
{\sc B.~Wackersreuther, P.~Wackersreuther, A.~Oswald, C.~B{\"o}hm, and
  K.~Borgwardt}, {\em Frequent subgraph discovery in dynamic networks}, in MLG,
  2010.

\bibitem{williams-ipl}
{\sc R.~Williams}, {\em Finding paths of length $k$ in {$O^*(2^k)$} time}, IPL,
  109 (2009).

\bibitem{WuCHKLX2014}
{\sc H.~Wu, J.~Cheng, S.~Huang, Y.~Ke, Y.~Lu, and Y.~Xu}, {\em Path problems in
  temporal graphs}, Proc. VLDB Endow., 7 (2014), pp.~721--732.

\bibitem{WuC2016}
{\sc H.~{Wu}, J.~{Cheng}, Y.~{Ke}, S.~{Huang}, Y.~{Huang}, and H.~{Wu}}, {\em
  Efficient algorithms for temporal path computation}, TKDE, 28 (2016),
  pp.~2927--2942.

\bibitem{yang2013community}
{\sc J.~Yang, J.~McAuley, and J.~Leskovec}, {\em Community detection in
  networks with node attributes}, in ICDM, 2013, pp.~1151--1156.

\bibitem{Zippel1979}
{\sc R.~Zippel}, {\em Probabilistic algorithms for sparse polynomials}, in
  Proc. International Symposium on Symbolic and Algebraic Computation, vol.~72
  of LNCS, 1979, pp.~216--226.

\end{thebibliography}
